\documentclass[11pt]{article}%
\usepackage{amssymb}
\usepackage{amsfonts}
\usepackage{sw20elba}
\usepackage{hyperref}
\usepackage{xcolor}
\usepackage{float}
\usepackage[longnamesfirst]{natbib}
\usepackage{amsmath}
\usepackage{comment}
\usepackage{graphicx}%
\providecommand{\U}[1]{\protect\rule{.1in}{.1in}}
\newtheorem{theorem}{Theorem}

\newtheorem{proposition}[theorem]{Proposition}
\newtheorem{remark}{Remark}

\newenvironment{proof}[1][Proof]{\noindent \textbf{#1.} }{\  \rule{0.5em}{0.5em}}
\begin{document}
	
	\title{Random Utility Models with Skewed Random Components: \\the Smallest versus Largest Extreme Value Distribution}
	\author{Richard T. Carson\thanks{Department of Economics, UC\ San\ Diego. Email:
			rcarson@ucsd.edu}, Derrick H. Sun\thanks{Department of Mathematics and College
			of Computing, Data Science, and Society, UC Berkeley. Email:
			dhsun@berkeley.edu}, and Yixiao Sun\thanks{Department of Economics,
			UC\ San\ Diego. Email: yisun@ucsd.edu}}
	\date{\today}
	\maketitle
	
	\begin{abstract}
		At the core of most random utility models (RUMs) involving multinomial choice
		behavior is an individual agent with a random utility component following a
		\emph{largest} extreme value Type I (LEVI) distribution. What if, instead, the
		random component follows its mirror image --- the\emph{ smallest} extreme
		value Type I (SEVI) distribution? Differences between these specifications,
		closely tied to the random component's skewness, can be quite profound. For
		the same preference parameters, the two RUMs, equivalent with only two choice
		alternatives, diverge progressively as the number of alternatives increases,
		resulting in substantially different estimates and predictions for key
		measures, such as elasticities and market shares.
		
		The LEVI model imposes the well-known independence-of-irrelevant-alternatives
		property, while SEVI does not. Instead, the SEVI choice probability for a
		particular option involves enumerating all subsets that contain this option.
		The SEVI model, though more complex to estimate, is shown to have
		computationally tractable closed-form choice probabilities. Much of the paper
		delves into explicating the properties of the SEVI model and exploring
		implications of the random component's skewness, including offering new
		insights into multinomial probit models. SEVI-based counterparts exist for
		most LEVI-based generalizations of the conditional logit model such as mixed
		logit, and the LEVI versus SEVI issue is largely orthogonal to the issues
		these models are intended to address.
		
		Conceptually, the difference between the LEVI and SEVI models centers on
		whether information, known only to the agent, is more likely to increase or
		decrease the systematic utility parameterized using observed attributes. LEVI
		does the former; SEVI the latter. An immediate implication is that if choice
		is characterized by SEVI random components, then the observed choice is more
		likely to correspond to the systematic-utility-maximizing choice than if
		characterized by LEVI. Examining standard empirical examples from different
		applied areas, we find that the SEVI model outperforms the LEVI model,
		suggesting the relevance of its inclusion in applied researchers' toolkits.
		
		\bigskip
		

	\bigskip
	
	\end{abstract}
		
	\thispagestyle{empty}\setcounter{page}{0}
	
	\bigskip\bigskip\bigskip\bigskip\bigskip
	
	\pagebreak

\section{Introduction}

Random utility models (RUM) are a cornerstone of applied microeconomics and
other fields such as marketing, health policy, and transportation research.
Since the seminal work of
\citet{McFadden1973}%
, the vast majority of RUMs have been formulated as conditional or multinomial
logit models, or their generalizations (%
\citet{Hensher_Rose_Greene_2015}
and
\citet{book_train_2009}%
).\footnote{Conditional logit typically refers to choices between alternatives
with different observable attributes. In contrast, multinomial logit focuses
on choices where the decision-makers themselves differ in characteristics.
These two models can be combined by interacting alternative attributes with
agent characteristics. Here, we primarily use the conditional logit
representation for simplicity.} These models posit that the total utility of
an alternative comprises two components: a systematic component parameterized
by observable attributes, and an idiosyncratic component known only to the
agent (%
\citet{Manski1977}%
). The latter is often referred to as the random component, since it is
unknown to the econometrician, who generally assumes it follows a specific
distribution. The agent chooses the alternative with the highest total
utility, making choice behavior appear random from the econometrician's
perspective, even though it may be regarded as deterministic from the agent's viewpoint.

The conditional logit model and its standard generalizations assume that the
random component follows the Largest Extreme Value Type I (LEVI) distribution,
also known as the Gumbel or double exponential distribution (%
\citet{Gumbel1958}%
). This distribution is continuous, has a single mode, is asymmetric and
skewed with a long \emph{right} tail. Mirroring the LEVI distribution is the
Smallest Extreme Value Type I (SEVI) distribution, also known as the Reverse
Gumbel. It is also asymmetric but with a long \emph{left} tail. This paper
presents a systematic study of RUMs based on the SEVI distribution,
highlighting their differences from conventional logit models and practical implications.

The natural question at this point for a choice modeler is: does the choice
between LEVI and SEVI random components matter empirically? Two important bits
of information suggest it should not. First, the difference between two
(standard) LEVI random variables or two (standard) SEVI random variables is
the same (standard) logistic distribution, and thus, RUMs based on LEVI and
SEVI random components are equivalent when there are only two alternatives.
Second, there is a long standing belief (%
\citet{Hausman_Wise1977}%
;
\citet{Horowitz1982}%
) that\ a multinomial probit (MNP) model produces essentially the same results
as the conditional/multinomial logit model.\footnote{As succinctly summarized
in the popular graduate econometric text by
\citet{Greene2018}%
: \textquotedblleft An MNP model that replaces a normal distribution with
$\Sigma=I$ [zero covariance between the random components] will yield
virtually identical results (probabilities and elasticities) compared to the
multinomial logit model.\textquotedblright}\ This belief has been taken to
suggest that it is the iid assumption imposed on the random components that
drives the well-known independence of irrelevant alternatives (IIA)
substitution pattern. However, insights from the case with only two
alternatives do not extend to the case with more than two alternatives. In
this paper, we show that as the number of alternatives increases, the choice
probabilities from a LEVI-based conditional logit and a SEVI-based conditional
logit model increasingly diverge from each other. We also show that the IIA
property does not follow from the iid assumption. Even though the random
components are still iid, a SEVI-based conditional logit model with more than
two alternatives exhibits a distinctly non-IIA substitution pattern.

Let's consider a simple stylized example. An agent (individual $i$) enters a
store with the intention of purchasing a pair of running shoes, and there are
five different alternatives $(j\in\{1,2,3,4,5\}$ and $J=5)$ to choose from.
Before trying them on, the individual rates each of the five pairs based on
observable features like brand and price. These ratings capture the systematic
components, $\left\{  V_{j}\right\}  $, of a random utility model. For
illustrative purposes, assume that the $V_{j}$'s are known to be
$\{0.25,0.50,0.75,1.50,$ and $2.00\},$ so there is no need to specify the shoe
attributes and the associated preference parameters.\footnote{The vector of
systematic utilities can vary among different individuals. Here, for
simplicity, we assume that $V_{j}$'s are the same across all individuals.}
Next, the individual tries on each of the five pairs of shoes, making an
idiosyncratic judgement on comfort and fit. This provides the random
component, $\varepsilon_{ij}$, to add to the systematic component, $V_{j}$, to
get the total utility level, $U_{ij}=V_{j}+\varepsilon_{ij},$ for each pair of
shoes. The chosen pair for purchase is the one with the highest $U_{ij}$.

Now let the store's customers comprise of a large number of individuals
($i=1,\ldots,n$), with identical and known systematic components $\left\{
V_{j}\right\}  _{j=1}^{5}$, but different random components $\varepsilon
_{ij}.$ They try on the shoes, therefore observing their individual random
components, then make their decision on which of the five pairs of shoes to
purchase. The question we posed can now be framed solely in terms of: what are
the choice probabilities cast in percentage terms (or equivalently, market
shares) for each of the five pairs of shoes, if the $\varepsilon_{ij}$ are
drawn from the following distributions: (a) the standard LEVI distribution,
(b) the standard SEVI distribution, (c) N$(0,\pi^{2}/6)$ (hereafter referred
to as NORM, a normal distribution normalized to have the same variance as the
standard LEVI and SEVI distributions)?

The answer to this question is displayed in Figure
\ref{Figure: Example_with_5_options}. The results are quite striking,
particularly for the alternatives with the smallest and largest systematic
utilities. With a LEVI error component, the alternative with the smallest
systematic utility has a market share 137\% higher than its SEVI counterpart
and 25\% higher than its independent normal counterpart. This pattern
continues for the alternatives with the second smallest and third smallest
systematic utilities, with LEVI giving rise to market shares 73\% and 37\%
higher than SEVI, respectively. However, this pattern of a higher LEVI market
share reverses around a 20\% market share (i.e., $1/J$ or equal share for all
alternatives), where the SEVI distribution leads to a higher market share than
the LEVI distribution. For the alternative with the second-largest systematic
utility, the SEVI market share is 11\% higher than its LEVI counterpart and
5\% higher than that under the NORM distribution. The pattern of a higher
market share under SEVI intensifies for the alternative with the largest
systematic utility. Here, SEVI leads to a market share 21\% higher than LEVI
and 16\% higher than NORM.%

\begin{figure}[h]%
\centering
\includegraphics[
height=3.0692in,
width=3.7758in
]%
{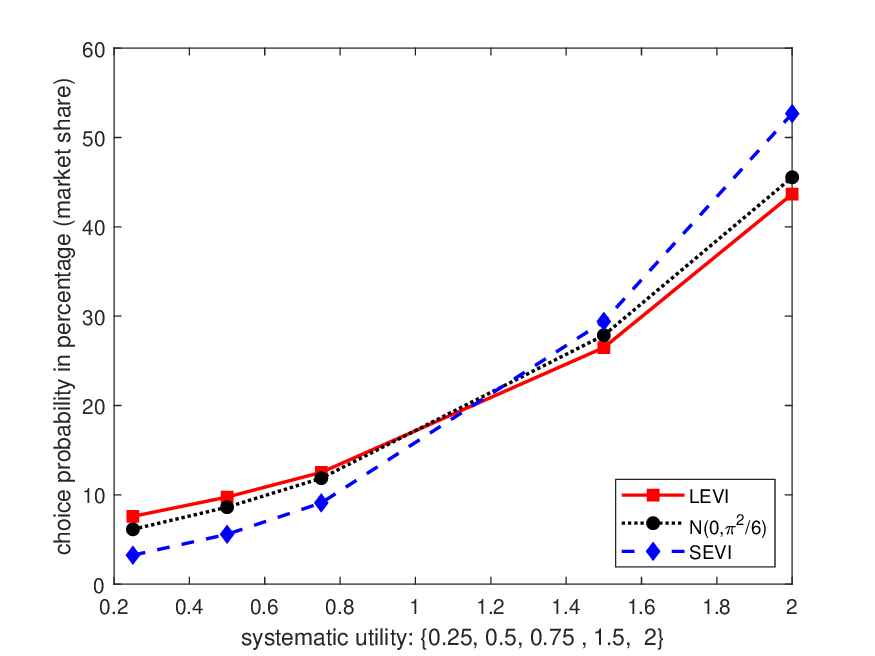}%
\caption{Market shares for all five alternatives in an RUM with identical
systematic utilities but different distributions of the random utility
components}%
\label{Figure: Example_with_5_options}%
\end{figure}

The magnitude of all of these differences is likely to be important in
empirical applications. For instance, if Alternative 1 were a new product, use
of a LEVI-based prediction when the true error component was SEVI would have
led to an overestimation of demand by more than double (7.6\% vs. 3.2\% market
shares). The potential for the random component to be SEVI rather than LEVI
provides one possible explanation for why most new products perform much worse
than predicted by marketing researchers. On the other hand, a LEVI-based
prediction for the market leader, when the true error component is SEVI, could
substantially underestimate its market share (43.7\% vs. 52.7\%). It is easy
to see how either of these faulty predictions could be disastrous from both
production and profit perspectives. Antitrust implications are also obvious.

Why do these differences in choice probabilities arise, despite these three
simple models sharing the same (known) systematic components and random
components from iid draws from simple, single-peaked distributions with the
same finite variance? The answer lies in the skewness of the random
components. From the econometrician's viewpoint, the fundamental question in
our stylized example is: for a randomly chosen agent, is the overall utility
from a randomly chosen pair of shoes (a) more likely to move up, (b) stay the
same, or (c) move down? Surprisingly, this fundamental property of the random
component of a RUM seems to have never been raised in the voluminous
literature on multinomial choice models and serves as the driving force behind
the present paper.

If the total utility $U_{ij}$ of an arbitrarily chosen agent $i$ for an
arbitrarily chosen alternative $j$ is more likely to move upward relative to
the systematic utility $V_{ij}$ after receiving the private signal
$\varepsilon_{ij}$ representing the random component, then the random
component $\varepsilon_{ij}$ is right-skewed. The classic exemplar in this
scenario is the LEVI distribution, and the canonical model is the conditional
logit. If the total utility for this agent is just as likely to move up as
move down, then the random component is symmetric. The normal distribution is
the exemplar, and the multinomial probit is the canonical model. If this
agent's total utility is more likely to move down, then the random component
is left-skewed. There is not a similar exemplar distribution or canonical
model for this left-skewed case. We propose employing the SEVI distribution
and its corresponding model for the left-skewed scenario, serving as the
natural counterpart to the standard LEVI formulation.

Once the potential importance of the skewness of the random component is
acknowledged, an entirely new line of inquiry is opened up. Are there
conditions (e.g., in-store vs. online purchases) or deep personality traits
(e.g., risk averse vs. risk loving) that influence the random component's
skewness? We don't provide answers to this question but hope that we provide
tools that can help researchers explore these issues.

Fitting an RUM based on LEVI random components to data generated with SEVI
random components results in substantially biased estimates of the preference
parameters and vice versa. For the same preference parameters in the RUM's
systematic components, SEVI-based random components result in substantively
different (compared to LEVI-based random components) estimates for quantities
such as elasticities and average partial derivatives routinely relied upon in
decision making.\footnote{We identify one specific scenario: a ratio
estimator, such as those used in many willingness-to-pay calculations, where
the LEVI-based estimator remains consistent even if the true error component
is SEVI, provided that an important set of restrictions on the attribute
matrix is met.} Standard generalizations of the multinomial/conditional logit
do not address issues that arise when a LEVI random component is assumed,
incorrectly, when SEVI is appropriate. Instead, SEVI-based analogues of those
models exist to deal with the specific deviations from the standard
multinomial/conditional logit that those models are intended to address.

While we make no claim that the SEVI assumption is always or even generally
preferable to that of LEVI, the empirical relevance of the SEVI model is
easily demonstrated. When examining a set of six datasets taken from standard
econometric textbooks and well-known papers, we find that the SEVI
distributional assumption is consistently preferred over LEVI and is favored
over the normal distribution in five out of the six datasets. For any
particular application, the skewness of the random component is obviously an
empirical question, and a priori, nothing rules out mixtures of random
component skewness types, an issue we examine near the end of this paper (see
Section \ref{Subsection: mixed LEVI-SEVI}).

The bulk of the paper is devoted to studying the theoretical properties of a
SEVI model\textbf{ }and delving into aspects of the estimation, inference, and
prediction within the SEVI framework. Here, we offer a preview of some
intriguing properties of the SEVI model.

First, both the choice probabilities and the surplus function in a SEVI model
have intuitive and closed-form representations, allowing us to explore new
avenues of discrete choice analysis. Using the strict convexity of the social
surplus function, we show that a SEVI model is identified under the same
conditions that ensure the identification of a LEVI model.

Second, the structure of the choice probability for any given alternative in a
SEVI model is considerably richer than in the LEVI model, where its dependence
on other alternatives is solely through the sum of the exponentiated
systematic components for each alternative. In the SEVI model, the choice
probability for a given alternative can be seen as an aggregation involving
all choice subsets in which the alternative occurs and the overall choice set.
This dependence of the choice probability on all subsets of alternatives is a
distinctive feature of the SEVI model that differentiates it from the LEVI
model. As a consequence of such dependence, IIA does not hold in the SEVI model.

Third, on its surface, the all-subsets representation of SEVI might seem
computationally challenging when the number of alternatives $J$ is of any
appreciable size. We overcome this issue by employing Gosper's hack
(\citet{KnuthDonald2011TAoC}) from the computer science literature to
efficiently enumerate all possible subsets containing a target alternative.
For a moderately large $J$ $\left(  J\leq15\right)  $ with typical sample
sizes and attribute numbers in economic analysis, the computation time for a
SEVI model is somewhat higher than its standard logit LEVI counterpart but
much lower than that of the MNP with iid normal errors. In the case of $J=12,$
the computation time for a SEVI model can be less than 2\% of that of the
corresponding independent MNP model.\footnote{It is an open issue how SEVI's
enumeration of subsets of alternatives interacts with the computational
intensity inherent in the generalizations of the conditional logit model,
where long run times are often the norm.}

Fourth, we explore the regions where the SEVI model is likely to produce
results similar to the LEVI model and where the results are most likely to
diverge. We develop a Vuong-type test, showing its effectiveness in
distinguishing between the SEVI and LEVI models (%
\citet{Vuong1989}%
). We also show that AIC/BIC can be employed to help guide appropriate model choice.

Fifth, we explore the implications of the skewness of the random component as
the number of alternatives increases, going beyond the case with only two
alternatives. One immediate implication is that in an RUM with SEVI random
components, the observed choice and the choice that maximizes the systematic
utility are more likely to correspond than in the case of LEVI. This is
because SEVI error components have less of an impact on determining the
selected alternative compared to LEVI ones. In contrast, SEVI random
components play a larger role than their LEVI counterparts in determining what
is not chosen.\footnote{One way to conceptualize the role of the random
components, consistent with this observation, is that an agent with SEVI
random components seeks idiosyncratic reasons to reject alternatives, whereas
under LEVI, the agent aims to find a high-quality idiosyncratic match
component to accept an alternative.}

The rest of the paper is organized as follows. Section
\ref{Sec: Literature Review} provides a literature overview, largely from a
historical perspective, that hints at why the importance of the skewness of
the random component was overlooked. Section \ref{section choice prob} studies
the structure of the choice probabilities in a SEVI model, contrasting its
properties with those of the conventional LEVI model. This section shows that
while Luce's choice axiom holds for the SEVI model when $J=2$, this
restriction does not hold when $J>2$. This section also provides a closed-form
expression for the surplus function and shows how it can be used to compute
the compensating variation and establish the identification of the SEVI model.
Section \ref{Section: MLE_QMLE} considers estimating the SEVI model by MLE and
QMLE, studies the consistency of the QMLE up to a scale normalization, and
describes the use of AIC/BIC and the Vuong test for selecting between the SEVI
and LEVI models. Section \ref{Sec: simulation} presents a variety of
simulation results, exploring different aspects of the SEVI model. Section
\ref{Sec: extension} introduces a mixed model that accommodates individuals of
either LEVI or SEVI type in a population. It also contains a straightforward
extension to discrete choice problems where agents aim to minimize, rather
than maximize, an objective function, such as cost and regret. Section
\ref{Sec:Emp} contains a set of empirical applications drawn from textbooks
and well-known papers, allowing us to compare the goodness of fit using LEVI,
SEVI, and independent MNP with the same error variance. The final section
offers concluding remarks and outlines some future research directions based
on generalizing the SEVI model. Proofs are provided in the appendix, and an
online supplementary appendix contains additional materials and figures.

\section{Literature Review: a Historical
Perspective\label{Sec: Literature Review}}

Early empirical research on individual-level choices was hindered by the lack
of data and appropriate statistical models. Operationally, the study of choice
requires a conceptualization involving a systematic component and a random
component. The psychologist
\citet{Thurstone1927}
is generally credited with pioneering this concept and proposing a variant of
the now-familiar binary probit model. The binary logit model soon emerged from
biology. Statistical advancements related to the random component, whether
logistic or normal, largely took place in the realm of biometrics (%
\citet{Cox1969}%
).

The transition from the binary choice case to the multinomial choice proved
not to be easy. While significant advancements occurred within biometrics
(e.g.,
\citet{Cox1969}%
), the pendulum swung back to psychology and, in particular, to the seminal
work of
\citet{luce1959individual}%
. This effort quickly spilled over into economics (e.g.,
\citet{BlockMarschak1959}%
). Formal requirements for choice models to meet the rationality requirements
of utility maximization were explored, and patterns of choice behavior, such
as the IIA property in the multinomial logit model, were examined. The
multinomial logit model and its statistically equivalent flipped version,
known as the conditional logit model, came together in its current form with
the tour-de-force synthesis of
\citet{McFadden1973}%
. The work of McFadden and those following in his footsteps spawned a truly
massive empirical enterprise in economics and other social sciences (see
\citet{McFadden2001}%
). This involves the construction and application of various random utility
models, almost all based on the LEVI kernel of the conditional logit.

There are many reasons why the LEVI distribution was originally chosen for
RUMs and has become popular in applied research.

First, utility is inherently unobservable and hence is typically represented
as a latent variable. The assumption of continuity makes economic theory and
applied work much more tractable. This ruled out discrete distributions.

Second, in the case of a binary response, there had been a long-running debate
over tacking on an error term, first from a normal and later from a logistic
distribution, in bioassay experiments. In these experiments, the error
component was perceived as coming from an unobserved tolerance distribution.
When response data were plotted, bell-shaped tolerance distributions seemed
more appropriate than either a uniform or bimodal distribution. Eventually,
the binary logit model prevailed over its probit counterpart because of its
computational advantages and the early literature's observation of largely
indistinguishable empirical results across the two models with sample sizes
available at the time (%
\citet{Chambers_Cox1967}%
). This outcome also influenced the adoption of the binary logit model in the
social sciences; see Chapter 9 of
\citet{Cramer2003}
for a historical account.

Third, there was a long-standing belief that the multinomial logit and the
multinomial probit with iid normal unobserved utilities were very close
approximations to each other.\footnote{\label{footnote_LEVI_MNP copy(1)}%
\citet{Hausman_Wise1977} state: \textquotedblleft We might expect the
independent probit and the logit models to have similar properties [including
inducing the IIA red bus/blue bus property via the iid random component
assumption] and, in fact, they lead to almost identical empirical
results.\textquotedblright\ \citet{Horowitz1982} proposes a Lagrange
multiplier test of the unrestricted MNP model, assuming that the MNL estimates
adequately approximate the restricted MNP results.} The belief may have led to
the misconception that there is no need to go beyond the multinomial logit if
the random utility components are iid. However, this belief turns out not to
be generally true. The misconception likely stems from limitations in
computational techniques and resources, restricting early literature to
examine only cases with a small number of alternatives. When the number of
alternatives, $J,$ is small, particularly when $J=2$ and $J=3$, the estimates
from an independent multinomial probit with a common variance normalization
tend to be very close to those from a LEVI multinomial logit. Nevertheless, as
shown in Figure \ref{Figure: Example_with_5_options}, the multinomial logit
and probit models can exhibit significant differences for even a moderate $J,$
such as $J=5.$

These, however, are not reasons for not exploring the implications of other
distributions, as
\citet{Manski1977}
had urged. This is particularly relevant in light of the increased
computational power and insights from behavioral economics, which suggest
frequent violations of the IIA property.

The \textquotedblleft non-IIA\textquotedblright\ property of a SEVI model
bears some resemblance to that of the mother/universal logit (%
\citet{McFaddenTrainTye1978}%
) and the random regret minimization model (e.g.,
\citet{MAI2017}%
). However, the mechanisms of removing the IIA property are fundamentally
different. By definition, a SEVI model considered here is a utility
maximization model, and hence the resulting choice probability function can be
rationalized by a probability distribution over linear preference orderings (%
\citet{BlockMarschak1959}%
). The absence of the IIA property is due solely to the distribution of the
error component. In contrast, most implementations of the mother/universal
logit and the random regret minimization models are inconsistent with
random\ utility maximization, and the lack of the IIA property in these models
is a result of the specification that the difference in the systematic (not
random) component of utility between any two alternatives depends on the
presence or absence of another alternative.

A significant portion of the history of this field revolves around building
models utilizing the LEVI kernel to relax the IIA constraint. This is achieved
by permitting agents to have different preferences or scale parameters or by
allowing subsets of alternatives to be correlated in some fashion, while still
maintaining computational feasibility (%
\citet{Hensher_Rose_Greene_2015}%
;
\citet{book_train_2009}%
;
\citet{Greene2018}%
). Among these models, random-parameter (mixed) logit, latent-class
multinomial logit, scale-heterogeneity logit, generalized multinomial logit,
and nested logit stand out as the most prominent examples. All these more
general specifications can be constructed using a SEVI rather than a LEVI
random component. Conceptually, specifying the random component's skewness
should precede relaxing restrictive assumptions of the conditional logit
model, such as the independence of random components and the homogeneity of
preferences across individuals.

Our discovery of the LEVI vs. SEVI distinction came about through an
inadvertent sign flip while generating draws from a standard LEVI
distribution. In doing so, a SEVI random component was generated instead. This
would likely have gone unnoticed if our objective had not been to study
whether the behavior of agents facing choice sets of varying sizes differed.
In a baseline simulation, we observed sizeable differences in the conditional
logit model estimates of preference parameters and other statistics, as the
number of alternatives increased. These unexpected results led us to realize
that we had fit a LEVI-based logit model to data generated by a SEVI logit. A
search of general and specialized econometrics texts revealed no discussion of
LEVI vs. SEVI or the broader issue of the role of the skewness of the random
component in choice behavior.

We eventually found one short and largely ignored paper in the literature (%
\citet{Hu2005}%
) which points out that the likelihood functions of SEVI and LEVI-based choice
models are different. However, that paper does not delve into any of the other
results presented here. It asserts that the SEVI model does not appear to have
a closed-form solution and hence has to be estimated using a simulation
approach. In Theorem \ref{Theorem: choice_prob_under_SEV}, we show this not to
be the case. SEVI has a mathematically elegant and intuitive closed-form
solution that facilitates direct comparisons between the properties of LEVI
vs. SEVI choice models. In an empirical marketing application with three
alternatives involving types of bread,
\citet{Hu2005}
shows that the SEVI model fits marginally better than the LEVI model, but the
difference is smaller than the variation due to different smoothing parameter
values used in the simulation approach. To our knowledge, the SEVI model has
not been used in subsequent empirical applications.

\section{The Model: Choice Probability and Surplus Function
\label{section choice prob}}

\subsection{The SEVI-based RUM}

We consider a probabilistic choice model based on utility maximization. We
assume that the utility of individual $i$ from alternative $j$ is given by
\begin{equation}
U_{ij}\left(  \beta_{0}\right)  =V_{ij}\left(  \beta_{0}\right)
+\varepsilon_{ij},\text{ }j=1,2,\ldots,J, \label{RUM}%
\end{equation}
where $J$ is the number of alternatives, $V_{ij}\left(  \beta_{0}\right)  $ is
the observable utility that individual $i$ obtains from choosing alternative
$j$, $\varepsilon_{ij}$ captures the utility component unobservable to the
econometrician, and $\{\varepsilon_{ij}\}$ is independent of $\left\{
V_{ij}\left(  \beta_{0}\right)  \right\}  $. As is customary in the
literature, we refer to $V_{ij}\left(  \beta_{0}\right)  $ as the\emph{
systematic} utility and $\varepsilon_{ij}$ as the \emph{random/stochastic}
utility. A common parametrization of the systematic utility is $V_{ij}\left(
\beta_{0}\right)  =Z_{i}\beta_{0,z,j}+X_{ij}\beta_{0,x}$ where $X_{ij}$ is a
vector of attributes of alternative $j$ as perceived by individual $i$ and
$Z_{i}$ is a vector of characteristics of individual $i$.

Individuals choose the alternative that maximizes their total utility. Let
$Y_{i}$ denote the choice of individual $i;$ then%
\begin{equation}
Y_{i}=\arg\max_{j}\left\{  U_{i1}\left(  \beta_{0}\right)  ,\ldots
,U_{iJ}\left(  \beta_{0}\right)  \right\}  .
\end{equation}
One of our goals is to estimate the model parameter $\beta_{0}$ based on
individual choice data.

For the stochastic utility, we assume that $\varepsilon_{ij}$ is iid over
$j=1,2,\ldots,J$ and follows the extreme value type I distribution based on
the \emph{smallest} extreme value. The CDF of the\ (standard) smallest extreme
value type I (SEVI) distribution is
\begin{equation}
F(a)=1-\exp\left(  -\exp(a)\right)  ,
\end{equation}
and the corresponding pdf is
\begin{equation}
f(a)=F^{\prime}(a)=\exp\left(  a-\exp(a)\right)  . \label{SEV_pdf}%
\end{equation}
Our distributional assumption is in contrast with that in standard logit
models, where the latent utility is given by $\tilde{U}_{ij}\left(  \beta
_{0}\right)  =V_{ij}\left(  \beta_{0}\right)  +\tilde{\varepsilon}_{ij}$, and
the stochastic utility $\tilde{\varepsilon}_{ij}$ follows the (standard)
extreme value Type I distribution based on the \emph{largest} extreme value
(i.e., the \emph{largest} extreme value type I (LEVI) distribution). These two
forms of extreme value distributions are related in that $\left(
\varepsilon_{i1},\ldots,\varepsilon_{iJ}\right)  ^{\prime}$ and $-\left(
\tilde{\varepsilon}_{i1},\ldots,\tilde{\varepsilon}_{iJ}\right)  ^{\prime}$
have the same distribution, so they can be regarded as mirror images of each
other. For easy reference, we call a stochastic utility following the SEVI
distribution a SEVI error and the resulting RUM as the SEVI model. The same
nomenclature applies to the LEVI case.\footnote{Throughout the paper, the SEVI
is the standard SEVI with pdf given in (\ref{SEV_pdf}) and variance $\pi
^{2}/6.$ A similar comment applies to the LEVI case.}%

\begin{figure}[ptb]%
\centering
\includegraphics[
height=2.5244in,
width=3.9038in
]%
{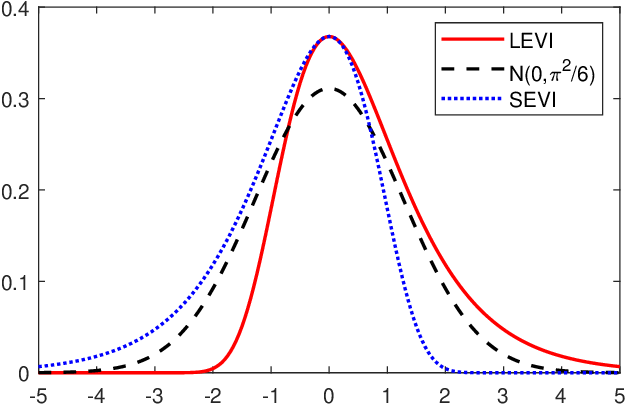}%
\caption{Probability density functions of the LEVI, NORM, and SEVI
distributions}%
\label{Figure: LEV and SEV}%
\end{figure}

Figure \ref{Figure: LEV and SEV} illustrates the pdfs of the LEVI and SEVI
distributions. The LEVI distribution exhibits a strong right-skewness with a
heavier upper tail in contrast to its lower tail, while SEVI exhibits the
opposite (mirror image) behavior. As a symmetric benchmark, the normal pdf
with the same variance is also plotted. The SEVI distribution is appropriate
when there is a prior belief that the random utility component is more likely
to have a negative rather than a positive or neutral effect on the total
utility. Such a belief is plausible, and we have been unable to find a strong
\emph{a priori} argument against using a SEVI random component or a mixture of
SEVI and LEVI components, as we discuss in Section
\ref{Subsection: mixed LEVI-SEVI}.

This notion of random component skewness has to be seen as conditional on what
alternative attributes and agent characteristics are observed by the
econometrician. To see this, note that in the case of a linear latent utility
model, $\varepsilon_{ij}$ consists of $(\tilde{X}_{ij}-E\tilde{X}_{ij}%
)\tilde{\beta}_{0,x}$ for a scalar attribute $\tilde{X}_{ij}$ known to
individual $i$ but not observable to the econometrician. If $\tilde{X}_{ij}$
has a negative skewness and $\tilde{\beta}_{0,x}>0$, then $(\tilde{X}%
_{ij}-E\tilde{X}_{ij})\tilde{\beta}_{0,x}$ will also have a negative skewness,
implying that $\varepsilon_{ij}$ may exhibit left-skewness and follow the SEVI
distribution. Thus, the inclusion or exclusion of a particular covariate may
change the skewness of the stochastic utility component. More generally, the
system utility may be misspecified, with the resulting specification error
effectively becoming part of the stochastic utility. Given that specification
errors can be manifest as left-skewed, symmetric, or right-skewed, the
stochastic utility can display any of these forms of skewness as a
consequence. This argument lends another support to the notion that the SEVI
distribution of the stochastic utility is at least as plausible as the LEVI distribution.

A logical question at this point is whether the SEVI distribution belongs to
the class of Generalized Extreme Value (GEV) distributions with the CDF given
by $F_{\mathrm{GEV}}(a;\xi)=\exp(-\left(  1+\xi a\right)  ^{-1/\xi})$ for some
parameter $\xi.$\ The SEVI distribution, however, does not align with a GEV
distribution.\footnote{When $\xi\neq0,$ the GEV distribution is not supported
on the whole real line. When $\xi=0,$ the GEV distribution is right-skewed.}
Nor does it correspond to the marginal distribution of a multivariate extreme
value (MEV) distribution with CDF:
\begin{equation}
F_{\mathrm{MEV}}(a_{1},\ldots,a_{J})=\exp(-\mathbb{G}\left(  \exp\left(
-a_{1}\right)  ,\ldots,\exp\left(  -a_{J}\right)  \right)  , \label{MEV}%
\end{equation}
for some function $\mathbb{G}$. In both cases, the marginal distribution is
right-skewed instead of being left-skewed like a SEVI distribution. See, for
example, Chapter 4 of
\citet{book_train_2009}%
. In principle, exploring an extension of the GEV or MEV family by using the
SEVI as the marginal distribution, potentially encompassing the SEVI as a
special case within this broader framework, could be considered for future research.

\subsection{Choice Probabilities in the SEVI-based RUM}

In a SEVI-based RUM, similar to a LEVI-based RUM, the choice probabilities
also have closed-form expressions, albeit a bit more complex. To simplify the
notation, we suppress the dependence of $V_{ij}\left(  \beta_{0}\right)  $ and
$V_{i}\left(  \beta_{0}\right)  :=(V_{i1}\left(  \beta_{0}\right)
,\ldots,V_{iJ}\left(  \beta_{0}\right)  )^{\prime}$ on $\beta_{0}$, and we
write $U_{ij}:=U_{ij}\left(  \beta_{0}\right)  ,$ $V_{ij}:=V_{ij}\left(
\beta_{0}\right)  ,$ and $V_{i}:=V_{i}\left(  \beta_{0}\right)  $. We adopt
the same convention for the realized values $v_{ij}\left(  \beta_{0}\right)  $
and $v_{i}\left(  \beta_{0}\right)  $ of $V_{ij}\left(  \beta_{0}\right)  $
and $V_{i}\left(  \beta_{0}\right)  $ in the theorem below.

\begin{theorem}
\label{Theorem: choice_prob_under_SEV}Assume that $U_{ij}=V_{ij}%
+\varepsilon_{ij}$, $j=1,2,\ldots,J$ where $\varepsilon_{i}:=\left(
\varepsilon_{i1},\ldots,\varepsilon_{i,J}\right)  ^{\prime}$ is independent of
$V_{i}$ and the elements of $\varepsilon_{i}$ follow independent and identical
SEVI distributions. Then, under utility maximization, the probability of
individual $i$ choosing alternative $j$ is
\begin{align}
P^{\mathrm{SEVI}}\left(  Y_{i}=j|v_{i}\right)   &  :=P^{\mathrm{SEVI}}\left(
Y_{i}=j|V_{i}=v_{i}\right) \nonumber\\
&  =1+\sum_{\ell=1}^{J-1}\left(  -1\right)  ^{\ell}\sum_{k_{1}<\ldots<k_{\ell
}:\left\{  k_{1},\ldots,k_{\ell}\right\}  \subseteq\mathcal{J}_{-j}}\frac
{\exp\left(  -v_{ij}\right)  }{\exp\left(  -v_{ij}\right)  +\sum_{k\in\left\{
k_{1},\ldots,k_{\ell}\right\}  }\exp\left(  -v_{ik}\right)  }%
\label{SEVI_prob_formula}\\
&  =1+\sum_{\ell=1}^{J-1}\left(  -1\right)  ^{\ell}\sum_{k_{1}<\ldots<k_{\ell
}:\left\{  k_{1},\ldots,k_{\ell}\right\}  \subseteq\mathcal{J}_{-j}}\frac
{1}{1+\sum_{k\in\left\{  k_{1},\ldots,k_{\ell}\right\}  }\exp\left(
-v_{ik}+v_{ij}\right)  },\nonumber
\end{align}
where, for $\mathcal{J}=\left\{  1,\ldots,J\right\}  ,$ $\mathcal{J}%
_{-j}=\mathcal{J}\backslash\left\{  j\right\}  $ is the subset of alternatives
leaving the $j$-th alternative out.
\end{theorem}

\begin{remark}
In a standard LEVI conditional logit model, the choice probability is
\begin{equation}
P^{\mathrm{LEVI}}\left(  Y_{i}=j|v_{i}\right)  =\frac{\exp(v_{ij})}{\sum
_{k=1}^{J}\exp\left(  v_{ik}\right)  }. \label{LEVI_prob_formula}%
\end{equation}
We can compare the choice probabilities across the SEVI and LEVI models using
a few examples. When there are only two alternatives so that $J=2,$ we have,
using Theorem \ref{Theorem: choice_prob_under_SEV},
\begin{align*}
P^{\mathrm{SEVI}}\left(  Y_{i}=1|v_{i}\right)   &  =1-\frac{\exp\left(
-v_{i1}\right)  }{\exp\left(  -v_{i1}\right)  +\exp\left(  -v_{i2}\right)  }\\
&  =\frac{\exp(-v_{i2})}{\exp(-v_{i1})+\exp(-v_{i2})}=\frac{\exp(v_{i1})}%
{\exp(v_{i1})+\exp(v_{i2})}.
\end{align*}
The choice probability is exactly the same as what we obtain in a standard
LEVI model. Therefore, in the presence of only two alternatives, it does not
matter whether the stochastic utility follows the LEVI or SEVI distribution.
This is because the difference between two independent LEVI or SEVI draws
follows the same logistic distribution.\footnote{By Lemma 5 of
\citet{YELLOTT1977}%
, the condition that the difference of two iid random components is logistic
is necessary for the IIA to hold. Our results presented later can be seen as
showing the existence of a simple distribution, SEVI, which meets this
necessary condition, yet IIA does not hold.}

When there are more than two alternatives so that $J>2,$ the choice
probabilities in a SEVI model are different from those in a LEVI model.
Consider the case with $J=3.$ It follows from Theorem
\ref{Theorem: choice_prob_under_SEV} that%
\begin{align}
P^{\mathrm{SEVI}}\left(  Y_{i}=1|v_{i}\right)   &  =1-\frac{\exp\left(
-v_{i1}\right)  }{\exp\left(  -v_{i1}\right)  +\exp\left(  -v_{i2}\right)
}-\frac{\exp\left(  -v_{i1}\right)  }{\exp\left(  -v_{i1}\right)  +\exp\left(
-v_{i3}\right)  }\label{Prob_intuition_J3}\\
&  +\frac{\exp\left(  -v_{i1}\right)  }{\exp\left(  -v_{i1}\right)
+\exp\left(  -v_{i2}\right)  +\exp\left(  -v_{i3}\right)  }.\nonumber
\end{align}
This is clearly different from the LEVI choice probability, which is equal to
\[
P^{\mathrm{LEVI}}\left(  Y_{i}=1|v_{i}\right)  :=\frac{\exp\left(
v_{i1}\right)  }{\exp\left(  v_{i1}\right)  +\exp\left(  v_{i2}\right)
+\exp\left(  v_{i3}\right)  }.
\]
To understand the formula for $P^{\mathrm{SEVI}}\left(  Y_{i}=1|v_{i}\right)
,$ we can examine $1-P^{\mathrm{SEVI}}\left(  Y_{i}=1|v_{i}\right)  ,$ which
represents the probability that alternative 1 is \emph{not} chosen. The latter
event can be expressed as a union of two events $\left\{  U_{i1}\leq
U_{i2}\right\}  \cup\left\{  U_{i1}\leq U_{i3}\right\}  $, indicating that
alternative 1 is dominated by either alternative 2 or alternative 3. But the
probability of this union is
\begin{align*}
&  \Pr\left[  \left\{  U_{i1}\leq U_{i2}\right\}  \cup\left\{  U_{i1}\leq
U_{i3}\right\}  \right] \\
&  =\Pr\left[  U_{i1}\leq U_{i2}\right]  +\Pr\left[  U_{i1}\leq U_{i3}\right]
-\Pr\left[  \left\{  U_{i1}\leq U_{i2}\right\}  \cap\left\{  U_{i1}\leq
U_{i3}\right\}  \right] \\
&  =\Pr\left[  -U_{i1}\geq-U_{i2}\right]  +\Pr\left[  -U_{i1}\geq
-U_{i3}\right]  -\Pr\left[  \left\{  -U_{i1}\geq-U_{i2}\right\}  \cap\left\{
-U_{i1}\geq-U_{i3}\right\}  \right]  .
\end{align*}
Each of the three probabilities above corresponds to a choice probability in a
conventional LEVI model. Substituting the familiar LEVI formulas yields
Equation (\ref{Prob_intuition_J3}), and\ Theorem
\ref{Theorem: choice_prob_under_SEV} is a generalization of these arguments.
\end{remark}

\begin{remark}
When $J=4,$ Theorem \ref{Theorem: choice_prob_under_SEV} reveals that%
\begin{align*}
&  P^{\mathrm{SEVI}}\left(  Y_{i}=1|v_{i}\right) \\
&  =1-\frac{\exp\left(  -v_{i1}\right)  }{\exp\left(  -v_{i1}\right)
+\exp\left(  -v_{i2}\right)  }-\frac{\exp\left(  -v_{i1}\right)  }{\exp\left(
-v_{i1}\right)  +\exp\left(  -v_{i3}\right)  }-\frac{\exp\left(
-v_{i1}\right)  }{\exp\left(  -v_{i1}\right)  +\exp\left(  -v_{i4}\right)  }\\
&  +\frac{\exp\left(  -v_{i1}\right)  }{\exp\left(  -v_{i1}\right)
+\exp\left(  -v_{i2}\right)  +\exp\left(  -v_{i3}\right)  }+\frac{\exp\left(
-v_{i1}\right)  }{\exp\left(  -v_{i1}\right)  +\exp\left(  -v_{i2}\right)
+\exp\left(  -v_{i4}\right)  }\\
&  +\frac{\exp\left(  -v_{i1}\right)  }{\exp\left(  -v_{i1}\right)
+\exp\left(  -v_{i3}\right)  +\exp\left(  -v_{i4}\right)  }\\
&  -\frac{\exp\left(  -v_{i1}\right)  }{\exp\left(  -v_{i1}\right)
+\exp\left(  -v_{i2}\right)  +\exp\left(  -v_{i3}\right)  +\exp\left(
-v_{i4}\right)  },
\end{align*}
which involves enumerating all choice subsets that contain the first
alternative. In contrast, the corresponding choice probability in a LEVI model
is
\[
P^{\mathrm{LEVI}}\left(  Y_{i}=1|v_{i}\right)  =\frac{\exp\left(
v_{i1}\right)  }{\left[  \exp\left(  v_{i1}\right)  +\exp\left(
v_{i2}\right)  +\exp\left(  v_{i3}\right)  +\exp\left(  v_{i4}\right)
\right]  },
\]
which involves only the largest choice set. As shown in Figure
\ref{Figure: Example_with_5_options}, and later, when $J$ becomes larger than
2, the choice probabilities in a SEVI model can be different from those in a
LEVI model in an economically meaningful way. This has practical implications,
underscoring the need for careful consideration when formulating distribution
assumptions in discrete choice modeling.
\end{remark}

\begin{remark}
\label{Remark: LEVI_SEVI_equal}When $v_{ij}=v_{ik}$ for all $j$ and $k,$ we
have%
\begin{align*}
P^{\mathrm{SEVI}}\left(  Y_{i}=j|v_{i}\right)   &  =1+\sum_{\ell=1}%
^{J-1}\left(  -1\right)  ^{\ell}\sum_{k_{1}<\ldots<k_{\ell}:\left\{
k_{1},\ldots,k_{\ell}\right\}  \subseteq\mathcal{J}_{-j}}\frac{1}{\ell+1}\\
&  =\sum_{\ell=0}^{J-1}\left(  -1\right)  ^{\ell}\binom{J-1}{\ell}\frac
{1}{\ell+1}=\int_{0}^{1}\sum_{\ell=0}^{J-1}\left(  -1\right)  ^{\ell}%
\binom{J-1}{\ell}s^{\ell}ds\\
&  =\int_{0}^{1}\left(  1-s\right)  ^{J-1}ds=\frac{1}{J}.
\end{align*}
This result can also be obtained by using a symmetric argument. Since all
alternatives have the same systematic utility, they are indistinguishable and
hence have an equal chance of being chosen. The result is the same as that
obtained under the LEVI assumption. On the other hand, when $v_{ij}$ increases
and dominates all other $v_{ik}$ for $k\neq j,$ $P^{\mathrm{SEVI}}\left(
Y_{i}=j|v_{i}\right)  $ approaches 1. The choice probability $P^{\mathrm{SEVI}%
}\left(  Y_{i}=j|v_{i}\right)  $ can, therefore, be regarded as a type of
softmax function that maps $\mathbb{R}^{J}$ into $\left(  0,1\right)  ^{J}.$
For example, it maps $\left(  1,2,8\right)  $ into $(4.24\times10^{-4},$
$2.29\times10^{-3},0.997),$ which amounts to assigning almost all probability
mass to the element with maximal systematic utility. This result is
qualitatively similar to that obtained under the LEVI assumption.
\end{remark}

\begin{remark}
The probability of choosing alternative $j$ in a SEVI model involves a
summation over all subsets of $\mathcal{J}_{-j},$ resulting in an iteration
over a total of $2^{J-1}$ subsets of alternatives. The choice probability in a
LEVI model takes a simpler, more convenient form. While the choice probability
in a SEVI model is considerably richer than that in a LEVI model, the SEVI
model has a computational disadvantage when $J$ is large. To mitigate this, we
employ Gosper's hack, a technique rooted in bitwise operations, to streamline
the iteration process efficiently. Additionally, by leveraging the arguments
following Proposition 4.4 in
\citet{ross2013first}%
, it can be shown that for any $M<J$:
\begin{equation}
\left\vert P^{\mathrm{SEVI}}\left(  Y_{i}=j|v_{i}\right)  -P_{M}%
^{\mathrm{SEVI}}\left(  Y_{i}=j|v_{i}\right)  \right\vert \leq\sum
_{k_{1}<\ldots<k_{M}:\left\{  k_{1},\ldots,k_{M}\right\}  \subseteq
\mathcal{J}_{-j}}\frac{1}{1+\sum_{k\in\left\{  k_{1},\ldots,k_{\ell}\right\}
}\exp\left(  -v_{ik}+v_{ij}\right)  } \label{equ: upper bound}%
\end{equation}
where
\[
P_{M}^{\mathrm{SEVI}}\left(  Y_{i}=j|v_{i}\right)  =1+\sum_{\ell=1}%
^{M-1}\left(  -1\right)  ^{\ell}\sum_{k_{1}<\ldots<k_{\ell}:\left\{
k_{1},\ldots,k_{\ell}\right\}  \subseteq\mathcal{J}_{-j}}\frac{1}{1+\sum
_{k\in\left\{  k_{1},\ldots,k_{\ell}\right\}  }\exp\left(  -v_{ik}%
+v_{ij}\right)  }.
\]
This allows us to approximate $P^{\mathrm{SEVI}}\left(  Y_{i}=j|v_{i}\right)
$ by $P_{M}^{\mathrm{SEVI}}\left(  Y_{i}=j|v_{i}\right)  $ with a
well-controlled approximation error. Note that $P_{M}^{\mathrm{SEVI}}\left(
Y_{i}=j|v_{i}\right)  $ takes the same form as $P^{\mathrm{SEVI}}\left(
Y_{i}=j|v_{i}\right)  $ but involves subsets of at most $M-1$ alternatives,
and so it is computationally less intensive when $M$ is smaller than $J.$ To
achieve a desired level of precision, we can choose $M$ large enough and stop
iterating over subsets of more than $M-1$ alternatives. Combining Gosper's
hack with an early stopping makes it computationally easy to estimate a SEVI
model for $J$ as large as 14. Section \ref{Sec: time} reports and compares the
computation times for various models using simulated data with typical sample
size and attribute numbers. Section \ref{Sec:Emp} contains a comparison of
computation times for real data applications.\footnote{By leveraging multiple
processors running in parallel to enumerate choice subsets, the feasible size
of $J$ can be significantly increased. Unlike the sequential approach required
for integration in the multinomial probit model, this can be accomplished
simultaneously.}
\end{remark}

\begin{remark}
\label{various_SEVI_models}For the SEVI model with a linear parametrization of
$V_{ij}\left(  \beta_{0}\right)  ,$ if $Z_{i}$ is not present so that
$V_{ij}\left(  \beta_{0}\right)  =X_{ij}\beta_{0,x},$ then the model is called
a conditional SEVI (logit) model. As an extension of the basic model, we can
allow the coefficient on $X_{ij}$ to be alternative specific and so we may
have $V_{ij}\left(  \beta_{0}\right)  =X_{ij}\beta_{0,x,j}.$ On the other
hand, if $X_{ij}$ is not present so that $V_{ij}\left(  \beta_{0}\right)
=Z_{i}\beta_{0,z,j},$ then the model is called a multinomial SEVI (logit)
model. These forms can be combined so that $V_{ij}\left(  \beta_{0}\right)
=\alpha_{j}+Z_{i}\beta_{0,z,j}+X_{ij}\beta_{0,x}+W_{ij}\beta_{0,w,j},$ which
provides the same flexibility as in a standard LEVI (logit) model. As in a
standard multinomial LEVI (logit) model, not all of $\beta_{0,z,j}$'s are
identified, and we need to normalize one of them, say $\beta_{0,z,J},$ to be zero.
\end{remark}

To illustrate the difference between the SEVI and LEVI models, we simulate the
choice probabilities when $V_{ij}\left(  \beta_{0}\right)  =X_{ij}\beta_{0,x}$
(there is no individual-specific characteristic or alternative-specific
constant), and the stochastic utility is iid SEVI or LEVI. We set the number
of covariates $L$ in $X_{ij}$ to be $3$ (there are three attributes for each
alternative) and let $\beta_{0}=\beta_{0,x}=(1,2,1)^{\prime}.$ We consider two
data generating processes (DGP) for the attributes $\left\{  X_{ij}\right\}
$. In the first DGP, for each $i$ and $j,$ we let $X_{ij,\ell}$ be iid
$N(0,\pi^{2}\omega_{j}^{2}/36)$ over $\ell=1,2,3$ where
\begin{equation}
\omega_{j}=\frac{j-\left(  J+1\right)  /2}{\sqrt{\left(  J^{2}-1\right)  /12}%
},\text{ }j=1,\ldots,J.\text{ } \label{omega1}%
\end{equation}
Here, $\left(  \omega_{1},\ldots,\omega_{J}\right)  $ is a just standardized
version of the sequence ($1,2,\ldots,J$). In this case, each attribute
$X_{ij,\ell}$ becomes more dispersed when the alternative label $j$ is closer
to $1$ or $J$. As a result, the mean choice probabilities (i.e., $E\Pr\left(
Y_{i}=j|V_{i}\right)  $) are not all the same across the alternatives.\ In the
second DGP, for each $i$ and $j,$ we let $X_{ij,\ell}$ be iid $N(0,\pi
^{2}/36)$ over $\ell=1,2,3.$ In this case, the attributes have the same
distribution across the alternatives, and the mean choice probabilities are
the same across all alternatives. Under these two DGPs for $\left\{
X_{ij}\right\}  ,$ the variances of the systematic utility $V_{ij}\left(
\beta_{0}\right)  =X_{ij}\beta_{0}$ are $\pi^{2}\omega_{j}^{2}/6$ and $\pi
^{2}/6,$ respectively. These variances are comparable to $\pi^{2}/6,$ the
variance of the stochastic utility.%

\begin{figure}[h]%
\centering
\includegraphics[
height=3.9167in,
width=4.7859in
]%
{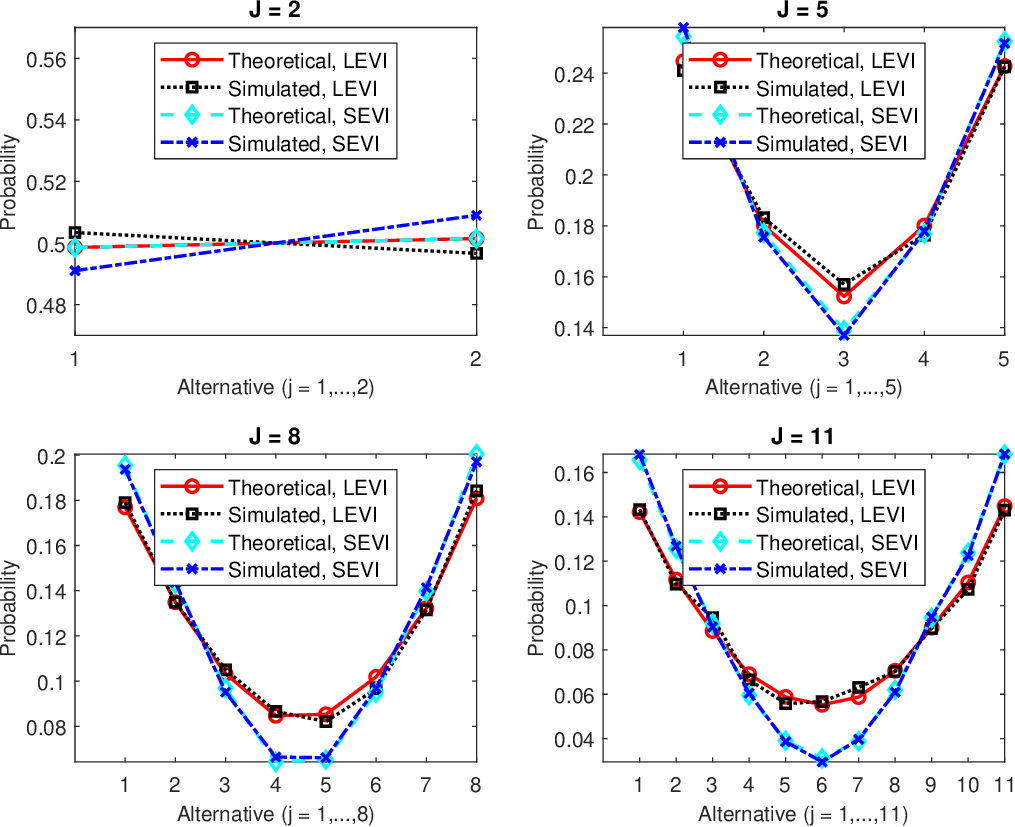}%
\caption{Simulated choice probabilities and the average of theoretical choice
probabilities against alternative labels $(j=1,\ldots,J)$ for a sample of
$n=10,000$ individuals when $X_{ij,\ell}$ is iid $N(0,\pi^{2}\omega_{j}%
^{2}/36)$}%
\label{Figure: Simu_theoretical_pro}%
\end{figure}

For each model, we compute the choice probabilities according to the
theoretical formulae (i.e., (\ref{SEVI_prob_formula}) under the SEVI and
(\ref{LEVI_prob_formula}) under the LEVI), and take the average of these
choice probabilities over 10,000 individuals. We also simulate the decisions
of 10,000 independent individuals and compute the proportion of individuals
choosing each alternative. Figure \ref{Figure: Simu_theoretical_pro} reports
results for various $J$ values when $X_{ij,\ell}$ $\thicksim$ iid $N(0,\pi
^{2}\omega_{j}^{2}/36).$ Not surprisingly, since the sample size is large, the
average theoretical choice probabilities match the simulated choice
proportions well for both the SEVI and LEVI models. When $J=2,$ the
theoretical choice probabilities are the same across the two models. With only
two alternatives, the simulated choice shares are very close across the two
models, but not identical due to simulation errors. When $J>2,$ for
alternatives with relatively higher choice probabilities under the LEVI, the
SEVI model assigns a higher choice probability than the LEVI model. On the
other hand, when $J>2,$ for alternatives with a relatively lower choice
probability under the LEVI, the SEVI model assigns a lower choice probability
than the LEVI model. The difference between the choice probabilities across
the two models grows with $J.$ The choice probability in a SEVI model is
closer to the \textquotedblleft arg max\textquotedblright\ function than that
in a LEVI model. In this context, a SEVI model leans more towards embodying
the concept of \textquotedblleft the winner takes it all\textquotedblright%
\ than a LEVI model. More specifically, given the same set of systematic
utilities, the probability of choosing the alternative with the maximum
systematic utility in a SEVI model tends to be larger than that in a LEVI model.

Next, we report the results when $X_{ij,\ell}$ is iid $N(0,\pi^{2}/36)$, but
we use a different type of plot. Figure
\ref{Figure: Simu_theoretical_pro_DIFF} provides a scatter plot of the ratio
of the choice probabilities\ in the SEVI and LEVI models (i.e.,
$P^{\mathrm{SEVI}}/P^{\mathrm{LEVI}})$ against $P^{\mathrm{LEVI}}$. Instead of
examining the average choice probabilities as done in Figure
\ref{Figure: Simu_theoretical_pro}, we consider the choice probability of
choosing a specific alternative (here, we use alternative 1 without loss of
generality) for each realization of the systematic utilities. Figure
\ref{Figure: Simu_theoretical_pro_DIFF2} provides a supplementary
illustration, offering a closer examination of certain details.%

\begin{figure}[h]%
\centering
\includegraphics[
height=4.0041in,
width=4.8905in
]%
{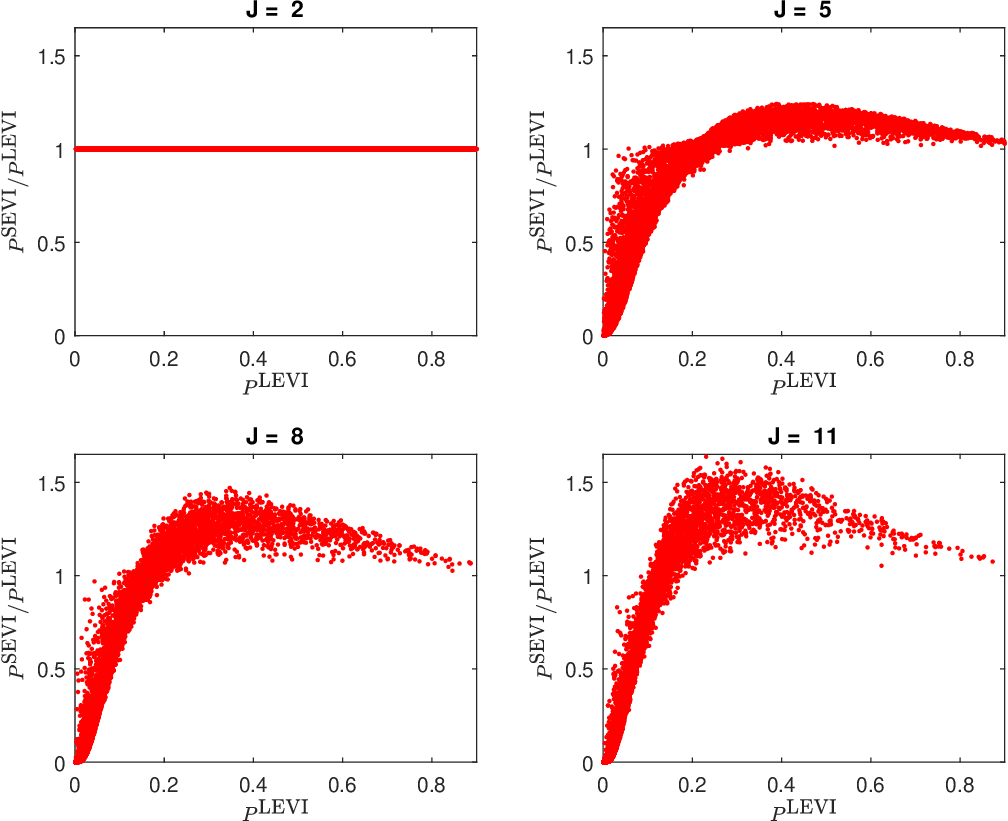}%
\caption{Scatter plot of the ratio of choice probabilities in the SEVI and
LEVI models ($P^{\mathrm{SEVI}}/P^{\mathrm{LEVI}})$ against the choice
probability in the LEVI model ($P^{\mathrm{LEVI}})$ with identical systematic
utilities when $X_{ij,\ell}$ is iid $N(0,\pi^{2}/36)$}%
\label{Figure: Simu_theoretical_pro_DIFF}%
\end{figure}
In both Figure \ref{Figure: Simu_theoretical_pro_DIFF} and Figure
\ref{Figure: Simu_theoretical_pro_DIFF2}, the systematic utilities are held
constant for each data point with either LEVI errors or SEVI errors. We note
from these two figures that when $J>2$ and $P^{\mathrm{LEVI}}$ is low,
$P^{\mathrm{SEVI}}$ can be more than 50 percent lower. For example, when there
are five alternatives ($J=5$) and $P^{\mathrm{LEVI}}=0.1,$ $P^{\mathrm{SEVI}}$
can be as low as $0.05$. The LEVI probability can be roughly twice as large as
the SEVI probability. Thus, from a practical point of view, a choice between
the two models can have an enormous impact on predicting a lift for a minor
brand with a market share between 5\% and 10\%. On the other hand, when $J>2$
and $P^{\mathrm{LEVI}}$ is high, $P^{\mathrm{SEVI}}$ tends to be higher. For
example, when $J=5$ and $P^{\mathrm{LEVI}}=0.4,$ $P^{\mathrm{SEVI}}$ can be 20
percent higher; when $J=11$ and $P^{\mathrm{LEVI}}=0.3,$ $P^{\mathrm{SEVI}}$
can be 50 percent higher.\ The patterns we observe here are similar to those
in Figure \ref{Figure: Example_with_5_options}. The difference in the choice
probabilities across the two models can be important in various contexts,
ranging from antitrust actions to production planning.
\begin{figure}[h]%
\centering
\includegraphics[
height=3.9799in,
width=4.9312in
]%
{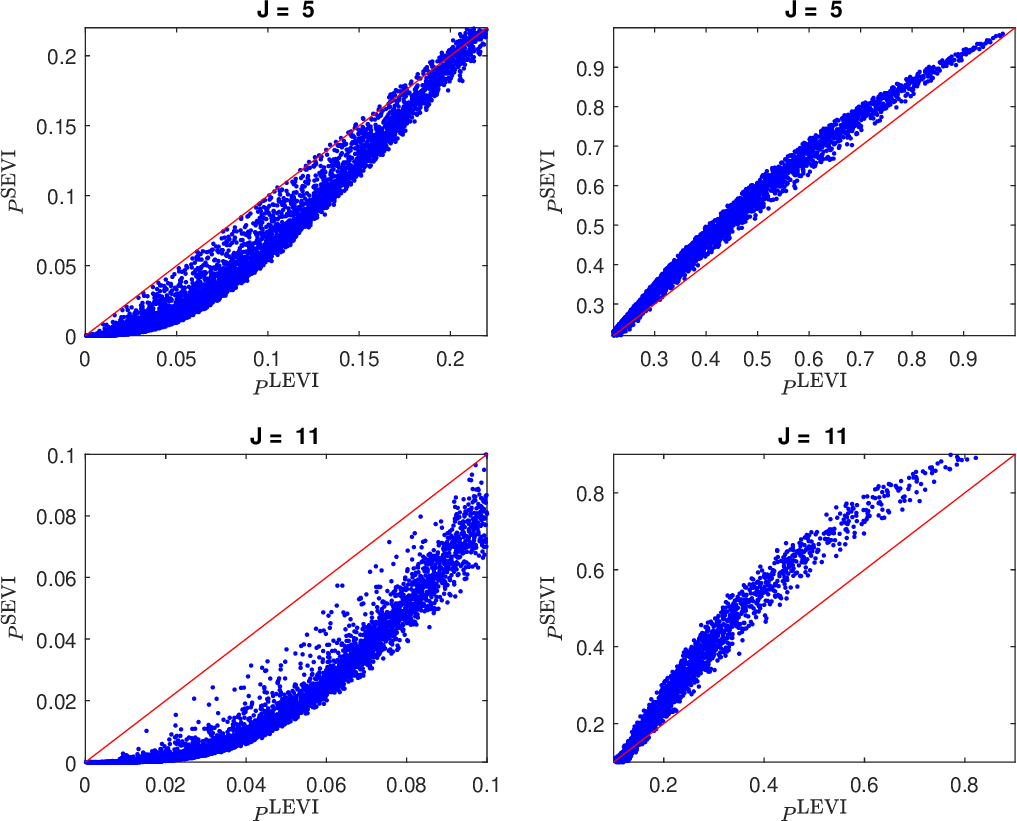}%
\caption{Scatter plot of the choice probability in the SEVI model versus that
in the LEVI model with left and right panels having different ranges for the
horizontal axis ($X_{ij,\ell}$ is iid $N(0,\pi^{2}/36))$}%
\label{Figure: Simu_theoretical_pro_DIFF2}%
\end{figure}

Figure \ref{Figure: Simu_theoretical_pro_DIFF2} shows that when $J=5$ and the
LEVI probability is below 0.20, the corresponding SEVI probability is lower.
When the LEVI probability is above 0.20, the corresponding SEVI probability is
higher. Similar patterns are found for other values of $J.$ These qualitative
patterns are consistent with what we find in Figure
\ref{Figure: Simu_theoretical_pro}. The vector of choice probabilities in the
SEVI model (i.e., $[P^{\mathrm{SEVI}}(Y=1),\ldots,P^{\mathrm{SEVI}}\left(
Y=J\right)  ]$) is closer to the vertices of the probability simplex than that
in the LEVI model. In other words, $[P^{\mathrm{SEVI}}(Y=1),\ldots
,P^{\mathrm{SEVI}}\left(  Y=J\right)  ]$ is closer to $\arg\max\left(
v_{1},\ldots,v_{J}\right)  :=(0,\ldots,1,\ldots,0),$ a $0$-$1$ vector with the
$j$-th element equal to 1 if and only if $v_{j}$ is the maximum value of
$\left(  v_{1},\ldots,v_{J}\right)  .$\footnote{For ease of discussion, we
assume that there is a unique maximum among $v_{1},\ldots,v_{J}.$ The 0-1
vector representation of $\arg\max$ is the so-called one-hot encoding of
$\arg\max.$}

We now offer insights into why the probability vector $[P^{\mathrm{SEVI}%
}(Y=1),\ldots,P^{\mathrm{SEVI}}\left(  Y=J\right)  ]$ is closer to the one-hot
encoding of $\arg\max\left(  v_{1},\ldots,v_{J}\right)  $ than
$[P^{\mathrm{LEVI}}(Y=1),\ldots,P^{\mathrm{LEVI}}\left(  Y=J\right)  ].$\ When
$J\geq3,$ the top two values among $J$ independent draws from the stochastic
utility distribution are more likely to be above or equal to the median of
this distribution rather than below it. For a draw from the LEVI distribution,
conditional on being above the median, the mean is 1.545 and the standard
deviation is 1.079. In contrast, for the SEVI case, these values are 0.391 and
0.500, respectively. Hence, the absolute difference between any two
above-the-median draws from the LEVI distribution is expected to be larger
than that for the SEVI distribution. In fact, Monte Carlo simulations show
that the former first-order stochastically dominates the latter (see Figure
\ref{Figure: cdf_diff_two_above_median_draws}). This implies that two
above-the-median draws from the SEVI distribution are more likely to be closer
to each other than those from the LEVI distribution. Therefore, with more than
two alternatives, it is less probable for the stochastic utility components
under the SEVI to make a difference in determining which alternative has the
highest total utility. As a result, in the presence of more than two
alternatives, the choice probabilities under the SEVI are better aligned with
the systematic utilities than those under the LEVI. In other words, in a SEVI
model with $J\geq3$, the observed choice is more likely to correspond to the
choice that maximizes the systematic utility than in a LEVI model. See Figure
\ref{figure; mean_differences_levi_norm_sev} for an additional illustration.

\subsection{Algebraic Relation between SEVI and LEVI Choice
Probabilities\label{Subsec: relation}}

To relate the choice probabilities under SEVI and LEVI assumptions, we
consider a general latent utility model $U_{ij}=V_{ij}+\varepsilon_{ij},$
$j=1,\ldots,J$ where $\varepsilon_{ij}$ is iid across $j=1,\ldots,J$ and may
follow either SEVI or LEVI, $\varepsilon_{i}=\left(  \varepsilon_{i1}%
,\ldots,\varepsilon_{iJ}\right)  $ is independent of $V_{i}=(V_{i1}%
,\ldots,V_{iJ})$. For any alternative $j$ and any subset of alternatives
$\mathcal{S}_{\ell}$ that does not include alternative $j,$ we let
\begin{align*}
P^{\mathrm{SEVI}}(Y_{i}  &  =j|\mathcal{S}_{\ell}\cup j,V_{i}=v_{i})=\Pr
(Y_{i}=j|\mathcal{S}_{\ell}\cup j,V_{i}=v_{i},\varepsilon_{im}\thicksim
iid\text{ }\mathrm{SEVI}\text{ over }m\in\mathcal{J}),\\
P^{\mathrm{LEVI}}(Y_{i}  &  =j|\mathcal{S}_{\ell}\cup j,V_{i}=v_{i})=\Pr
(Y_{i}=j|\mathcal{S}_{\ell}\cup j,V_{i}=v_{i},\varepsilon_{im}\thicksim
iid\text{ }\mathrm{LEVI}\text{ over }m\in\mathcal{J}),
\end{align*}
which are the probabilities of individual $i$ choosing alternative $j$ from
the subset $\mathcal{S}_{\ell}\cup j:=$ $\mathcal{S}_{\ell}\cup\left\{
j\right\}  $ under SEVI and LEVI errors, respectively. It is important to
point out that in the above definition, the probabilities depend on the choice
subset under consideration.

\begin{proposition}
\label{prop: mirror image}The choice probabilities under the SEVI and LEVI
satisfy the following relations:
\begin{align*}
P^{\mathrm{SEVI}}(Y_{i}  &  =j|\mathcal{J},V_{i}=v_{i})\\
&  =1+\sum_{\ell=1}^{J-1}\left(  -1\right)  ^{\ell}\sum_{k_{1}<\ldots<k_{\ell
}:\text{ }\mathcal{S}_{\ell}=\left\{  k_{1},\ldots,k_{\ell}\right\}
\subseteq\mathcal{J}_{-j}}P^{\mathrm{LEVI}}\left(  Y_{i}=j|\mathcal{S}_{\ell
}\cup j,V_{i}=-v_{i}\right)  ;
\end{align*}%
\begin{align*}
P^{\mathrm{LEVI}}(Y_{i}  &  =j|\mathcal{J},V_{i}=v_{i})\\
&  =1+\sum_{\ell=1}^{J-1}\left(  -1\right)  ^{\ell}\sum_{k_{1}<\ldots<k_{\ell
}:\text{ }\mathcal{S}_{\ell}=\left\{  k_{1},\ldots,k_{\ell}\right\}
\subseteq\mathcal{J}_{-j}}P^{\mathrm{SEVI}}\left(  Y_{i}=j|\mathcal{S}_{\ell
}\cup j,V_{i}=-v_{i}\right)  .
\end{align*}

\end{proposition}

The proposition shows an intrinsic connection between the choice probabilities
under the two assumptions of the stochastic utility. Note that the right-hand
side of each relation involves the probabilities of choosing alternative $j$
over all subsets of alternatives that contain alternative $j$ but under a
\textquotedblleft reflected\textquotedblright\ model for $V_{i}$ ($V_{i}%
=v_{i}$ on the left-hand side becomes $V_{i}=-v_{i}$ on the right-hand side).

By taking the difference between the two equations in Proposition
\ref{prop: mirror image}, we have%
\begin{align*}
P^{\mathrm{SEVI}}(Y_{i}  &  =j|\mathcal{J},V_{i}=v_{i})-P^{\mathrm{LEVI}%
}(Y_{i}=j|\mathcal{J},V_{i}=v_{i})\\
&  =\sum_{\ell=2}^{J-1}\left(  -1\right)  ^{\ell-1}\sum_{k_{1}<\ldots<k_{\ell
}:\text{ }\mathcal{S}_{\ell}=\left\{  k_{1},\ldots,k_{\ell}\right\}
\subseteq\mathcal{J}_{-j}}\left[  P^{\Delta}\left(  Y_{i}=j|\mathcal{S}_{\ell
}\cup j,V_{i}=-v_{i}\right)  \right]  ,
\end{align*}
where
\[
P^{\Delta}\left(  Y_{i}=j|\mathcal{S}_{\ell}\cup j,V_{i}=-v_{i}\right)
=P^{\mathrm{SEVI}}\left(  Y_{i}=j|\mathcal{S}_{\ell}\cup j,V_{i}%
=-v_{i}\right)  -P^{\mathrm{LEVI}}\left(  Y_{i}=j|\mathcal{S}_{\ell}\cup
j,V_{i}=-v_{i}\right)  .
\]
So, the difference in the SEVI and LEVI probabilities of choosing alternative
$j$ from the entire choice set $\mathcal{J}$ is an aggregation of the
probability differences from choosing alternative $j$ from all subsets of
alternatives that contain $j$ but with flipped systematic utilities. This is a
recursive relationship, as each probability difference on the right-hand side
can be represented in the same fashion as the left-hand side. An intriguing
question emerges: can the difference in choice probabilities be fully captured
by a simple index that measures the concentration of the systematic utilities?
Given the intricate recursive relationship described above and the highly
nonlinear dependence of the probability difference on systematic utilities, as
shown in Figure \ref{Figure: contourDIFF}, it appears unlikely that commonly
used measures like the Gini coefficient or the Herfindahl-Hirschman index
would be sufficient for this task. However, in empirical practice,
particularly for moderately sized $J$, these measures are likely to prove
useful in predicting differences in choice probabilities.

The second relation in Proposition \ref{prop: mirror image} can be presented
equivalently as follows:
\begin{align}
P^{\mathrm{LEVI}}(Y_{i}  &  =j|\mathcal{J},V_{i}=-v_{i})\nonumber\\
&  =1+\sum_{\ell=1}^{J-1}\left(  -1\right)  ^{\ell}\sum_{k_{1}<\ldots<k_{\ell
}:\mathcal{S}_{\ell}:=\left\{  k_{1},\ldots,k_{\ell}\right\}  \subseteq
\mathcal{J}_{-j}}P^{\mathrm{SEVI}}\left(  Y_{i}=j|\mathcal{S}_{\ell}\cup
j,V_{i}=v_{i}\right)  . \label{Block_Marschak1}%
\end{align}
Since $P^{\mathrm{LEVI}}(Y_{i}=j|\mathcal{J},V_{i}=-v_{i})$ is nonnegative,
the sum in (\ref{Block_Marschak1}) is also nonnegative. This sum is a special
Block-Marschak polynomial (%
\citet{BlockMarschak1959}%
). In the presence of a finite number of alternatives,\ it is well-known in
the stochastic choice theory that the nonnegativity of Block-Marschak
polynomials is both necessary and sufficient for a random choice rule to have
an additive random utility representation. For a detailed discussion, see, for
example, Chapters 1 and 2 of the recent book by\
\citet{Strzalecki2023}%
. In the present setting, the choice probabilities under SEVI are derived from
random utility maximization and as such, they trivially have an additive
random utility representation. The significance of (\ref{Block_Marschak1}) is
that it provides a simple characterization of a special Block-Marschak polynomial.

More generally, let $j$ be a target alternative of interest and
\[
\mathcal{D}_{j}:=\left\{  \left(  d_{1},\ldots,d_{\rho},j\right)
:d_{1}<\ldots<d_{\rho}\text{ and }\left(  d_{1}<\ldots<d_{\rho}\right)
\subseteq\mathcal{J}_{-j}\right\}
\]
be a subset of alternatives that \emph{includes} alternative $j.$ Then, the
(general) Block-Marschak polynomial for our probability choice rule under SEVI
errors is defined to be
\begin{align*}
&  \mathcal{K}^{\mathrm{SEVI}}(j,\mathcal{D}_{j})\\
&  :=\Pr\left(  Y_{i}=j|\mathcal{D}_{j},V_{i}=v_{i},\varepsilon_{im}\thicksim
iid\text{ }\mathrm{SEVI}\text{ over }m\in\mathcal{J}\right) \\
&  +\sum_{\ell=1}^{J-\left\vert \mathcal{D}_{j}\right\vert }\left(  -1\right)
^{\ell}\sum_{k_{1}<\ldots<k_{\ell}:\mathcal{S}_{\ell}:=\left\{  k_{1}%
,\ldots,k_{\ell}\right\}  \subseteq\mathcal{J}\backslash\mathcal{D}_{j}%
}P^{\mathrm{SEVI}}\left(  Y_{i}=j|\mathcal{S}_{\ell}\cup\mathcal{D}_{j}%
,V_{i}=v_{i}\right)  .
\end{align*}
Note that when $\mathcal{D}_{j}=\left\{  j\right\}  ,$ the first term in
$\mathcal{K}^{\mathrm{SEVI}}(j,\mathcal{D}_{j})$ becomes $1$ and
$\mathcal{K}^{\mathrm{SEVI}}(j,\mathcal{D}_{j})$ reduces to the special
Block-Marschak polynomial in (\ref{Block_Marschak1}). We can show that
$\mathcal{K}^{\mathrm{SEVI}}(j,\mathcal{D}_{j})$ is the probability that the
$j$-th alternative dominates all other alternatives in $\mathcal{D}_{j}$ but
is dominated by the alternatives that are not in $\mathcal{D}_{j}.$ See, for
example,
\citet{BarbaraPattanaik1986}%
. Hence, $\mathcal{K}^{\mathrm{SEVI}}(j,\mathcal{D}_{j})\geq0$ for all
$\left(  j,\mathcal{D}_{j}\right)  ,$ which is consistent with the
nonnegativity of Block-Marschak polynomials for any additive RUM model.

The Block-Marschak polynomial $\mathcal{K}^{\mathrm{SEVI}}(j,\mathcal{D}_{j})$
is precisely the probability relevant for rank-ordered data. While using the
SEVI model for rank-ordered data necessitates additional research, the formula
for $\mathcal{K}^{\mathrm{SEVI}}(j,\mathcal{D}_{j})$ provides a good starting point.

\subsection{Patterns of Substitution and the Lack of IIA\label{Subsec: IIA}}

In this subsection, we compute $\partial P^{\mathrm{SEVI}}\left(
Y_{i}=k|v_{i}\right)  /\partial v_{ij}$, the effect of an infinitesimal change
in alternative $j$'s systematic utility (as perceived by individual $i$) on
the probability of individual $i$ choosing alternative $k.$ This computation
is essential for determining $\partial P^{\mathrm{SEVI}}\left(  Y_{i}%
=k|v_{i}\right)  /\partial X_{ij},$ the effect of changing alternative $j$'s
attributes on the probability of choosing alternative $k.$

For notational economy, we let $P_{ik}^{\mathrm{SEVI}}:=P^{\mathrm{SEVI}%
}\left(  Y_{i}=k|v_{i}\right)  .$ We have, using Proposition
\ref{prop: mirror image} and\ some simple calculations,%
\begin{equation}
\frac{\partial P_{ij}^{\mathrm{SEVI}}}{\partial v_{ij}}=-\sum_{\ell=1}%
^{J-1}\left(  -1\right)  ^{\ell}\sum_{k_{1}<\ldots<k_{\ell}:\mathcal{S}_{\ell
}=\left\{  k_{1},\ldots,k_{\ell}\right\}  \subseteq\mathcal{J}_{-j}%
}P^{\mathrm{LEVI}}(Y_{i}=j|\mathcal{S}_{\ell}\cup j,-v_{i})P^{\mathrm{LEVI}%
}(Y_{i}\neq j|\mathcal{S}_{\ell}\cup j,-v_{i}), \label{substitution1}%
\end{equation}
where for $\mathcal{S}_{\ell-1}=\left\{  k_{1},\ldots,k_{\ell-1}\right\}
\subseteq\mathcal{J}_{-j},$
\begin{align*}
P^{\mathrm{LEVI}}(Y_{i}  &  =j|\mathcal{S}_{\ell}\cup j,-v_{i})=\frac
{\exp(-v_{ij})}{\sum_{m\in\mathcal{S}_{\ell}\cup j}\exp\left(  -v_{im}\right)
},\\
P^{\mathrm{LEVI}}(Y_{i}  &  \neq j|\mathcal{S}_{\ell}\cup j,-v_{i}%
)=1-P^{\mathrm{LEVI}}(Y_{i}=j|\mathcal{S}_{\ell}\cup j,-v_{i}).
\end{align*}
Furthermore, for $k\neq j,$
\begin{align}
\frac{\partial P_{ij}^{\mathrm{SEVI}}}{\partial v_{ik}}  &  =\sum_{\ell
=1}^{J-1}\left(  -1\right)  ^{\ell}\sum_{k_{1}<\ldots<k_{\ell-1}%
:\mathcal{S}_{\ell-1}=\left\{  k_{1},\ldots,k_{\ell-1}\right\}  \subseteq
\mathcal{J}_{-\left(  j,k\right)  }}\nonumber\\
&  P^{\mathrm{LEVI}}(Y_{i}\overset{}{=}j|\mathcal{S}_{\ell-1}\cup\left\{
j,k\right\}  ,-v_{i})\cdot P^{\mathrm{LEVI}}(Y_{i}\overset{}{=}k|\mathcal{S}%
_{\ell-1}\cup\left\{  j,k\right\}  ,-v_{i}), \label{substitution2}%
\end{align}
where for $\mathcal{S}_{\ell-1}=\left\{  k_{1},\ldots,k_{\ell-1}\right\}
\subseteq\mathcal{J}_{-\left(  j,k\right)  },$ a subset of $\mathcal{J}$
excluding $j$ and $k,$
\begin{align*}
P^{\mathrm{LEVI}}(Y_{i}  &  =j|\mathcal{S}_{\ell}\cup\left\{  j,k\right\}
,-v_{i})=\frac{\exp(-v_{ij})}{\sum_{m\in\mathcal{S}_{\ell}\cup\left\{
j,k\right\}  }\exp\left(  -v_{im}\right)  },\\
P^{\mathrm{LEVI}}(Y_{i}  &  =k|\mathcal{S}_{\ell}\cup\left\{  j,k\right\}
,-v_{i})=\frac{\exp(-v_{ik})}{\sum_{m\in\mathcal{S}_{\ell}\cup\left\{
j,k\right\}  }\exp\left(  -v_{im}\right)  }.
\end{align*}
In equation (\ref{substitution2}), the first term in the summation (i.e., when
$\ell=1)$ is understood to be
\[
-\frac{\exp(-v_{ij})}{\left[  \exp(-v_{ij})+\exp\left(  -v_{ik}\right)
\right]  }\frac{\exp(-v_{ik})}{\left[  \exp(-v_{ij})+\exp\left(
-v_{ik}\right)  \right]  }.
\]

Equation (\ref{substitution2}) reveals that
\begin{equation}
\frac{\partial P_{ij}^{\mathrm{SEVI}}}{\partial v_{ik}}=\frac{\partial
P_{ik}^{\mathrm{SEVI}}}{\partial v_{ij}}%
\end{equation}
for any $k$ and $j.$ That is, the choice probabilities satisfy a symmetry
property analogous to Slutsky symmetry in continuous choice models.

The expression of $\partial P_{ij}^{\mathrm{SEVI}}/\partial v_{ik}$ can be
used to derive the score function, which is essential for finding the MLE and
developing the asymptotic theory. It can also be employed to examine whether
the IIA property holds in a SEVI model when $J\geq3$. For the IIA to hold, the
cross (semi)elasticity $\partial\log P_{ij}^{\mathrm{SEVI}}/\partial v_{ik}$
should be the same for all $j\neq k.$ That is,
\[
\frac{\partial\log P_{ij_{1}}^{\mathrm{SEVI}}}{\partial v_{ik}}=\frac
{\partial\log P_{ij_{2}}^{\mathrm{SEVI}}}{\partial v_{ik}}%
\]
for all $j_{1}\neq j_{2}\neq k.$ However, this condition is not met, as is
clear from (\ref{substitution2}). Alternatively, for the IIA\ to hold, the
probability ratio $P_{ij_{1}}^{\mathrm{SEVI}}/P_{ij_{2}}^{\mathrm{SEVI}}$
should not depend on the systematic utilities of alternatives other than
alternatives $j_{1}$ and $j_{2}.$ However, it is evident that the probability
ratio \emph{does} depend on the systematic utilities of all alternatives,
indicating that the IIA does not hold.

To illustrate the violation of the IIA in a SEVI model with $J\geq3$, we start
with a scenario featuring two initial alternatives: individuals must choose
between alternative 1 and alternative 2. In this context, the choice between
the SEVI or LEVI does not matter. For a fixed scale value $s_{1}$ (we set
$v_{1}$ to $\log1$ so that $s_{1}=\exp(v_{1})=1$)$,$\footnote{Following
\citet{YELLOTT1977}%
), we refer to the exponentiated systematic utility for alternative $j$
$\left(  \text{i.e.},\text{ }\exp\left(  v_{j}\right)  \right)  $ as the scale
value ($s_{j}$) for this alternative.} we examine a range of possible values
for $s_{2}:=\exp\left(  v_{2}\right)  \in\left\{  0.01:0.01:3\right\}  ,$ a
grid of values from 0.01 to 3 with increment 0.01. Each combination of
$\left(  s_{1},s_{2}\right)  $ corresponds to a probability ratio:
\[
\frac{P(Y=2|v_{1},v_{2})}{P(Y=1|v_{1},v_{2})}=\frac{P^{\mathrm{SEVI}%
}(Y=2|v_{1},v_{2})}{P^{\mathrm{SEVI}}(Y=1|v_{1},v_{2})}=\frac{P^{\mathrm{LEVI}%
}(Y=2|v_{1},v_{2})}{P^{\mathrm{LEVI}}(Y=1|v_{1},v_{2})}=\frac{\exp(v_{2}%
)}{\exp(v_{1})}=s_{2}.
\]
Next, we introduce a third alternative with its scale value $s_{3}=\exp\left(
v_{3}\right)  \in\left\{  0.01:0.01:3\right\}  $, the same grid for $s_{2}.$
We investigate how the probability ratio of choosing alternative 2 versus
alternative 1 changes. This ratio now depends on whether the stochastic
utilities follow the SEVI or LEVI distribution.

The difference of the ratios under SEVI and LEVI models, denoted as
$D_{\mathrm{PR}},$ is
\begin{align*}
D_{\mathrm{PR}}  &  =\frac{P^{\mathrm{SEVI}}(Y=2|v_{1},v_{2},v_{3}%
)}{P^{\mathrm{SEVI}}(Y=1|v_{1},v_{2},v_{3})}-\frac{P^{\mathrm{LEVI}}%
(Y=2|v_{1},v_{2},v_{3})}{P^{\mathrm{LEVI}}(Y=1|v_{1},v_{2},v_{3})}\\
&  =\frac{P^{\mathrm{SEVI}}(Y=2|v_{1},v_{2},v_{3})}{P^{\mathrm{SEVI}%
}(Y=1|v_{1},v_{2},v_{3})}-\frac{P^{\mathrm{LEVI}}(Y=2|v_{1},v_{2}%
)}{P^{\mathrm{LEVI}}(Y=1|v_{1},v_{2})}\\
&  =\frac{P^{\mathrm{SEVI}}(Y=2|v_{1},v_{2},v_{3})}{P^{\mathrm{SEVI}%
}(Y=1|v_{1},v_{2},v_{3})}-\frac{P^{\mathrm{SEVI}}(Y=2|v_{1},v_{2}%
)}{P^{\mathrm{SEVI}}(Y=1|v_{1},v_{2})}\\
&  =\frac{P^{\mathrm{SEVI}}(Y=2|v_{1},v_{2},v_{3})}{P^{\mathrm{SEVI}%
}(Y=1|v_{1},v_{2},v_{3})}-s_{2}.
\end{align*}
Here we have used properties of the LEVI probabilities. By the third equality
above, $D_{\mathrm{PR}}$ also gauges the extent of IIA violation. With $s_{1}$
set to 1, $D_{\mathrm{PR}}$ is a function of $s_{2}$ and $s_{3}$, so we can
write it as $D_{\mathrm{PR}}(s_{2},s_{3}).$ We present a contour plot of this
function in Figure \ref{Figure: contour_plot_IIA}. To aid interpretation, we
label the horizontal axis as $P(Y=2|v_{1},v_{2})/P(Y=1|v_{1},v_{2})=\exp
(v_{2})=s_{2},$ representing the initial probability ratio with only two
alternatives available.

It is clear from Figure \ref{Figure: contour_plot_IIA} that $D_{\mathrm{PR}%
}(s_{2},s_{3})$ is not a zero function. After the introduction of the third
alternative, the market shares of alternatives 1 and 2 do not decrease
proportionally in the SEVI\ model as what happens in the LEVI model,
indicating a violation of the IIA.

Figure \ref{Figure: contour_plot_IIA} also shows that there are two scenarios
where $D_{\mathrm{PR}}(s_{2},s_{3})$ equals zero or becomes approximately
zero. First, $D_{\mathrm{PR}}(s_{2},s_{3})$ equals zero when alternatives 1
and 2 have the same initial market share (i.e., $P(Y=2|v_{1},v_{2}%
)/P(Y=1|v_{1},v_{2})=s_{2}=1)$. This is expected, as in this case, the first
two alternatives are indistinguishable, and their market shares reduce
proportionally in both the LEVI and SEVI models when the third alternative is
introduced. Second, $D_{\mathrm{PR}}(s_{2},s_{3})$ approaches zero when one of
the alternatives becomes almost completely dominant in the market. This
corresponds to a scenario where either $s_{2}$ is small and $s_{3}$ is very
large, or $s_{3}$ is small and $s_{2}$ is very large. Hence, the difference
between a SEVI model and a LEVI model may not be large when there is clearly a
dominant player in the market who commands almost the entire market, and their
disparities become more apparent when no alternative has nearly all market share.

\begin{figure}[h]
\centering
\includegraphics[
height=2.739in,
width=3.694in
]{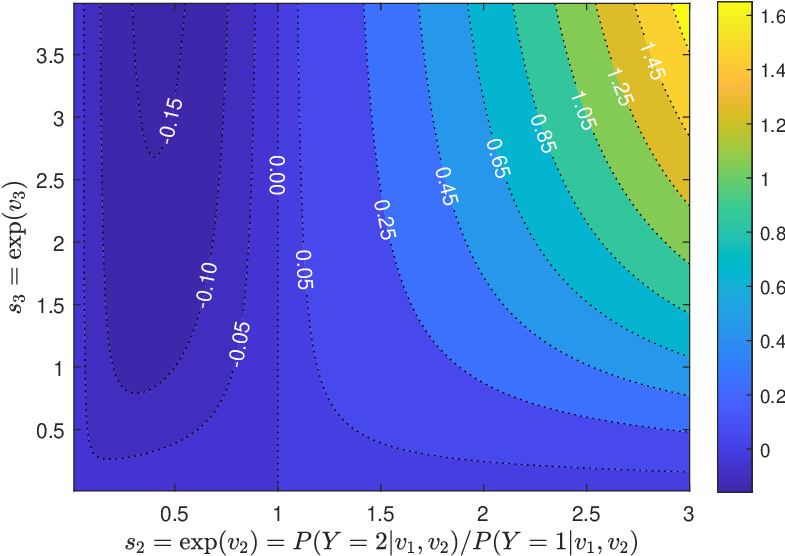}\caption{Contour plot of the difference of the choice
probability ratios $D_{\mathrm{PR}}(s_{2},s_{3})$ between the SEVI and LEVI
models against $\left[  s_{2},s_{3}\right]  $ when $s_{1}=1$}%
\label{Figure: contour_plot_IIA}%
\end{figure}

In an additively separable RUM with $U_{ij}=V_{ij}+\varepsilon_{ij},$ as
considered here, it is well known that the IIA property holds for any $J\geq3$
if and only if the stochastic utility $\varepsilon_{ij}$ follows the LEVI
distribution (cf.,
\citet{YELLOTT1977}%
). Under certain conditions that permit an infinite $J,$ it can be shown that
in a random utility model with possibly nonseparable utility functions, the
IIA holds for any $J\geq3$ only if the latent utility $U_{ij}$ can be
represented as $H\left(  V_{ij}+\varepsilon_{ij}\right)  $ for some strictly
increasing function $H\left(  \cdot\right)  $ and LEVI-distributed
$\varepsilon_{ij}.$ This result is a corollary to Theorem 1 of
\citet{Dagsvik2016}%
; see also
\citet{Lindberg}
and the references therein.\footnote{When $J$ is finite, the latent utility
presentation of $H\left(  V_{ij}+\varepsilon_{ij}\right)  $ is not necessary
for IIA to hold, but the counterexample given in
\citet{Lindberg}
is unlikely to be relevant to most empirical applications.} The SEVI model
does not have the required representation, and this provides another
perspective on the absence of the IIA property in the SEVI model when
$J\geq3.$

One consequence of the lack of IIA is that we cannot consistently estimate the
preference parameters in a SEVI model using the random sampling approach. In
this approach, for each decision maker, we sample a random subset of
alternatives that contains the chosen alternative and treat this subset as if
it were the choice set faced by this decision maker. In the simplest case with
three total alternatives $(J=3)$, the reason for the inconsistency is that%
\[
\frac{P^{\mathrm{SEVI}}(Y=1|v_{1},v_{2},v_{3})}{P^{\mathrm{SEVI}}%
(Y=1|v_{1},v_{2},v_{3})+P^{\mathrm{SEVI}}(Y=2|v_{1},v_{2},v_{3})}\neq
P^{\mathrm{SEVI}}(Y=1|v_{1},v_{2}).
\]
This contrasts with a LEVI-based logit, where such a random sampling procedure
can be used to reduce the number of alternatives without inducing
inconsistency. For a detailed discussion of the LEVI case, see page 64 of
\citet{book_train_2009}%
.

We note in passing that IIA also does not hold for standard multinomial probit
models, i.e., the multinomial probit with iid normal random utility
components. See Section \ref{Subsec: MNP_IIA} in the supplementary appendix
for details.

\subsection{The Surplus Function: Compensating Variation and Identification}

A fundamental concept in the discrete choice analysis is the surplus function,
defined as:
\[
W(v)=E\left[  \max_{j\in\mathcal{J}}\left\{  v_{j}+\varepsilon_{j}\right\}
\right]  -E\max_{j\in\mathcal{J}}\left\{  \varepsilon_{j}\right\}  ,
\]
where the expectations are taken with respect to the distribution of $\left\{
\varepsilon_{1},\ldots,\varepsilon_{J}\right\}  .$ In the above, the first
term is the expected value of the maximum utility achieved, given the vector
of the systematic utilities $v=\left(  v_{1},\ldots,v_{J}\right)  ^{\prime}$
and the second term serves as a constant benchmark.\footnote{If these terms
are not well defined, $W(v)$ can always be expressed as $E\left[  \max
_{j\in\mathcal{J}}\left\{  v_{j}+a_{j}\right\}  -\max_{j\in\mathcal{J}%
}\left\{  a_{j}\right\}  \right]  ,$ which is well defined. \ } See, for
example,
\citet{McFadden1981}%
. Among other uses, the surplus function can be used to obtain the choice
probabilities. More specifically, the Williams-Daly-Zachary
theorem\footnote{It was termed the WDZ theorem in Section 5.8 of
\citet{McFadden1981}. McFadden refers to \citet{Daly1978} and
\citet{Williams1977} as forerunners of the result.} says that
\[
P_{ij}=\frac{\partial W(v_{i1},\ldots,v_{iJ})}{\partial v_{ij}}.
\]

The specific form of the surplus function $W(v)$ depends on the distribution
of $\varepsilon_{1},\ldots,\varepsilon_{J}.$ When $\varepsilon_{j}\thicksim
iid$ LEVI, it is well-known that
\[
W^{\mathrm{LEVI}}(v)=\ln\left[  \frac{1}{J}\sum_{j=1}^{J}\exp\left(
v_{j}\right)  \right]  ,
\]
taking the form of a \textquotedblleft log-sum\textquotedblright. Here the
superscript \textquotedblleft$\mathrm{LEVI}$\textquotedblright\ signifies that
the surplus function is based on the assumption that $\varepsilon_{j}\thicksim
iid$ LEVI.

The next proposition gives the form of the surplus function when
$\varepsilon_{j}\thicksim iid$ SEVI.

\begin{proposition}
\label{Prop: surplus}If $\varepsilon_{j}\thicksim iid$ $\mathrm{SEVI}$ across
$j=1,\ldots,J,$ then the surplus function, denoted by $W^{\mathrm{SEVI}}(v),$
is
\begin{align}
W^{\mathrm{SEVI}}(v)  &  =\sum_{\ell=1}^{J}\left(  -1\right)  ^{\ell}%
\sum_{k_{1}<k_{2}<\ldots<k_{\ell}:\left\{  k_{1},\ldots,k_{\ell}\right\}
\subseteq\mathcal{J}}\ln\left[  \frac{1}{\ell}\sum_{k\in\left\{  k_{1}%
,\ldots,k_{\ell}\right\}  }\exp\left(  -v_{k}\right)  \right] \nonumber\\
&  =\sum_{\ell=1}^{J}\left(  -1\right)  ^{\ell}\sum_{k_{1}<k_{2}%
<\ldots<k_{\ell}:\left\{  k_{1},\ldots,k_{\ell}\right\}  \subseteq\mathcal{J}%
}W^{\mathrm{LEVI}}(-\left(  v_{k_{1}},\ldots,v_{k_{\ell}}\right)  ),
\label{W_LEVI_negative_v}%
\end{align}
where
\[
W^{\mathrm{LEVI}}(-\left(  v_{k_{1}},\ldots,v_{k_{\ell}}\right)  )=\ln\left[
\frac{1}{\ell}\sum_{k\in\left\{  k_{1},\ldots,k_{\ell}\right\}  }\exp\left(
-v_{k}\right)  \right]  .
\]
Furthermore, $W^{\mathrm{SEVI}}(v)$ is strictly convex.
\end{proposition}

Although the surplus function still contains \textquotedblleft
log-sum\textquotedblright, $W^{\mathrm{SEVI}}(v)$ is more complex than
$W^{\mathrm{LEVI}}(v)$. Nevertheless, Proposition \ref{Prop: surplus} gives us
a closed-form expression for $W^{\mathrm{SEVI}}(v)$ and shows that it is
strictly convex. The second representation given in (\ref{W_LEVI_negative_v})
shows how $W^{\mathrm{SEVI}}(v)$ is related to the log-sum over all subsets of
choice alternatives but with the systematic utilities flipped.

As the first application of Proposition \ref{Prop: surplus}, consider a SEVI
model of the form:%
\begin{equation}
U_{ij}=V_{ij}+\varepsilon_{ij}\text{, }V_{ij}=\tilde{V}_{ij}+\left(
\mathcal{I}_{i}-K_{j}\right)  \lambda, \label{SEVI_model_price}%
\end{equation}
where $\varepsilon_{ij}\thicksim iid$ $\mathrm{SEVI}$ for $j=1,2,\ldots,J.$
Here we have introduced an additional covariate $\mathcal{I}_{i}-K_{j}$ as
part of the systematic utility, where $\mathcal{I}_{i}$ is the income of
individual $i$ and $K_{j}$ is the cost/price of alternative $j.$ The
coefficient $\lambda$ on this covariate represents the marginal utility of
income. In this model, the systematic utility that individual $i$ obtains from
alternative $j$ consists of two components: the utility $\left(
\mathcal{I}_{i}-K_{j}\right)  \lambda$ derived from the expenditure on the
numeraire and the utility $\tilde{V}_{ij}$ derived from the characteristics of
alternative $j$. Such a formulation is common in applied welfare analysis
using discrete choice models.\footnote{In this formulation, income does not
influence the choice probabilities. This assumption may be reasonable when
costs associated with the alternatives are negligible compared to income or
where measurement errors in income outweigh the costs. While it is possible to
account for income effects, we do not explore this extension here.}

Now suppose we change the cost/price for the $m$-th alternative from $K_{m}$
to $K_{m}+\Delta_{m}$ while keeping all else the same as before. For each
individual $i,$ the compensating variation $\mathrm{CV}$ is defined implicitly
via the equation below:
\begin{align*}
&  \max_{j\in\mathcal{J}}\left\{  \tilde{V}_{ij}+\left[  \mathcal{I}%
_{i}+\mathrm{CV}-\left(  K_{j}+1\left\{  m=j\right\}  \Delta_{m}\right)
\right]  \lambda+\varepsilon_{ij}\right\} \\
&  =\max_{j\in\mathcal{J}}\left\{  \tilde{V}_{ij}+\left(  \mathcal{I}%
_{i}-K_{j}\right)  \lambda+\varepsilon_{ij}\right\}  ,
\end{align*}
where $1\left\{  \cdot\right\}  $ is the indicator function. That is,
$\mathrm{CV}$ is the additional income that individual $i$ would need to
maintain the same utility when the cost of alternative $m$ changes from
$K_{m}$ to $K_{m}+\Delta_{m}.$ Solving for $\mathrm{CV}$ yields
\[
\mathrm{CV}=\frac{1}{\lambda}\left[  \max\left\{  \tilde{V}_{ij}+\left(
\mathcal{I}_{i}-K_{j}\right)  \lambda+\varepsilon_{ij}\right\}  -\max
_{j}\left\{  \tilde{V}_{ij}+\left[  \mathcal{I}_{i}-\left(  K_{j}+1\left\{
m=j\right\}  \Delta_{m}\right)  \right]  \lambda+\varepsilon_{ij}\right\}
\right]  .
\]
\ Using Proposition \ref{Prop: surplus}, we can find that the expected value
of $\mathrm{CV}$ conditional on $V$ and $\Delta_{m}$ is%
\begin{equation}
E\left(  \mathrm{CV}|V,\Delta_{m}\right)  =\frac{1}{\lambda}\left[
W^{\mathrm{SEVI}}(V-e_{m}\Delta_{m}\lambda)-W^{\mathrm{SEVI}}(V)\right]  ,
\label{ECV_price}%
\end{equation}
where $e_{m}$ is the unit vector with a value of 1 at the $m$-th element and
$0$ elsewhere. We can average $E\left(  \mathrm{CV}|V,\Delta_{m}\right)  $
over the distribution of $V$ to get the expected compensating variation over
the population:
\[
E\left(  \mathrm{CV}|\Delta_{m}\right)  :=E\left[  E\left(  \mathrm{CV}%
|V,\Delta_{m}\right)  |\Delta_{m}\right]  .
\]
\ In empirical applications, we can\ estimate $E\left(  \mathrm{CV}%
|V,\Delta_{m}\right)  $ and $E\left(  \mathrm{CV}|\Delta_{m}\right)  $ using
sample averages after plugging in parameter estimates.

As a second and closely related application, we can employ the surplus
function to define the compensating variation when an alternative is removed
from the choice set. Using the same model as in (\ref{SEVI_model_price}), we
investigate the minimum compensation necessary for individual $i$ so that they
would not be worse off if alternative $k$ were eliminated from their choice
set. Denote this compensation as $\mathrm{CV}\left(  k\right)  $. Then,
$\mathrm{CV}\left(  k\right)  $ solves%
\[
\max_{j\in\mathcal{J}\backslash k}\left\{  \tilde{V}_{ij}+\left[
\mathcal{I}_{i}+\mathrm{CV}\left(  k\right)  -K_{j}\right]  \lambda
+\varepsilon_{ij}\right\}  =\max_{j\in\mathcal{J}}\left\{  \tilde{V}%
_{ij}+\left(  \mathcal{I}_{i}-K_{j}\right)  \lambda+\varepsilon_{ij}\right\}
.
\]
By a similar argument used to derive (\ref{ECV_price}), we find
\begin{align*}
\mathrm{CV}\left(  k\right)   &  =\frac{1}{\lambda}\left[  \max_{j\in
\mathcal{J}}\left\{  \tilde{V}_{ij}+\left(  \mathcal{I}_{i}-K_{j}\right)
\lambda+\varepsilon_{ij}\right\}  -\max_{j\in\mathcal{J}\backslash k}\left\{
\tilde{V}_{ij}+\left[  \mathcal{I}_{i}-K_{j}\right]  \lambda+\varepsilon
_{ij}\right\}  \right] \\
&  =\frac{1}{\lambda}\left[  W^{\mathrm{SEVI}}(V)-W^{\mathrm{SEVI}}%
(V_{-k})+W^{\mathrm{SEVI}}(0_{J})-W^{\mathrm{SEVI}}(0_{J-1})\right]  ,
\end{align*}
where
\[
W^{\mathrm{SEVI}}(0_{J})=E\left[  \max_{j\in\mathcal{J}}\varepsilon
_{ij}\right]  ,W^{\mathrm{SEVI}}(0_{J-1})=E\left[  \max_{j\in\mathcal{J}%
\backslash k}\varepsilon_{ij}\right]  \text{ for }\varepsilon_{ij}\thicksim
iid\mathrm{SEVI.}%
\]
By definition, $\mathrm{CV}\left(  k\right)  $ represents compensation for an
individual with systematic utility $V.$ To emphasize the dependence of
$\mathrm{CV}\left(  k\right)  $ on $V$ and the distribution of random
utilities, we denote it as $\mathrm{CV}^{\mathrm{SEVI}}\left(  k,V\right)  $ hereafter.

We can define $\mathrm{CV}^{\mathrm{LEVI}}\left(  k,V\right)  $ similarly
under the LEVI model. For each value of $V,$ the ratio $\mathrm{CV}%
^{\mathrm{SEVI}}\left(  k,V\right)  /\mathrm{CV}^{\mathrm{LEVI}}\left(
k,V\right)  $ quantifies the relative magnitude of the (conditional)
compensating variations across the SEVI and LEVI models. Figure
\ref{Figure: cv_notequalpro_norm_ratio_vs_levi} plots this ratio against
$\mathrm{CV}^{\mathrm{LEVI}}\left(  k,V\right)  $ when $V$ follows the same
data generating process as that in Figure \ref{Figure: Simu_theoretical_pro},
with one of the $X$'s regarded as the price/cost. Each point in the figure
represents the ratio for a particular instance of $k$ and $V.$ Qualitatively,
the figure resembles Figure \ref{Figure: Simu_theoretical_pro_DIFF}, where the
general patterns are applicable to the DGP in Figure
\ref{Figure: Simu_theoretical_pro}. Specifically, when the CV under the LEVI
model is low, the CV under the SEVI model tends to be even lower, and
conversely, when the CV under the LEVI model is high, the CV under the SEVI
model tends to be higher. These results are intuitive: removing an alternative
that is chosen more (or less) often should have a larger (or smaller) welfare implication.

A drawback of Figure \ref{Figure: cv_notequalpro_norm_ratio_vs_levi} is that
each point in the figure does not clearly indicate which alternative is
excluded from the choice set. To address this, we simulate the expected CVs,
$E[\mathrm{CV}^{\mathrm{SEVI}}\left(  k,V\right)  ]$ and $E[\mathrm{CV}%
^{\mathrm{LEVI}}\left(  k,V\right)  ],$ for each alternative $k$ excluded,
with respect to the distribution of $V$. Figure
\ref{Figure: cv_notequalpro_norm_ratio_vs_ki} in the supplementary appendix
plots the ratio of these two expected CVs against the label of the excluded
alternative ($k$). From the figure, it can be observed that the expected CV
under the LEVI appears to be larger than that under the SEVI, regardless of
which alternative is excluded. In addition, the percentage difference between
the two expected CVs appears to grow with $J,$ the number of alternatives.%

\begin{figure}[h]%
\centering
\includegraphics[
height=3.7836in,
width=4.5186in
]%
{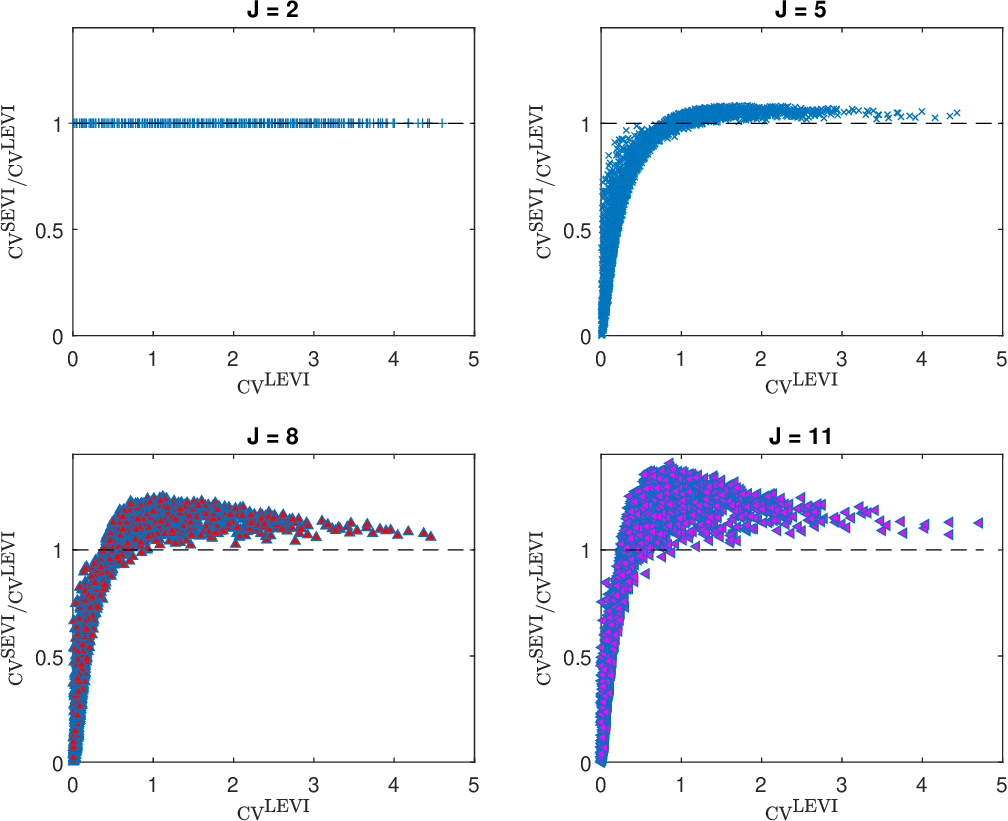}%
\caption{Scatter plot of the ratio of compensating variations in the SEVI and
LEVI models $(\mathrm{CV}^{\mathrm{SEVI}}(k,V)/\mathrm{CV}^{\mathrm{LEVI}%
}(k,V))$ against the compensating variation in the LEVI model $(\mathrm{CV}%
^{\mathrm{LEVI}}(k,V))$ with the same DGPs as in Figure
\ref{Figure: Simu_theoretical_pro} ($X_{ij,\ell}$ is iid $N(0,\pi^{2}%
\omega_{j}^{2}/36))$}%
\label{Figure: cv_notequalpro_norm_ratio_vs_levi}%
\end{figure}

As a third application, we use Proposition \ref{Prop: surplus} to establish
the identification of a SEVI model. Note that the model under the vector of
systematic utilities $V=(V_{1},\ldots,V_{J})$ is equivalent to that under
$V^{o}=(V_{1}-V_{J},\ldots,V_{J-1}-V_{J},0)$, as the choice problem depends
only on utility differences. For the purpose of studying identification, we
normalize the systematic utility of the last alternative to zero and maintain
this normalization in the rest of this subsection. Each vector of systematic
utilities $V\in\mathbb{R}_{o}^{J}:=\mathbb{R}^{J-1}\otimes\left\{  0\right\}
$ maps into a vector of choice probabilities $P\in\mathbb{S}^{J-1}%
:=\{P=\left(  P_{1},\ldots,P_{J-1}\right)  :P_{j}\geq0$ and $\sum_{j=1}%
^{J-1}P_{j}\leq1\}.$ The question remains whether we can deduce the systematic
utilities from the choice probabilities.

Define the conjugate surplus function:%
\[
W_{\ast}^{\mathrm{SEVI}}\left(  P\right)  :=\sup_{v\in\mathbb{R}_{o}^{J}%
}\left[  \sum_{j=1}^{J-1}v_{j}P_{j}-W^{\mathrm{SEVI}}(v)\right]  \text{ for
any }P\in\mathbb{S}^{J-1}.
\]
Since $W^{\mathrm{SEVI}}(v)$ is known and given in Proposition
\ref{Prop: surplus}, it follows by construction that all of $W_{\ast
}^{\mathrm{SEVI}}\left(  \cdot\right)  :$ $\mathbb{S}^{J}\rightarrow
\mathbb{R}$ and $\left\{  \partial W_{\ast}^{\mathrm{SEVI}}\left(
\cdot\right)  /\partial P_{j}:\mathbb{S}^{J}\rightarrow\mathbb{R}\right\}
_{j=1}^{J-1}$ are known functions. Using Theorem 7 of
\citet{SORENSEN2022}%
, we have:\footnote{This may be of independent interest in IO/marketing
applications if the objective is to derive the systematic utilities from the
observed market shares, but here we focus on identification.}
\[
V_{j}=\left.  \frac{\partial W_{\ast}^{\mathrm{SEVI}}\left(  P\right)
}{\partial P_{j}}\right\vert _{P=\left[  P^{\mathrm{SEVI}}(Y=1|V),\ldots
,P^{\mathrm{SEVI}}\left(  Y=J-1|V\right)  \right]  }%
\]
for $j=1,\ldots,J-1.$ That is, we can plug the vector of choice probabilities
into $\left\{  \partial W_{\ast}^{\mathrm{SEVI}}\left(  \cdot\right)
/\partial P_{j}\right\}  _{j=1}^{J-1}$ to recover the vector of systematic utilities.

Given that the vector of choice probabilities is identified, the vector of
systematic utilities is also identified. The problem of identifying the model
parameter $\beta_{0}$ in a SEVI model reduces to the problem of identifying
$\beta_{0}$ from the (joint) distribution of $V^{\ast}\left(  \beta
_{0}\right)  :=(V_{1}\left(  \beta_{0}\right)  -V_{J}\left(  \beta_{0}\right)
,\ldots,V_{J-1}\left(  \beta_{0}\right)  -V_{J}\left(  \beta_{0}\right)
)\in\mathbb{R}^{J-1}.$ Intuitively, if we have an infinite number of
independent draws from the distribution of $V^{\ast}\left(  \beta_{0}\right)
,$ can we pin down $\beta_{0}?$ Mathematically, if
\begin{equation}
\Pr(V^{\ast}\left(  \beta_{1}\right)  \neq V^{\ast}\left(  \beta_{2}\right)
)>0\text{ for any }\beta_{1}\neq\beta_{2}, \label{ID}%
\end{equation}
then $\beta_{0}$ is identified. This identification condition requires that,
for two different model parameters, there is a positive probability that the
resulting vectors of (normalized) systematic utilities are different.\ For
example, in the linear case where $V_{i}^{\ast}\left(  \beta\right)  =\left[
\left(  X_{i1}-X_{iJ}\right)  \beta,\ldots,\left(  X_{iJ-1}-X_{iJ}\right)
\beta\right]  ,$ this requires that there is no perfect multicollinearity
among the variables in $X_{ij}-X_{iJ}$ for at least one $j\in\left\{
1,\ldots,J-1\right\}  .$

The identification condition in (\ref{ID}) does not involve the distribution
of $\varepsilon_{ij}.$ Therefore, it is applicable to any RUM model, provided
that the unobserved random components $\left(  \varepsilon_{i1},\ldots
,\varepsilon_{iJ}\right)  $ is independent of the systematic utility $\left(
V_{i1},\ldots,V_{iJ}\right)  $ and is absolutely continuous with full support
$\mathbb{R}^{J}.$ See Corollary 1 of
\citet{SORENSEN2022}
for more discussions. In particular, this condition is applicable to a LEVI
model. Hence, a SEVI model is identified if and only if the corresponding LEVI
model is identified. No additional identification challenges arise in a SEVI model.

\section{The MLE and QMLE of a SEVI Model \label{Section: MLE_QMLE}}

\subsection{The Estimators}

Consider the following conditional SEVI model%
\[
U_{ij}\left(  \beta_{0}\right)  =V_{ij}\left(  \beta_{0}\right)
+\varepsilon_{ij}\text{ and }V_{ij}\left(  \beta_{0}\right)  =X_{ij}\beta_{0}%
\]
for $j=1,\ldots,J$ where $X_{ij}\in\mathbb{R}^{L}$ and $\varepsilon
_{ij}\thicksim iid$ $\mathrm{SEVI.}$ All other models in Remark
\ref{various_SEVI_models} can be reformulated to take the above form. For
example, if we want to include alternative-specific constants, we can
introduce a dummy variable for each alternative and include the dummies for
$\left(  J-1\right)  $ alternatives as part of $X_{ij}$. This dummy-variable
method can also be used to allow for alternative-specific coefficients for any
covariate, including individual-specific characteristics.

Given a simple random sample $\left\{  Y_{i},X_{i1},\ldots,X_{iJ}\right\}
_{i=1}^{n}$ from the above model, we consider estimating $\beta_{0}$ by MLE
and QMLE. The MLE is the maximum likelihood estimator when the likelihood
function is based on a correctly specified SEVI model, and the QMLE is the
maximum likelihood estimator based on an incorrectly specified LEVI model.

The negative log-likelihood function of the model under the correct SEVI
specification is%
\[
\ell_{\mathrm{SEVI}}\left(  \beta\right)  =\sum_{i=1}^{n}\ell_{i}%
^{\mathrm{SEVI}}\left(  \beta\right)  \text{ for }\ell_{i}^{\mathrm{SEVI}%
}\left(  \beta\right)  =-\sum_{k=1}^{J}Y_{ik}\log\left(  P_{ik}^{\mathrm{SEVI}%
}\left(  \beta\right)  \right)  ,
\]
where $Y_{ik}:=1\left\{  Y_{i}=k\right\}  $ and $P_{ik}^{\mathrm{SEVI}}\left(
\beta\right)  :=P^{\mathrm{SEVI}}\left(  Y_{i}=k|v_{i}\left(  \beta\right)
\right)  $ is defined as in (\ref{SEVI_prob_formula}) with $v_{i}$ replaced by
$v_{i}\left(  \beta\right)  .$ By minimizing the above negative log-likelihood
function $\ell_{\mathrm{SEVI}}\left(  \beta\right)  ,$ we obtain the MLE:%
\[
\hat{\beta}_{\mathrm{MLE}}=\arg\min_{\beta\in\mathcal{B}}\ell_{\mathrm{SEVI}%
}\left(  \beta\right)  ,
\]
where $\mathcal{B}$ is a large enough compact set in $\mathbb{R}^{L}.$

To compute the MLE, we can use a gradient-based algorithm. The gradient
function is
\[
\frac{\partial\ell_{\mathrm{SEVI}}\left(  \beta\right)  }{\partial\beta}%
=\sum_{i=1}^{n}\frac{\partial\ell_{i}^{\mathrm{SEVI}}\left(  \beta\right)
}{\partial\beta}:=\sum_{i=1}^{n}S_{i}^{\mathrm{SEVI}}\left(  \beta\right)  ,
\]
where
\[
S_{i}^{\mathrm{SEVI}}\left(  \beta\right)  =\sum_{j=1}^{J}\frac{\partial
\ell_{i}^{\mathrm{SEVI}}\left(  \beta\right)  }{\partial v_{ij}\left(
\beta\right)  }\frac{\partial v_{ij}\left(  \beta\right)  }{\partial\beta
}=-\sum_{j=1}^{J}\sum_{k=1}^{J}\frac{Y_{ik}}{P_{ik}^{\mathrm{SEVI}}\left(
\beta\right)  }\frac{\partial P_{ik}^{\mathrm{SEVI}}\left(  \beta\right)
}{\partial v_{ij}\left(  \beta\right)  }X_{ij}^{\prime}.
\]
In the above, the expression of the derivative $\frac{\partial P_{ik}%
^{\mathrm{SEVI}}\left(  \beta\right)  }{\partial v_{ij}\left(  \beta\right)
}$ is given in (\ref{substitution1}) or (\ref{substitution2}) after plugging
in the parameter $\beta.$ The closed-form expression for the gradient can be
supplied to a gradient-based optimization algorithm.

Given the score $S_{i}^{\mathrm{SEVI}}\left(  \beta\right)  $ and the
estimator $\hat{\beta}_{\mathrm{MLE}},$ we can estimate the asymptotic
variance of $\hat{\beta}_{\mathrm{MLE}}$ by
\[
\hat{\Omega}=\left[  \frac{1}{n}\sum_{i=1}^{n}S_{i}^{\mathrm{SEVI}}(\hat
{\beta}_{\mathrm{MLE}})S_{i}^{\mathrm{SEVI}}(\hat{\beta}_{\mathrm{MLE}%
})^{\prime}\right]  ^{-1}.
\]
Following standard arguments, we can show that $\hat{\Omega}^{-1/2}\sqrt
{n}(\hat{\beta}_{\mathrm{MLE}}-\beta_{0})\rightarrow^{d}N\left(
0,I_{L}\right)  .$ Details are omitted here, as the theory is standard; See,
for example, Chapter 13 of \citet{Wooldridge2002}.

If we want to test $H_{0}:R\beta_{0}=r$ against $H_{1}:R\beta_{0}\neq r$ for
some $q\times L$ matrix $R$ with row rank $q$ and $q\times1$ vector $r,$ we
can form the Wald statistic and show that it is asymptotically chi-squared
under the null:
\[
n\left(  R\hat{\beta}_{\mathrm{MLE}}-r\right)  \left(  R\hat{\Omega}R^{\prime
}\right)  ^{-1}\left(  R\hat{\beta}_{\mathrm{MLE}}-r\right)  \rightarrow
^{d}\chi_{q}^{2}.
\]
In particular, we can use the square root of the diagonal elements of $\left[
\sum_{i=1}^{n}S_{i}^{\mathrm{SEVI}}(\hat{\beta}_{\mathrm{MLE}})S_{i}%
^{\mathrm{SEVI}}(\hat{\beta}_{\mathrm{MLE}})^{\prime}\right]  ^{-1}$, denoted
by $\hat{\sigma}_{\ell},\ell=1,\ldots,L,$ as the standard errors for the
elements of $\hat{\beta}_{\mathrm{MLE}}.$ We can construct a 95\% confidence
interval as $[\hat{\beta}_{\mathrm{MLE},\ell}-1.96\hat{\sigma}_{\ell}%
,\hat{\beta}_{\mathrm{MLE},\ell}+1.96\hat{\sigma}_{\ell}].$

If the stochastic utilities follow SEVI, but we employ the maximum likelihood
method under the incorrect assumption that they follow LEVI, we obtain the
quasi-MLE (QMLE):\footnote{Here, we focus on the QMLE based on the LEVI
likelihood. Other QMLEs may be considered. For instance, we might explore the
QMLE derived from the specification that the random utility components follow
an MEV distribution (see (\ref{MEV})). As long as the marginal distribution is
incorrectly specified, such an alternative QMLE will generally be inconsistent
for the preference parameters.}
\begin{align}
\hat{\beta}_{\mathrm{QMLE}}  &  =\arg\min_{\beta\in\mathcal{B}}\ell
_{\mathrm{LEVI}}\left(  \beta\right)  \text{ where }\ell_{\mathrm{LEVI}%
}\left(  \beta\right)  =\sum_{i=1}^{n}\ell_{i}^{\mathrm{LEVI}}\left(
\beta\right)  \text{ for }\label{l_LEVI}\\
\ell_{i}^{\mathrm{LEVI}}\left(  \beta\right)   &  =-\sum_{k=1}^{J}Y_{ik}%
\log\left(  P_{ik}^{\mathrm{LEVI}}\left(  \beta\right)  \right)  =-\sum
_{k=1}^{J}Y_{ik}\left[  X_{ik}\beta\right]  -\log\left\{  \sum_{j=1}^{J}%
\exp\left[  X_{ij}\beta\right]  \right\}  .\nonumber
\end{align}
The asymptotic properties of the QMLE are also standard results. In
particular, we use the sandwich form for estimating the asymptotic variance of
$\hat{\beta}_{\mathrm{QMLE}}$ when we know that the LEVI model is not the
correct model.

\subsection{On the Estimation of Parameter Ratios}

Under some standard regularity conditions, $\hat{\beta}_{\mathrm{MLE}}$ is
consistent for $\beta_{0},$ but in general, $\hat{\beta}_{\mathrm{QMLE}}$ is
not. However, the ratio of the individual elements of $\hat{\beta
}_{\mathrm{QMLE}}$, say $\hat{\beta}_{\mathrm{QMLE},2}/\hat{\beta
}_{\mathrm{QMLE},1},$ may still be consistent for the true ratio, namely
$\beta_{0,2}/\beta_{0,1},$ for a nonzero $\beta_{0,1}.$ In this subsection, we
provide sufficient conditions for the consistency of the ratio estimator, even
though the distribution of the random utility component is misspecified.

We consider an RUM with $V_{ij}\left(  \beta_{0}\right)  =X_{ij}\beta_{0}$ and
\emph{specify} $\varepsilon_{ij}\thicksim iid\mathrm{Q}$ for some distribution
$\mathrm{Q}.$ This includes LEVI as an example but $\mathrm{Q}$ may be more
general and incorrectly specified. Denote the choice probability derived under
$\mathrm{Q}$ by $P^{\mathrm{Q}}(Y_{i}=k|V_{i}\left(  \beta\right)  )$. Let
$\beta_{\mathrm{Q}}^{\ast}$ be the pseudo-true value defined by
\[
\beta_{\mathrm{Q}}^{\ast}=\arg\min_{\beta\in\mathcal{B}}-E\sum_{k=1}^{J}%
Y_{ik}\log P^{\mathrm{Q}}(Y_{i}=k|V_{i}\left(  \beta\right)  ),
\]
which is also the probability limit of $\hat{\beta}_{\mathrm{Q}}:=\arg
\min_{\beta\in\mathcal{B}}-n^{-1}\sum_{i=1}^{n}\sum_{k=1}^{J}Y_{ik}\log
P^{\mathrm{Q}}(Y_{i}=k|V_{i}\left(  \beta\right)  ).$

We assume that $\beta_{\mathrm{Q}}^{\ast}$ is an interior point of the compact
set $\mathcal{B}$ so that $\beta_{\mathrm{Q}}^{\ast}$ satisfies the first
order zero-derivative conditions. The derivative of the population objective
function with respect to $\beta$ is
\[
-E\sum_{k=1}^{J}Y_{ik}\frac{\partial\log P^{\mathrm{Q}}(Y_{i}=k|\mathcal{J}%
,V_{i}\left(  \beta\right)  )}{\partial\beta}=-EX_{i}^{\prime}\sum_{k=1}%
^{J}Y_{ik}\lambda_{ik}^{\mathrm{Q}}\left(  X_{i}\beta\right)  =-EX_{i}%
^{\prime}\Lambda_{i}^{\mathrm{Q}}\left(  X_{i}\beta\right)  \overrightarrow{Y}%
_{i},
\]
where%
\begin{align*}
X_{i}  &  =(X_{i1}^{\prime},\ldots,X_{iJ}^{\prime})^{\prime}\in\mathbb{R}%
^{J\times L},\\
\lambda_{ik}^{\mathrm{Q}}\left(  X_{i}\beta\right)   &  =\left.
\frac{\partial\log P^{\mathrm{Q}}(Y_{i}=k|\mathcal{J},V_{i}=v_{i})}{\partial
v_{i}}\right\vert _{v_{i}=X_{i}\beta}\in\mathbb{R}^{J\times1},\\
\Lambda_{i}^{\mathrm{Q}}\left(  X_{i}\beta\right)   &  =\left[  \lambda
_{i1}^{\mathrm{Q}}\left(  X_{i}\beta\right)  ,\ldots,\lambda_{iJ}^{\mathrm{Q}%
}\left(  X_{i}\beta\right)  \right]  \in\mathbb{R}^{J\times J},\text{
}\overrightarrow{Y}_{i}=(Y_{i1},\ldots,Y_{iJ})^{\prime}\in\mathbb{R}%
^{J\times1}.
\end{align*}
So, $\beta_{\mathrm{Q}}^{\ast}$ satisfies $E\left[  X_{i}^{\prime}\Lambda
_{i}^{\mathrm{Q}}(X_{i}\beta_{\mathrm{Q}}^{\ast})\overrightarrow{Y}%
_{i}\right]  =0.$

\begin{proposition}
\label{Prop: proportional}Assume that (i) $\hat{\beta}_{\mathrm{Q}}%
\rightarrow^{p}\beta_{\mathrm{Q}}^{\ast}\in Int(\mathcal{B)}$, and
$\beta_{\mathrm{Q}}^{\ast}$ is the unique solution to $E\left[  X_{i}^{\prime
}\Lambda_{i}^{\mathrm{Q}}\left(  X_{i}\beta\right)  \overrightarrow{Y}%
_{i}\right]  =0;$

(ii) $\beta_{0}\neq0$, and there exists a nonzero scalar $\delta_{1}^{\circ}$
such that $E\left[  X_{i}^{\prime}\Lambda_{i}^{\mathrm{Q}}\left(  X_{i}%
\beta_{0}\delta_{1}^{\circ}\right)  \overrightarrow{Y}_{i}\right]  =0;$

(iii) $E\left[  X_{i,j,\ell}|X_{i1}\beta_{0},X_{i2}\beta_{0},\ldots
,X_{iJ}\beta_{0}\right]  =\Theta_{0,\ell}+\left(  X_{ij}\beta_{0}\right)
\Theta_{1,\ell}$ for some scalars $\Theta_{0,\ell}$ and $\Theta_{1,\ell}$ and
for $\ell=2,\ldots,L.$

Then: $\beta_{\mathrm{Q}}^{\ast}=\delta_{1}^{\circ}\beta_{0}.$
\end{proposition}

Proposition \ref{Prop: proportional} shows that while the QMLE $\hat{\beta
}_{\mathrm{Q}}$ may not be consistent for $\beta_{0}$, it converges to
$\beta_{0}$ up to a multiplicative factor. In particular, if $\beta_{0,m_{1}%
}\neq0,$ then
\[
\frac{\hat{\beta}_{\mathrm{Q,}m_{2}}}{\hat{\beta}_{\mathrm{Q,}m_{1}}%
}\rightarrow^{p}\frac{\beta_{0,m_{2}}}{\beta_{0,m_{1}}}.
\]
This says that the plug-in ratio estimator is consistent for the true ratio.

Assumption (iii) in Proposition \ref{Prop: proportional} holds if
$X_{i1},\ldots,X_{iJ}$ are independent and each follows a multivariate normal
or elliptical distribution. The latter assumption is sufficient but not
necessary. There may exist other distributions of $X_{i1},\ldots,X_{iJ}$ or
specific ways to relax the independence assumption such that the
proportionality result holds. In empirical applications, the ratio estimator
could still be close to the true ratio even though Assumption (iii) does not hold.

Proposition \ref{Prop: proportional} is proved along the lines of the proof of
Theorem 1 in
\citet{RUUD1986}%
, which does not cover multinomial choice models. Some adjustments to the
proof in\
\citet{RUUD1986}
are needed to account for two main differences: there is no constant term in a
multinomial choice model and the covariate $X_{i}$ in this model is a matrix
rather than a vector.

\subsection{The Choice between SEVI and LEVI \label{Section: SEVI vs LEVI}}

In empirical applications, it is often the case that neither the SEVI nor the
LEVI is correctly specified. However, we still have to decide which model to
use. To aid in this decision, we can employ the Akaike Information Criterion
(AIC) or the Bayesian Information Criterion (BIC). Since both models have the
same number of parameters, we can compare their empirical (quasi)
log-likelihoods, which leads to the likelihood ratio statistic given by
\begin{align*}
\mathrm{LR}  &  :=\ell_{\mathrm{SEVI}}(\hat{\beta}_{\mathrm{SEVI}}%
)-\ell_{\mathrm{LEVI}}(\hat{\beta}_{\mathrm{LEVI}})\\
&  =\sum_{i=1}^{n}\left[  \ell_{i}^{\mathrm{SEVI}}\left(  \hat{\beta
}_{\mathrm{SEVI}}\right)  -\ell_{i}^{\mathrm{LEVI}}\left(  \hat{\beta
}_{\mathrm{LEVI}}\right)  \right]  .
\end{align*}
Note that in the previous section, $\hat{\beta}_{\mathrm{SEVI}}$ and
$\hat{\beta}_{\mathrm{LEVI}}$ were denoted as $\hat{\beta}_{\mathrm{MLE}}$ and
$\hat{\beta}_{\mathrm{QMLE}},$ respectively. However, in this section, both
the SEVI model and the LEVI model can be misspecified and so both $\hat{\beta
}_{\mathrm{SEVI}}$ and $\hat{\beta}_{\mathrm{LEVI}}$ can be quasi-maximum
likelihood estimators.

In empirical applications, the choice between the SEVI model and the LEVI
model can be based on the likelihood ratio. If $\mathrm{LR}>0$, then the LEVI
model is preferred.\footnote{Note that $\ell_{\mathrm{SEVI}}\left(
\cdot\right)  $ and $\ell_{\mathrm{LEVI}}\left(  \cdot\right)  $ are the
negative log-likelihood functions so when \ $LR>0,$ we may conclude that the
LEVI model fits the data better.} Conversely, if $\mathrm{LR}<0$, then the
SEVI model is preferred. This simple empirical rule of thumb is derived from
the empirical goodness of fit and serves as a practical guide for selecting
between the two models.\ It is worth pointing out again that this rule is also
the model selection rule based on the AIC or BIC, as the two models have the
same number of parameters.

Another method to determine which model to use is by employing Vuong's test (%
\citet{Vuong1989}%
), which formally assesses the null hypothesis that the SEVI model and the
LEVI model provide equally good fits to the data. Let $\beta_{\mathrm{SEVI}%
}^{\ast}$ and $\beta_{\mathrm{LEVI}}^{\ast}$ be the respective limits of
$\hat{\beta}_{\mathrm{SEVI}}$ and $\hat{\beta}_{\mathrm{LEVI}}$. In the
present setting, the null hypothesis of Vuong's test can be expressed as:
\[
H_{0}:E\left[  \ell_{\mathrm{SEVI}}\left(  \beta_{\mathrm{SEVI}}^{\ast
}\right)  |V\right]  =E\left[  \ell_{\mathrm{LEVI}}\left(  \beta
_{\mathrm{LEVI}}^{\ast}\right)  |V\right]
\]
almost surely, where $V$ is the vector collecting the systematic utilities.
Note that
\[
E\left[  \ell_{\mathrm{SEVI}}\left(  \beta_{\mathrm{SEVI}}^{\ast}\right)
|V\right]  =-\sum_{k=1}^{J}P_{ik}^{o}\log\left(  P_{ik}^{\mathrm{SEVI}}\left(
\beta_{\mathrm{SEVI}}^{\ast}\right)  \right)
\]
where $P_{ik}^{o}:=\Pr(Y_{i}=k|V_{i})$ is the true choice probability. The
right-hand side of the above equation is the cross entropy between $\left(
P_{i1}^{0},\ldots,P_{iJ}^{o}\right)  $ and ($P_{i1}^{\mathrm{SEVI}}\left(
(\beta_{\mathrm{SEVI}}^{\ast}\right)  ,\ldots,P_{iJ}^{\mathrm{SEVI}}\left(
\beta_{\mathrm{SEVI}}^{\ast}\right)  $), which quantifies the difference
between these two distributions. By adding a constant $\sum_{k=1}^{J}%
P_{ik}^{o}\log\left(  P_{ik}^{o}\right)  $ to both sides, we have
\[
E\left[  \ell_{\mathrm{SEVI}}\left(  \beta_{\mathrm{SEVI}}^{\ast}\right)
|V\right]  +\sum_{k=1}^{J}P_{ik}^{o}\log\left(  P_{ik}^{o}\right)
=-\sum_{k=1}^{J}P_{ik}^{o}\log\left(  \frac{P_{ik}^{\mathrm{SEVI}}\left(
\beta_{\mathrm{SEVI}}^{\ast}\right)  }{P_{ik}^{o}}\right)  ,
\]
where the right-hand side is the KL distance between the two distributions
$\left(  P_{i1}^{0},\ldots,P_{iJ}^{o}\right)  $ and ($P_{i1}^{\mathrm{SEVI}%
}\left(  \beta_{\mathrm{SEVI}}^{\ast}\right)  ,\ldots,P_{iJ}^{\mathrm{SEVI}%
}\left(  \beta_{\mathrm{SEVI}}^{\ast}\right)  $) conditional on $V.$ Similar
calculations can be performed for $E\left[  \ell_{\mathrm{LEVI}}\left(
\beta_{\mathrm{LEVI}}^{\ast}\right)  |V\right]  .$ Therefore, the null of
Vuong's test is that the SEVI and LEVI models are equidistant from the true
data generating process, with the \textquotedblleft distance\textquotedblright%
\ measured by either the cross-entropy or the KL distance.

Since the SEVI model and the LEVI model are nonnested, $n^{-1/2}\left[
\ell_{\mathrm{SEVI}}(\hat{\beta}_{\mathrm{SEVI}})-\ell_{\mathrm{LEVI}}%
(\hat{\beta}_{\mathrm{LEVI}})\right]  $ is asymptotically normal. We can
standardize the LR statistic to obtain the $\mathcal{V}_{n}$ statistic of
\citet{Vuong1989}%
:
\[
\mathcal{V}_{n}=\frac{\sum_{i=1}^{n}\left[  \ell_{i}^{\mathrm{SEVI}}\left(
\hat{\beta}_{\mathrm{SEVI}}\right)  -\ell_{i}^{\mathrm{LEVI}}\left(
\hat{\beta}_{\mathrm{LEVI}}\right)  \right]  }{\left\{  \sum_{i=1}^{n}\left[
\ell_{i}^{\mathrm{SEVI}}\left(  \hat{\beta}_{\mathrm{SEVI}}\right)  -\ell
_{i}^{\mathrm{LEVI}}\left(  \hat{\beta}_{\mathrm{LEVI}}\right)  -n^{-1}%
\mathrm{LR}\right]  ^{2}\right\}  ^{1/2}}.
\]
Following the arguments in
\citet{Vuong1989}%
, we can show that under the null $\mathcal{V}_{n}\rightarrow^{d}N(0,1).$ This
serves as the basis for both one-sided tests and two-sided tests. For example,
if $\left\vert \mathcal{V}_{n}\right\vert >1.96,$ then we may reject the null
that the SEVI and LEVI models fit the data equally well at the 5\%
significance level. On the other hand, if $\mathcal{V}_{n}>1.645,$ we may
reject the null that the SEVI model fits the data better at the 5\%
significance level. If $\mathcal{V}_{n}<-1.645,$ we may reject the null that
the LEVI model fits the data better at the 5\% significance level.

\section{Simulation Evidence\label{Sec: simulation}}

We generate the data from a conditional SEVI model. More specifically, we
assume that%
\[
Y_{i}=\arg\max_{j=1,\ldots,J}\left\{  X_{ij}\beta+\varepsilon_{ij}\right\}
\]
where $\varepsilon_{ij}\thicksim iid$ SEVI over $j$ and $X_{ij,\ell}\thicksim
iidN(0,\pi^{2}\omega_{j}^{2}/36)$ or $X_{ij,\ell}\thicksim iidN(0,\pi^{2}/36)$
over $\ell=1,2$ and $3$. The DGPs are the same as the SEVI DGPs in Figures
\ref{Figure: Simu_theoretical_pro}--\ref{Figure: Simu_theoretical_pro_DIFF2}.
Given a simple random sample $\left\{  Y_{i},X_{i1},\ldots,X_{iJ}\right\}
_{i=1}^{n},$ we use the MLE and QMLE in Section \ref{Section: MLE_QMLE} to
estimate $\beta_{0}=\left(  1,2,1\right)  ^{\prime}.$ To reiterate, the
MLE\ is based on the correct specification that $\varepsilon_{ij}\thicksim
iid$ SEVI, and the QMLE\ is based on the incorrect specification that
$\varepsilon_{ij}$ $\thicksim iid$ LEVI.

\subsection{Simulation with a Large Sample Size}

We first consider the case that the sample size is large, and there is only
one simulation replication. The purpose of this exercise is to see how large
the misspecification bias is. Table \ref{Table:MLE_QMLE} reports the estimates
when the sample size $n$ is $10,000.$ When $J=2,$ the MLE and QMLE give the
same estimate of $\beta_{0}.$ This is consistent with the result that the SEVI
model and the LEVI model are the same when there are only two alternatives.
When $J>2,$ the MLE and QMLE are quite different. While the MLE is close to
the true parameter vector for all values of $J$ considered, the QMLE grows
with $J.$ For example, when $J=8$ and $X_{ij,\ell}\thicksim iidN(0,\pi
^{2}\omega_{j}^{2}/36),$ the QMLE is $\left(  1.48,2.99,1.49\right)  $, each
element of which is higher than the corresponding true value by about 50\%.
The misspecification bias of the QMLE is substantial. In line with Proposition
\ref{Prop: proportional}, the ratio of the estimates of any two elements of
$\beta_{0}$, say $\hat{\beta}_{\mathrm{QMLE},2}/\hat{\beta}_{\mathrm{QMLE}%
,1},$ is close to the true ratio $\beta_{0,2}/\beta_{0,1}.$ This is true
regardless of whether the attributes have the same variance across the
alternatives or not. Proposition \ref{Prop: proportional} permits variance
heterogeneity across the alternatives.%

\begin{table}[tbp] \centering
%

\begin{tabular}
[c]{l|ccccccccccccc}\hline
& \textbf{True} &  & \multicolumn{2}{c}{$J=2$} &  & \multicolumn{2}{c}{$J=5$}
&  & \multicolumn{2}{c}{$J=8$} &  & \multicolumn{2}{c}{$J=11$}\\\cline{1-2}%
\cline{4-5}\cline{7-8}\cline{10-11}\cline{13-14}
& $(\boldsymbol{\beta}_{\boldsymbol{0}})$ &  & \textbf{MLE} & \textbf{QMLE} &
& \textbf{MLE} & \textbf{QMLE} &  & \textbf{MLE} & \textbf{QMLE} &  &
\textbf{MLE} & \textbf{QMLE}\\\hline
&  &  & \multicolumn{11}{c}{$X_{ij,\ell}\thicksim iidN(0,\pi^{2}\omega_{j}%
^{2}/36)$}\\\hline
$\beta_{1}$ & 1.00 &  & 1.01 & 1.01 &  & 1.00 & 1.37 &  & 0.98 & 1.48 &  &
1.02 & 1.64\\
$\beta_{2}$ & 2.00 &  & 1.95 & 1.95 &  & 2.01 & 2.75 &  & 1.99 & 2.99 &  &
1.98 & 3.17\\
$\beta_{3}$ & 1.00 &  & 0.99 & 0.99 &  & 1.00 & 1.37 &  & 1.01 & 1.49 &  &
0.98 & 1.57\\\hline
&  &  & \multicolumn{11}{c}{$X_{ij,\ell}\thicksim iidN(0,\pi^{2}/36)$}\\\hline
$\beta_{1}$ & 1.00 &  & 1.01 & 1.01 &  & 1.01 & 1.35 &  & 0.99 & 1.51 &  &
1.00 & 1.64\\
$\beta_{2}$ & 2.00 &  & 2.03 & 2.03 &  & 2.05 & 2.80 &  & 1.98 & 3.09 &  &
1.98 & 3.37\\
$\beta_{3}$ & 1.00 &  & 1.03 & 1.03 &  & 1.03 & 1.38 &  & 1.03 & 1.56 &  &
1.00 & 1.64\\\hline
\end{tabular}

\caption{MLE and QMLE of the model parameters for various J }\label{Table:MLE_QMLE}%
\end{table}%

In a nonlinear model, instead of focusing on the parameter values, it can be
more informative to examine the (empirical) average partial effects, defined
by
\begin{align*}
APE(j;\beta,\mathrm{SEVI})  &  =\frac{1}{n}\sum_{i=1}^{n}\frac{\partial
P_{ij}^{\mathrm{SEVI}}\left(  \beta\right)  }{\partial X_{ij,1}}=\frac{1}%
{n}\sum_{i=1}^{n}\frac{\partial P_{ij}^{\mathrm{SEVI}}\left(  \beta\right)
}{\partial v_{ij}\left(  \beta\right)  }\beta_{1},\\
APE(j;\beta,\mathrm{LEVI})  &  =\frac{1}{n}\sum_{i=1}^{n}\frac{\partial
P_{ij}^{\mathrm{LEVI}}\left(  \beta\right)  }{\partial X_{ij,1}}=\frac{1}%
{n}\sum_{i=1}^{n}P_{ij}^{\mathrm{LEVI}}\left(  \beta\right)  \left(
1-P_{ij}^{\mathrm{LEVI}}\left(  \beta\right)  \right)  \beta_{1}.
\end{align*}
In the above, $P_{ij}^{\mathrm{SEVI}}$ and $P_{ij}^{\mathrm{LEVI}}$ are the
choice probabilities under the SEVI and LEVI specifications, respectively.
$APE(j;\beta,\mathrm{SEVI})$ and $APE(j;\beta,\mathrm{LEVI})$ are the
corresponding empirical APEs that measure the effect on the probability of
choosing alternative $j$ when a marginal change is made to one of alternative
$j$'s attributes. We focus on the first attribute $X_{ij,1}$ and plot the
ratios:
\[
\frac{APE(j;\hat{\beta}_{\mathrm{MLE}},\mathrm{SEVI})}{APE(j;\beta
_{0},\mathrm{SEVI})}\text{ and }\frac{APE(j;\hat{\beta}_{\mathrm{QMLE}%
},\mathrm{LEVI})}{APE(j;\beta_{0},\mathrm{SEVI})}%
\]
against the label of each alternative $j=1,\ldots,J.$ Figure \ref{Figure: APE}
reports the results for the case that $X_{ij,\ell}\thicksim iidN(0,\pi
^{2}\omega_{j}^{2}/36)$ and makes it clear that the empirical APE is not close
to the target APE (i.e., $APE(j;\beta_{0},\mathrm{SEVI})$) when the model is
misspecified. This shows that the difference in the parameter estimates due to
model misspecification has substantial effects on APE estimation.%

\begin{figure}[tbh]%
\centering
\includegraphics[
height=3.9902in,
width=4.8456in
]%
{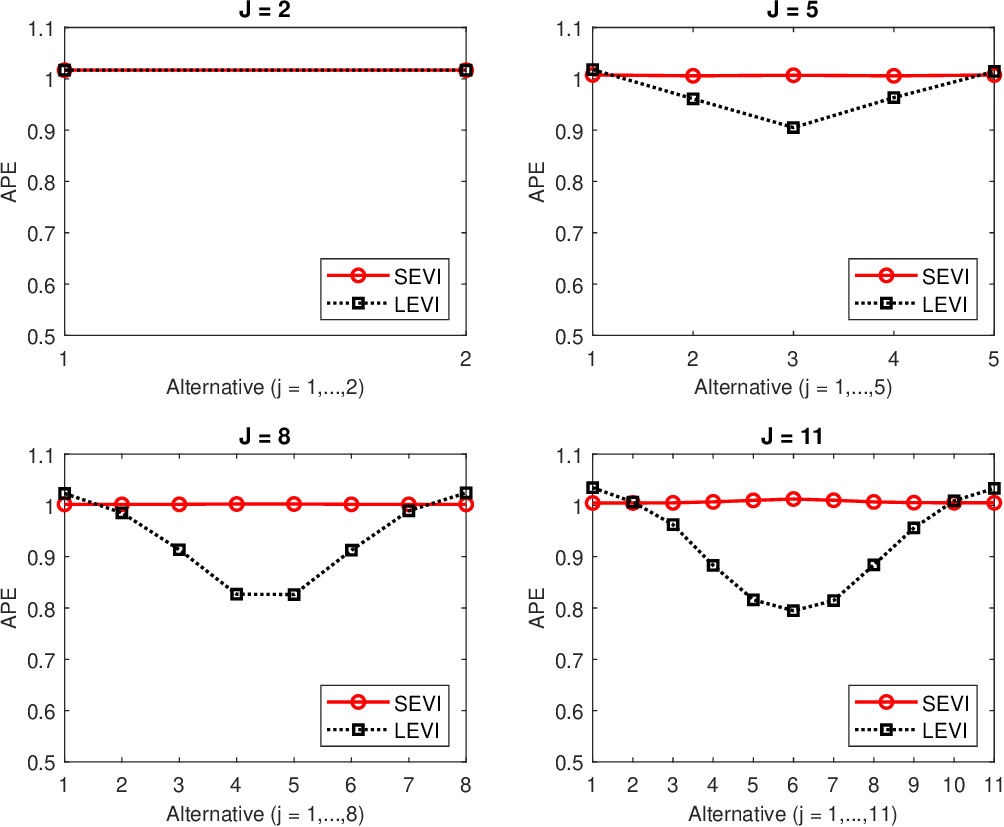}%
\caption{Plots of the empirical APE under the correct SEVI specification and
the incorrect LEVI specification against alternative labels for a sample of
$n=10,000$ individuals}%
\label{Figure: APE}%
\end{figure}

To further illustrate the difference between the SEVI and LEVI specifications,
we conduct an out-of-sample prediction exercise. Let $\hat{\beta
}_{\mathrm{MLE}}^{\left(  J\right)  }$ and $\hat{\beta}_{\mathrm{QMLE}%
}^{\left(  J\right)  }$ be the MLE and QMLE of $\beta_{0}$, respectively, when
there are $J$ alternatives. We use $\hat{\beta}_{\mathrm{MLE}}^{\left(
J\right)  }$ and $\hat{\beta}_{\mathrm{QMLE}}^{\left(  J\right)  }$ to compute
out-of-sample choice probabilities for a certain choice problem. More
specifically, we assume that an out-of-sample individual chooses between the
first two alternatives only, and the rest $J-2$ alternatives are not
available. In this case, the out-of-sample choice probabilities are%
\[
\hat{P}_{i_{o},\mathrm{MLE}}:=P\left(  Y_{i_{o}}=1|X_{i_{o}};\hat{\beta
}_{\mathrm{MLE}}^{\left(  J\right)  }\right)  \text{ and }\hat{P}%
_{i_{o},\mathrm{QMLE}}=P\left(  Y_{i_{o}}=1|X_{i_{o}};\hat{\beta
}_{\mathrm{QMLE}}^{\left(  J\right)  }\right)
\]
where
\[
P\left(  Y_{i_{o}}=1|X_{i_{o}};\beta\right)  =\frac{\exp(X_{i_{o},1}\beta
)}{\exp(X_{i_{o},1}\beta)+\exp(X_{i_{o},2}\beta)}%
\]
and $i_{o}\in\left\{  n+1,\ldots,2n\right\}  $ represents an out-of-sample
individual whose observation $(X_{i_{o}},Y_{i_{o}})$ was not used in
constructing the MLE or QMLE, but $\left(  X_{i_{o}},Y_{i_{o}}\right)  $ is
generated in the same way as any in-sample observation. Note that the number
of \textquotedblleft out-of-sample\textquotedblright\ individuals is the same
as the number of \textquotedblleft in-sample\textquotedblright\ individuals
and both are equal to 10,000.

When there are only two alternatives in the out-of-sample prediction exercise,
the formulae for the choice probabilities are the same under SEVI or LEVI.
Hence, as given above, $\hat{P}_{i_{o},\mathrm{MLE}}$ and $\hat{P}%
_{i_{o},\mathrm{QMLE}}$ are different only because they are based on different
estimated parameters.

Figure \ref{Figure: KDE unequalJ511}(a\&b) reports the kernel density
estimators (KDE) of $\left\{  P\left(  Y_{i_{o}}=1|X_{i_{o}};\beta\right)
\right\}  _{i=n+1}^{2n}$ for $\beta=\beta_{0},\hat{\beta}_{\mathrm{MLE}%
}^{\left(  J\right)  }$ and $\hat{\beta}_{\mathrm{QMLE}}^{\left(  J\right)  }$
when $J=5$ and $X_{ij,\ell}\thicksim iidN(0,\pi^{2}\omega_{j}^{2}/36).$ It is
not surprising that the KDE under $\hat{\beta}_{\mathrm{MLE}}^{\left(
J\right)  }$ is very close to the target KDE, that is, the KDE under
$\beta_{0}.$ When the sample size is large, $\hat{\beta}_{\mathrm{MLE}%
}^{\left(  J\right)  }$ is very close to $\beta_{0}$, as demonstrated in Table
\ref{Table:MLE_QMLE}. As a result, $P(Y_{i_{o}}=1|X_{i_{o}};\hat{\beta
}_{\mathrm{MLE}}^{\left(  J\right)  })$ is close to $P\left(  Y_{i_{o}%
}=1|X_{i_{o}};\beta_{0}\right)  $ and therefore, the corresponding KDEs are
close to each other. On the other hand, the KDE under $\hat{\beta
}_{\mathrm{QMLE}}^{\left(  J\right)  }$ is quite different from the target
KDE. This demonstrates that the model misspecification introduces not only a
significant bias in the parameter estimator but also a substantial discrepancy
in the predicted choice probabilities.

The KDE figure for the case that $J=5$ and $X_{ij,\ell}\thicksim
iidN(0,\pi^{2}/36)$ is qualitatively similar to Figure
\ref{Figure: KDE unequalJ511}(a\&b) and is thus omitted. The KDE figure for
the case that $J=11$ and $X_{ij,\ell}\thicksim iidN(0,\pi^{2}\omega_{j}%
^{2}/36)$ is reported in Figure \ref{Figure: KDE unequalJ511}(c\&d). The
difference between the KDE under $\hat{\beta}_{\mathrm{QMLE}}^{(J)}$ and the
target KDE appears to increase with $J.$%

\begin{figure}[h]%
\centering
\includegraphics[
height=4.0231in,
width=4.8931in
]%
{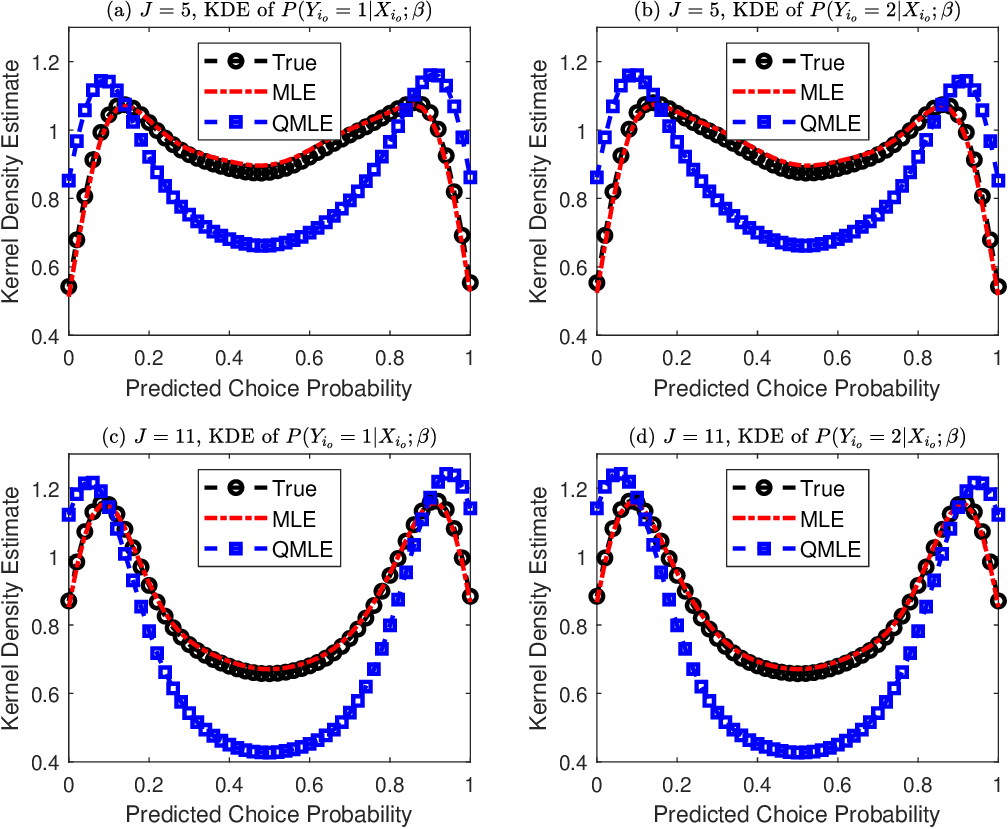}%
\caption{Kernel density estimation of the out-of-sample predictive
probabilities when $n=10,000,J=5,11$ and $X_{ij,\ell}\thicksim iidN(0,\pi
^{2}\omega_{j}^{2}/36)$}%
\label{Figure: KDE unequalJ511}%
\end{figure}

As final evidence that it matters whether the stochastic utility follows the
SEVI or LEVI distribution, we compute the out-of-sample predicted choice
probabilities when there are four alternatives $(J=4)$ for both the in-sample
estimation and out-of-sample prediction. For the latter, we compute $\hat
{P}_{i_{o},\mathrm{MLE}}^{\mathrm{SEVI}}(j):=P(Y_{i_{o}}=j|X_{i_{o}%
};\mathrm{SEVI},$ $\hat{\beta}_{\mathrm{MLE}}^{\left(  4\right)  })$ and
$\hat{P}_{i_{o},\mathrm{QMLE}}^{\mathrm{LEVI}}\left(  j\right)  :=P(Y_{i_{o}%
}=j|X_{i_{o}};\mathrm{LEVI},\hat{\beta}_{\mathrm{QMLE}}^{\left(  4\right)  })$
where we have emphasized the dependence of the choice probability on the
distributional assumption imposed on the stochastic utility. We plot the
difference $\hat{P}_{i_{o},\mathrm{MLE}}^{\mathrm{SEVI}}(j)-\hat{P}%
_{i_{o},\mathrm{QMLE}}^{\mathrm{LEVI}}(j)$ against the true choice probability
$P_{i_{o},0}^{\mathrm{SEVI}}(j):=P(Y_{i_{o}}=j|X_{i_{o}};\mathrm{SEVI},$
$\beta_{0})$ as a scatter plot. We omit the figure here, but it is available
as Figure \ref{Figure: difference in choice prob} in the supplementary
appendix. The figure clearly shows that, relative to the true choice
probability, the difference can be substantial.

\subsection{Simulation with Smaller Sample Sizes}

In this subsection, we study the performances of MLE and QMLE in finite
samples. Using the same DGP as in the previous subsection, we consider the
sample size $n=500$, $1000$ and the number of alternatives $J=2,5,8,11.$ The
number of simulation replications is 5000. For each estimator, we compute the
bias, the empirical standard derivation, and the average of the standard
errors across the simulation replications. Furthermore, we construct a 95\%
confidence interval for each element of $\beta_{0}$ following the standard
procedure of adding and subtracting 1.96 times the standard error and report
the empirical coverage of the CI (i.e., the proportion of times the
so-constructed confidence interval contains the true parameter value).

Table \ref{Table: finite sample MLE QMLE} reports the results when
$X_{ij,\ell}\thicksim iidN(0,\pi^{2}\omega_{j}^{2}/36)$ and $n=500$. The
reported results are representative of other cases with a different DGP for
$X_{ij,\ell}$ or sample size. A few patterns emerge. First, the bias of the
MLE is small and decreases with $J.$ On the other hand, the bias of the
QMLE\ is large and increases with $J.$ This is consistent with our large
sample result when $n=10,000$ and there is only one simulation replication.
Here, the finite sample biases (with $n=500$) are averaged over 5000
simulation replications. Second, for both the MLE and QMLE, the empirical
standard deviation decreases with $J.$ A larger $J$ can be regarded as having
more observations and, hence, smaller sampling noises. Third, for both
estimators, the average of the standard errors across the simulation
replications matches with the corresponding standard deviation very well. This
indicates that the standard errors are reliable estimates of the standard
deviation. Finally, for the MLE, the empirical coverage of the 95\% confidence
intervals is very close to 95\%. This shows that the asymptotic normal
approximation is reliable. On the other hand, for the QMLE, the empirical
coverage of the 95\% confidence intervals is much smaller than 95\%. The only
exception is the case when $J=2$, where MLE and QMLE are identical. In this
case, there is a small and nearly invisible difference in the empirical
coverages, because the standard error for the MLE is based on the non-robust
asymptotic variance estimator while the standard error for the QMLE is based
on the robust sandwich asymptotic variance estimator. Both variance estimators
are consistent when $J=2.$ When $J>2,$ the under-coverage of the CIs based on
the QMLE is due to the bias in the QMLE. The larger the number of alternatives
$J$ is, the larger the bias is, and the smaller the empirical coverage is.%

\begin{table}[tbp] \centering
%

\begin{tabular}
[c]{lccccccccc}\hline\hline
& \multicolumn{4}{c}{MLE} &  & \multicolumn{4}{c}{QMLE}\\\hline\hline
& {\small Empirical} & {\small Empirical} & {\small Average of} &
{\small Empirical} &  & {\small Empirical} & {\small Empirical} &
{\small Average of} & {\small Empirical}\\
& {\small Bias} & {\small StD} & {\small SE's} & {\small Coverage} &  &
{\small Bias} & {\small StD} & {\small SE's} & {\small Coverage}\\\hline
& \multicolumn{9}{c}{$J=2$}\\
$\beta_{1}$ & 0.009 & 0.211 & 0.210 & 0.951 &  & 0.009 & 0.211 & 0.209 &
0.948\\
$\beta_{2}$ & 0.023 & 0.240 & 0.241 & 0.950 &  & 0.023 & 0.240 & 0.239 &
0.950\\
$\beta_{3}$ & 0.010 & 0.209 & 0.210 & 0.950 &  & 0.010 & 0.209 & 0.209 &
0.950\\\hline
& \multicolumn{9}{c}{$J=5$}\\
$\beta_{1}$ & 0.003 & 0.105 & 0.107 & 0.957 &  & 0.366 & 0.141 & 0.142 &
0.261\\
$\beta_{2}$ & 0.010 & 0.136 & 0.137 & 0.950 &  & 0.738 & 0.175 & 0.173 &
0.004\\
$\beta_{3}$ & 0.004 & 0.106 & 0.107 & 0.958 &  & 0.367 & 0.142 & 0.142 &
0.266\\\hline
& \multicolumn{9}{c}{$J=8$}\\
$\beta_{1}$ & 0.006 & 0.087 & 0.088 & 0.952 &  & 0.514 & 0.130 & 0.129 &
0.017\\
$\beta_{2}$ & 0.011 & 0.117 & 0.116 & 0.951 &  & 1.026 & 0.159 & 0.158 &
0.000\\
$\beta_{3}$ & 0.006 & 0.088 & 0.088 & 0.951 &  & 0.513 & 0.129 & 0.129 &
0.017\\\hline
& \multicolumn{9}{c}{$J=11$}\\
$\beta_{1}$ & 0.004 & 0.080 & 0.080 & 0.950 &  & 0.600 & 0.124 & 0.123 &
0.000\\
$\beta_{2}$ & 0.007 & 0.105 & 0.107 & 0.955 &  & 1.201 & 0.149 & 0.150 &
0.000\\
$\beta_{3}$ & 0.007 & 0.081 & 0.080 & 0.946 &  & 0.605 & 0.126 & 0.123 &
0.001\\\hline\hline
\end{tabular}

\caption{Finite sample performances of the MLE and QMLE}\label{Table: finite sample MLE QMLE}%
\end{table}%

Next, we consider testing IIA using the case when $X_{ij,\ell}\thicksim
iidN(0,\pi^{2}/36),\ell=1,2,3$ and $n=500$ as an example. We follow the
standard procedure to perform the Hausman-McFadden test (cf.,
\citet{HausmanMcFadden1984}%
). For each simulation replication, we first estimate the model under the
(incorrect) LEVI specification using the full sample. We then carry out the
same estimation procedure but use only the subsample of individuals who choose
either alternative 1 or 2. We compare these two estimates using the Wald
statistic based on the 95\% critical value from the chi-squared distribution
with three degrees of freedom ($\chi^{2}(3)$). The empirical rates of
rejection over 5000 simulation replications for $J=5,8,$ and $11$ are 30.30\%,
40.18\%, and 47.52\%, respectively. This indicates a fairly good chance of the
IIA being rejected. Following an IIA rejection, the standard practice is to
employ discrete choice models that do not exhibit the IIA property, such as
the mixed logit. However, in the present setting, such a practice would lead
to a wrong characterization of agents' preferences.

We note that the LEVI model is misspecified for both the full sample and the
subsample of the individuals who choose either alternative 1 or 2. The
misspecification of the LEVI model for the full sample is obvious as the
stochastic utility follows the SEVI distribution rather than the LEVI
distribution. For the subsample with the two alternatives selected by the
econometrician, the LEVI model is still misspecified. From the point of view
of the econometrician, the conditional probability of choosing alternative 1
given that either alternative 1 or alternative 2 is chosen is%
\begin{equation}
\frac{\Pr(Y=1|\mathcal{J},V=v,\varepsilon_{j}\thicksim iid\text{
}\mathrm{SEVI}\text{ over }j\in\mathcal{J})}{\Pr(Y=1|\mathcal{J}%
,V=v,\varepsilon_{j}\thicksim iid\text{ }\mathrm{SEVI}\text{ over }%
j\in\mathcal{J})+\Pr(Y=2|\mathcal{J},V=v,\varepsilon_{j}\thicksim iid\text{
}\mathrm{SEVI}\text{ over }j\in\mathcal{J})}. \label{Generalized_GM}%
\end{equation}
This is not equal to $\exp(v_{1})/[\exp(v_{1})+\exp(v_{2})]$ unless
$\mathcal{J}=\left\{  1,2\right\}  .$ However, the Hausman--McFadden test
requires at least three alternatives in total, as otherwise the test statistic
equals zero and the IIA null is never rejected. So, whenever the
Hausman-McFadden test is meaningful, the LEVI model is misspecified for any
subsample with at least two alternatives.

Finally, we consider using Vuong's test and AIC/BIC to decide which model fits
the data better. When the SEVI is the correct model, $X_{ij,\ell}\thicksim
iidN(0,\pi^{2}\omega_{j}^{2}/36)$ and $n=500$, the fractions of the test
statistic $\mathcal{V}_{n}$ less than $-1.645$ across 5000 simulation
replications are 47.08\%, 69.08, and 80.82\%, respectively, for $J=5,8,$ and
$11$. The corresponding fractions of the test statistic less than $0$ are
92.56\%, 97.86\%, and 99.32\%, respectively. These results underscore the
reasonable power of Vuong's test and the high accuracy of the AIC/BIC in
selecting the correct SEVI model. Results are similar when $X_{ij,\ell
}\thicksim iidN(0,\pi^{2}/36)$ and for other sample sizes. On the other hand,
when the LEVI is the correct model, $X_{ij,\ell}\thicksim iidN(0,\pi^{2}%
\omega_{j}^{2}/36)$ and $n=500$, we observe that the test statistic
$\mathcal{V}_{n}$ exceeds 1.645 in 30.58\%, 48.22\%, and 60.72\% of the cases
across 5000 simulation replications for $J=5,8,$ and $11$, respectively. The
corresponding percentages for $\mathcal{V}_{n}$ greater than zero are 87.56\%,
94.94\%, and 97.18\% for $J=5,8,$ and $11$, respectively. This affirms the
power of Vuong's test and the effectiveness of using the AIC/BIC in selecting
the correct LEVI model. Therefore, regardless of whether the SEVI or LEVI
model is the true model, Vuong's test has good power and the AIC and BIC
perform very well in correctly identifying the true model.

\subsection{Computation Time Comparison\label{Sec: time}}

In this subsection, we compare the computation times required for estimating
the LEVI and SEVI models. Additionally, we include the multinomial probit
model with iid normal errors with variance $\pi^{2}/6$ in our comparison. For
the sake of uniformity with the nomenclature used for the other two models, we
will henceforth refer to this restricted multinomial probit model as the NORM model.

The data are generated according to $Y_{i}=\max_{j=1}^{J}\left\{  X_{ij}%
\beta+\varepsilon_{ij}\right\}  $ where $X_{ij,\ell}\thicksim iidN(0,\pi
^{2}\omega_{j}^{2}/36)$ over $\ell=1,\ldots,L$ (cf., equation (\ref{omega1}))
and $\varepsilon_{ij}\thicksim iid$ $\mathrm{SEVI}$ over $j.$ We fit the LEVI,
SEVI, NORM models to the so-generated data. All models have the same correct
linear specification $V_{ij}\left(  \beta\right)  =X_{ij}\beta$ for the
systematic utility but differ with respect to the distributional specification
of the random components.

We use a sample size of 1000, explore different numbers of alternatives $(J)$
ranging from 2 to 16, and consider two different numbers of attributes $(L=3$
and $L=10)$.\ To mimic real data analysis, we simulate only one sample with
model parameter $\beta_{0}$ drawn randomly from a multivariate standard normal distribution.

We document the time needed to obtain the ML/QML estimate of the model
parameters for each of the three models. The time represents the seconds spent
by our MATLAB program to find each point estimate. All computations are
performed on a desktop PC equipped with an Intel i7-13700F processor and 32GB
of RAM. While no parallel method is employed,\footnote{As noted earlier, this
represents one potential direction to explore for improving the computational
speed of the SEVI model for larger values of $J.$} we strive to vectorize all
operations as much as possible. For the SEVI model, the choice probabilities
are computed using the formula in Theorem \ref{Theorem: choice_prob_under_SEV}%
. In the case of the NORM model, the choice probabilities are obtained using
the Geweke--Hajivassiliou--Keane (GHK) simulator with 500 simulations (%
\citet{Geweke2001}%
;
\citet{Hajivassiliou1998}%
;
\citet{Keane1994}%
).
\begin{figure}[ptb]%
\centering
\includegraphics[
height=2.5036in,
width=4.9926in
]%
{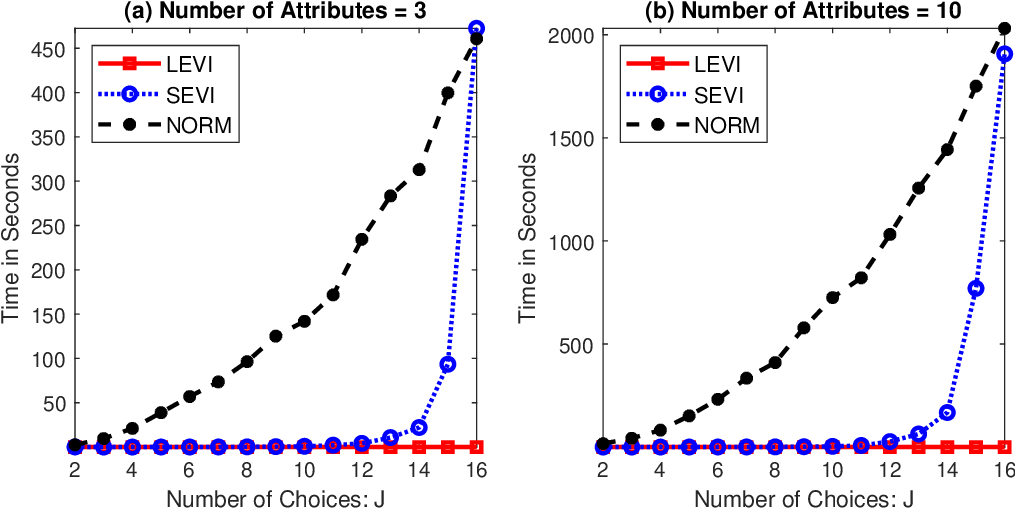}%
\caption{Comparison of computation times for the LEVI, SEVI, and NORM (MNP
with $\varepsilon_{ij}\thicksim iidN(0,\pi^{2}/6))$ models when $V_{ij}%
=X_{ij}\beta$}%
\label{non_asc_sevi_levi_norm_time_n1000}%
\end{figure}

Figure \ref{non_asc_sevi_levi_norm_time_n1000} plots the computation time
against the number of alternatives. On one hand, when $J\leq12,$ the
computation time for the SEVI model is comparable to that for the LEVI model.
When $J>12$, the computation time for the SEVI model starts to increase with
$J$ and becomes significantly larger than that for the LEVI model, especially
when $J$ reaches $16.$ On the other hand, when $J\leq12,$ the NORM model
exhibits a clear computational disadvantage compared to the LEVI and
SEVI\ models. While the computation time for each of the latter two models is
less than 27 seconds when the number of attributes is 10 and the number of
alternatives is 12, the corresponding computation time for the NORM model is
more than 1000 seconds. In our experiment, the SEVI model starts to become
computationally as costly as the NORM model only when $J=16.$

While the first DGP contains no alternative specific constant, the second
group of DGPs we consider does. We examine two scenarios: one with only
alternative-specific constants where $V_{ij}=Z_{i}\beta_{zj}$ for $Z_{i}=1$
(hereafter, ASC-only) and another with alternative-specific constants $\left(
Z_{i}\right)  $ and three attributes $X_{ij,\ell}\thicksim iidN(0,\pi
^{2}\omega_{j}^{2}/36),\ell=1,2,3$. In the latter case, $V_{ij}=Z_{i}%
\beta_{zj}+X_{ij,1}\beta_{x1}+X_{ij,2}\beta_{x2}+X_{ij,3}\beta_{x3}$ where
$Z_{i}=1.$ In both cases, $\varepsilon_{ij}\thicksim iid$ $\mathrm{SEVI}$ over
$j=1,\ldots,J.$ For each dataset, we fit the LEVI,\ SEVI, and NORM models with
correctly specified system utility. Figure \ref{asc_sevi_levi_norm_time_n1000}
reports the computation times. The general patterns observed in Figure
\ref{non_asc_sevi_levi_norm_time_n1000} carry over to this figure. When
$J\leq12,$ the SEVI and LEVI models require less than 65 seconds for
estimation, with much shorter times for a smaller $J.$ In contrast, the time
needed for the NORM model grows rapidly as $J$ increases from 2 to 12. For
example, when $J=10$ and $12,$ the required times for the NORM model for the
ASC-only case are 2274 and 5159 seconds, respectively. In other words,
estimating a standard multinomial probit model takes between half an hour and
two hours when $J$ is about 10 or 12, whereas the corresponding time for a
SEVI model is just about 1 minute. As before, the SEVI model becomes
computationally less competitive for $J>12$ and has a computational cost
starting to approach the NORM model when $J$ reaches 16.
\begin{figure}[ptb]%
\centering
\includegraphics[
height=2.3739in,
width=5.0522in
]%
{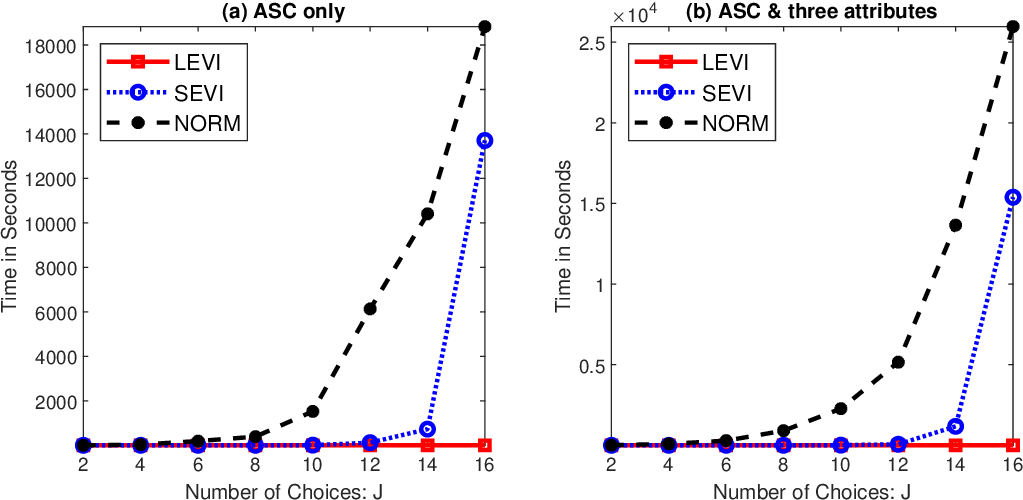}%
\caption{Comparison of computation times for the LEVI, SEVI, and NORM
(standard MNP) models when $V_{ij}=Z_{i}\beta_{zj}$ or $V_{ij}=Z_{i}\beta
_{zj}+X_{ij}\beta_{x}$}%
\label{asc_sevi_levi_norm_time_n1000}%
\end{figure}

\section{Some Extensions \label{Sec: extension}}

\subsection{Mixed LEVI-SEVI Model\label{Subsection: mixed LEVI-SEVI}}

One way to understand the random components in an RUM is to conceptualize them
as capturing private signals exclusive to the individual agent making the
choice, yet undisclosed to the econometrician (cf.,
\citet{Manski1977}%
). These signals may exhibit left or right skewness, varying across
individuals. This variability suggests the possibility of two distinct groups
of individuals: one characterized by the LEVI type and the other by the SEVI
type. While the group membership of a specific individual remains
unobservable, it is still useful to estimate the relative sizes of these two
groups for the purpose of prediction and policy analysis.

Assume that individuals in the population have identical preference parameters
but may differ in terms of random utility components. We postulate that with a
probability of $\rho_{0},$ $\varepsilon_{ij}\thicksim iid$ $\mathrm{SEVI}$
over $j=1,\ldots,J$ and with a probability of $1-\rho_{0},$ $\varepsilon
_{ij}\thicksim iid$ $\mathrm{LEVI}$ over $j=1,\ldots,J.$ In such a mixed
LEVI-SEVI model, the probability of a randomly chosen individual, say
individual $i,$ selecting alternative $j,$ given the vector of systematic
utilities $v_{i},$ is%
\[
\Pr\left(  Y_{i}=j|V_{i}=v_{i}\right)  =\rho_{0}P^{\mathrm{SEVI}}\left(
Y_{i}=j|V_{i}=v_{i}\right)  +\left(  1-\rho_{0}\right)  P^{\mathrm{LEVI}%
}\left(  Y_{i}=j|V_{i}=v_{i}\right)  .
\]
Here, $\rho_{0}$ and $\left(  1-\rho_{0}\right)  $ can also be interpreted as
the proportions of individuals belonging to the SEVI and LEVI types,
respectively, in the population of interest.

This new latent class mixed LEVI-SEVI model (mixed LS model hereafter) differs
from a usual latent class logit model where the classes have different
preference parameters but the same random component. In the mixed LS model,
identical preference parameters are assumed, but different random components
are considered. Hence, it can be regarded as a mirror image of the usual
latent class logit model.

While the identification, estimation, and inference of the mixed LS model
warrant rigorous theoretical analysis, we present some simulation evidence
suggesting that the key mixing parameter $\rho_{0}$ can be reliably estimated.
We consider a parametrization of the systematic utilities as before where
$X_{ij,\ell}\thicksim iidN(0,\pi^{2}\omega_{j}^{2}/36)$ for $\ell=1,2,3$ and
$\beta_{0}=\left(  1,2,1\right)  ^{\prime}$. We let the sample size be 5000
and consider $J=3,5,7,9.$ We estimate the model by MLE. Figure
\ref{Fig: boxplot_rho_hat_vs_rho0_no_outlier_n5000} presents a standard
boxplot of the MLEs over 1000 simulation replications against the true value
$\rho_{0}=0,0.1,0.2,\ldots,1.$ The diamond and bar markers inside each box
indicate, respectively, the average and median of the MLEs for a given true
value $\rho_{0},$ and the 45-degree line represents the true target
line.\ Both the average and the median of the estimates are fairly close to
the true value. The interquartile range (the difference between the 75th
percentile and the 25th percentile) indicates some uncertainty in the
estimator, particularly when $J$ is smaller, but this uncertainty diminishes
as $J$ increases. Figures not reported here (see, for example, Figure
\ref{Fig: boxplot_bhat1_vs_rho0_no_outlier_n5000}) show that the preference
parameters can also be reliably estimated. Overall, our simulation results
suggest that, with a reasonably large sample size, we can extract useful
information on the mixing parameter and preference parameters. Consequently, a
mixed LS model is an interesting alternative in empirical applications,
especially given that it includes a SEVI model $(\rho_{0}=1)$ and a LEVI model
$(\rho_{0}=0)$ as special cases. A natural extension is to specify the mixing
parameter as a function of observable covariates that may influence whether an
agent's random component is LEVI or SEVI.
\begin{figure}[h]%
\centering
\includegraphics[
trim=0.286507in 0.317730in 0.921146in 0.384273in,
height=3.6841in,
width=4.7202in
]%
{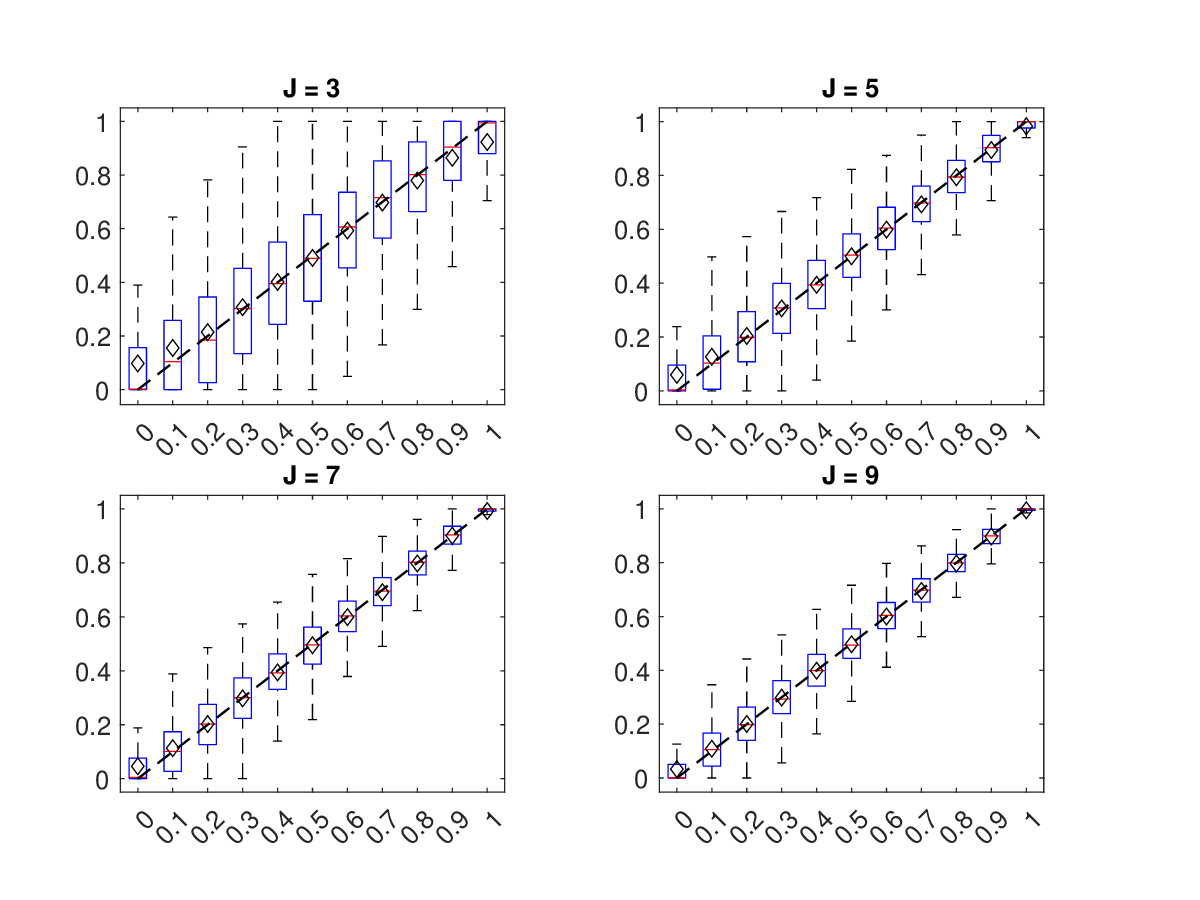}%
\caption{Boxplot of $\hat{\rho}$ for each true value $\rho_{0}=0,0.1,\ldots,1$
with sample size $5000$ and $J=3,5,7,9$}%
\label{Fig: boxplot_rho_hat_vs_rho0_no_outlier_n5000}%
\end{figure}

It is interesting to see how the QMLE under the usual conditional logit
behaves when the true DGP is a mixed LS model. Figure
\ref{Fig: bhat_levi_vs_rho0_n5000} plots the average of $\hat{\beta
}_{\mathrm{LEVI}}$ across the simulation replications against the true value
of $\rho_{0}.$ When $\rho_{0}=0,$ the LEVI model is correctly specified,
$\hat{\beta}_{\mathrm{LEVI}}$ is close to the true value $\beta_{0}$ in an
average sense. When $\rho_{0}$ increases, signaling a higher misspecification
of the LEVI model, the bias of $\hat{\beta}_{\mathrm{LEVI}}$ increases. For a
given $\rho_{0}>0,$ the bias of $\hat{\beta}_{\mathrm{LEVI}}$ is larger when
there are more alternatives. When $\rho_{0}=1,$ the DGP reduces to a SEVI
model, and the figure aligns with the simulation results reported earlier.
\begin{figure}[h]%
\centering
\includegraphics[
height=3.7913in,
width=4.6769in
]%
{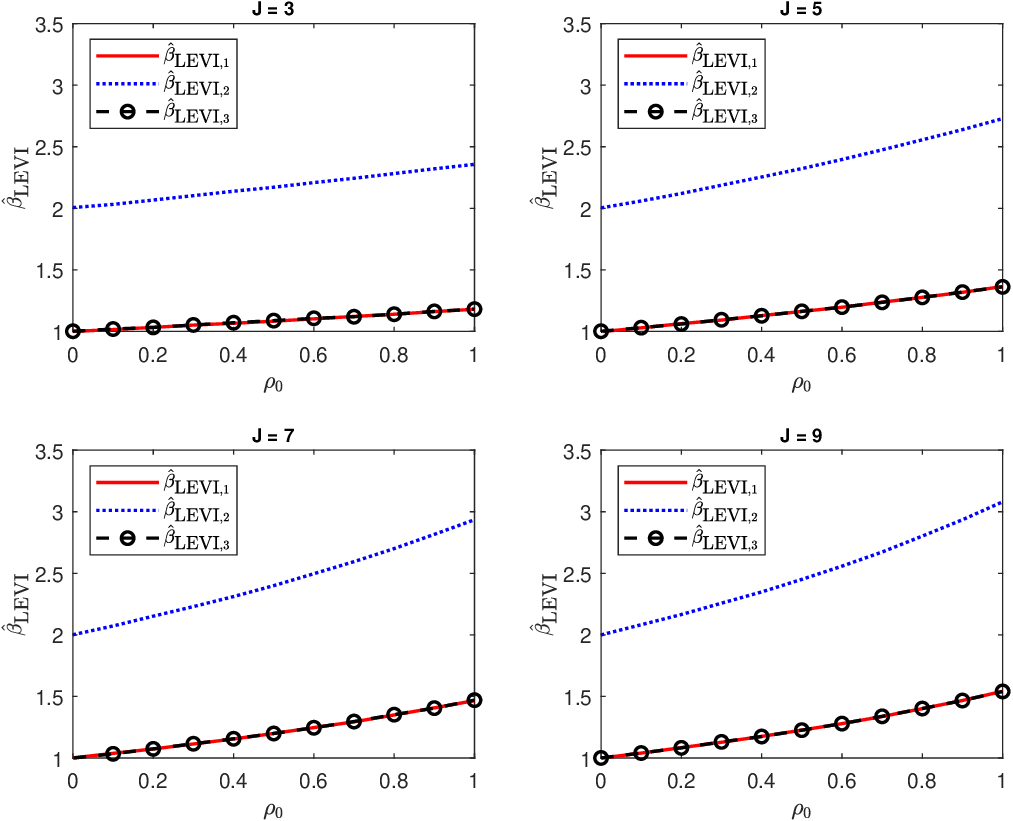}%
\caption{Plot of the average of $\hat{\beta}_{\mathrm{LEVI}}$ against
$\rho_{0}$ with sample size 5000 and $J=3,5,7,9$}%
\label{Fig: bhat_levi_vs_rho0_n5000}%
\end{figure}

Finally, we evaluate the performance of a conventional mixed logit model when
applied to data generated by a mixed LS model where $X_{ij,\ell}\thicksim
iidN(0,\pi^{2}\omega_{j}^{2}/36)$ for $\ell=1,2,3$ and $\beta_{0}=\left(
1,2,1\right)  ^{\prime}$. We let $\rho_{0}=0,0.25,0.50,0.75,$ and $1$ in the
mixed LS model. We consider a large sample size of 50,000 and generate only
one sample. Our experiment can be regarded as a large sample experiment. The
preference parameters in the proposed mixed logit model are assumed to follow
independent normal distributions. More specifically, $U_{ij}=X_{ij,1}%
\beta_{i,1}+X_{ij,2}\beta_{i,2}+X_{ij,3}\beta_{i,3}+\varepsilon_{ij}$ where
$\varepsilon_{ij}$ is iid LEVI and $\beta_{i}:=\left(  \beta_{i,1},\beta
_{i,2},\beta_{i,3}\right)  ^{\prime}\thicksim\mathrm{iidN}(\beta
,\mathrm{diag}\left(  \sigma_{1}^{2},\sigma_{2}^{2},\sigma_{3}^{3}\right)  ).$
We refer to $\beta$ as the mean parameter and $\sigma=(\sigma_{1},\sigma
_{2},\sigma_{3})^{\prime}$ as the standard deviation parameter. After fitting
this mixed logit to the data generated by the mixed LS model, we consider
three types of tests. First, we test whether $\beta=\beta_{0}$ using a
standard Wald test. Second, we use the likelihood-ratio (LR) test for testing
the null hypothesis of a logit model with fixed taste parameters (i.e.,
$\sigma^{2}=0)$ against the mixed logit. Finally, we use a generalized
Hausman--McFadden (HM) test (%
\citet{Hahn2020}%
) for testing the null hypothesis that the mixed logit model is correctly
specified. As in
\citet{HausmanMcFadden1984}%
, we first estimate the mixed logit model using the full sample and the
subsample obtained by excluding individuals who chose the last alternative. We
then compare these two estimates using a chi-square test.\footnote{Similar to
(\ref{Generalized_GM}), the choice probabilities with the last alternative
removed represent the probabilities as perceived by the econometrician. They
are equal to the conditional choice probabilities, given that the last
alternative was not chosen. Importantly, they do not correspond to the choice
probabilities when individuals are presented with the first $J-1$
alternatives. Additional discussions on this point can be found in\
\citet{Hahn2020}%
.} The comparison for the generalized HM test can be based on either the mean
parameters or both the mean and standard deviation parameters. The generalized
HM test targets the IIA property of mixed logit models at the individual
level, and can be regarded as an IIA-oriented specification test that assesses
whether the random utility components are iid LEVI. We set the nominal level
for all three tests at 5\%.

We provide a brief discussion of the simulation results. First, for all $J$
values $(3,5,7,9)$ considered, the mean parameter estimator has a positive
bias, and the bias increases as $\rho_{0}$ or $J$ increases. The Wald test
rejects the null that the mean parameter is equal to $\beta_{0}$ when
$\rho_{0}\geq0.25$ for all $J$ values. This suggests that a mixed logit model
fitted to the data generated by a mixed LS model can not reliably recover the
mean preference parameters.

Second, the LR test fails to reject the null hypothesis of fixed taste
parameters when either $\rho_{0}$ smaller than or equal to $0.5$ or when $J$
$\leq5$. In these cases, the mixed logit does not significantly improve the
fit relative to the fixed-parameter logit model, with the standard deviation
parameters typically being insignificant. This suggests that the normally
distributed taste parameters are not able to pick up the SEVI random
component. On the other hand, the LR test rejects the null of fixed taste
parameters when both $\rho_{0}$ and $J$ are large. In these cases, a large
$\rho_{0}$ $(\rho_{0}\geq0.75)$ and a large $J$ $\left(  J\geq7\right)  $
prompt the LR test to mistake heterogeneity in the random utility components
for heterogeneity in the taste parameters. While the LR test tends to be
conservative as the variance parameters are on the boundary under the null,
our simulation results highlight that its use can lead to the incorrect
conclusion of taste heterogeneity when the taste parameters are the same for
all individuals.

Finally, the generalized HM test based on the mean parameters rejects the null
of iid LEVI random utility components when $\rho_{0}\geq0.75$ and $J\geq7$. On
the other hand, the generalized HM test based on both the mean and standard
deviation parameters rejects the null of iid LEVI random utility components
for almost all $\rho_{0}$ and $J$ values considered.

These findings generally remain valid with minimal exceptions when the random
preference parameters follow a normal distribution with an unrestricted
variance matrix.

In summary, a mixed logit model with correlated or uncorrelated preference
parameters will not be far off when $\rho_{0}$ is small, but it may yield
misleading preference estimates even with a moderately large $\rho_{0}$.
Notably, the mixed logit model can mistake heterogeneity in the random utility
components for taste heterogeneity, particularly in situations with a large
$\rho_{0}$ and a large $J$. Consequently, the mixed logit model is not a
reliable substitute for the mixed LS model.

\subsection{Cost Minimizing Choice\label{Subsection: Cost}}

In some empirical applications, economic agents may choose one alternative
from a set of alternatives to minimize a specific objective function, such as
cost, regret, disutility, or loss, rather than maximizing it. For instance, in
the study by
\citet{Fowlie2010}%
, which we will examine in the next section, plant managers choose an
environmental compliance strategy in order to minimize their overall cost.

To model the minimization behavior, we assume that there is a latent function
$C_{ij}$ that encompasses an observable component $D_{ij}\left(  \beta
_{0}\right)  $ and an unobservable component $\varepsilon_{ij},$ so that%
\[
C_{ij}\left(  \beta_{0}\right)  =D_{ij}\left(  \beta_{0}\right)
+\varepsilon_{ij}\text{ for }j=1,\ldots,J\text{ }%
\]
where $C_{ij}\left(  \beta_{0}\right)  $ is the `cost' that individual $i$
incurs from choosing alternative $j.$ Individuals then make a choice that
results in the smallest $C$ value. More specifically, the observed choice of
individual $i$, denoted by $Y_{i},$ equals
\[
Y_{i}=\arg\min_{j=1,\ldots,J}C_{ij}\left(  \beta_{0}\right)  .
\]

Assuming that $\left\{  D_{ij}\left(  \beta_{0}\right)  \right\}  _{j=1}^{J}$
and $\left\{  \varepsilon_{ij}\right\}  _{j=1}^{J}$ are independent of each
other, we consider two cases: either $\varepsilon_{ij}\thicksim^{iid}%
\mathrm{LEVI}$ or $\varepsilon_{ij}\thicksim^{iid}\mathrm{SEVI.}$ Let
$D_{i}\left(  \beta_{0}\right)  =(D_{i1}\left(  \beta_{0}\right)
,\ldots,D_{iJ}\left(  \beta_{0}\right)  )^{\prime}\in\mathbb{R}^{J}$ and
$d_{i}\in\mathbb{R}^{J}$ be a point in the support of $D_{i}\left(  \beta
_{0}\right)  .$ In the former case, we have
\begin{align*}
\Pr(Y_{i}  &  =j|D_{i}\left(  \beta_{0}\right)  =d_{i},\varepsilon
_{i,m}\overset{iid}{\thicksim}\mathrm{LEVI})\\
&  =\Pr\left(  d_{ij}+\varepsilon_{i,j}\leq d_{ik}+\varepsilon_{i,k}\text{,
}k\in\mathcal{J},\varepsilon_{i,m}\overset{iid}{\thicksim}\mathrm{LEVI}\right)
\\
&  =\Pr\left(  -d_{ij}-\varepsilon_{i,j}\geq d_{ik}-\varepsilon_{i,k},\text{
}k\in\mathcal{J},\text{ }\varepsilon_{i,m}\overset{iid}{\thicksim
}\mathrm{LEVI}\right) \\
&  =\Pr\left(  -d_{ij}+\varepsilon_{i,j}\geq-d_{ik}+\varepsilon_{i,k},\text{
}k\in\mathcal{J},\varepsilon_{i,m}\overset{iid}{\thicksim}\mathrm{SEVI}\right)
\\
&  =P^{\mathrm{SEVI}}\left(  Y_{i}=j|-d_{i}\right)  ,
\end{align*}
where $P^{\mathrm{SEVI}}\left(  Y_{i}=j|-d_{i}\right)  $ is defined in
(\ref{SEVI_prob_formula}). In the latter case, we have, using the same
argument,
\[
\Pr(Y_{i}=j|D_{i}\left(  \beta_{0}\right)  =d_{i},\varepsilon_{i,m}%
\overset{iid}{\thicksim}\mathrm{SEVI})=P^{\mathrm{LEVI}}\left(  Y_{i}%
=j|-d_{i}\right)  ,
\]
where $P^{\mathrm{LEVI}}\left(  Y_{i}=j|-d_{i}\right)  $ is defined in
(\ref{LEVI_prob_formula}). So, for discrete choice problems aiming at
minimizing an objective function, when $\varepsilon_{i,j}%
\overset{iid}{\thicksim}\mathrm{LEVI,}$ the choice probability takes the form
of (\ref{SEVI_prob_formula}). In contrast, when $\varepsilon_{i,j}%
\overset{iid}{\thicksim}\mathrm{SEVI,}$ the choice probability\ takes the
usual logit form.

Note that there is a sign change. It is $-d_{i}$ rather than $d_{i}$ itself
that enters the expression of the choice probability derived under utility
maximization. If we parameterize $D_{ij}\left(  \beta_{0}\right)  =X_{ij}%
\beta_{0}$ for some covariate $X_{ij}$ and run a conditional logit using
$X_{ij}$ as the covariate in a statistical package, then we are assuming that
$\varepsilon_{i,j}\overset{iid}{\thicksim}\mathrm{SEVI}$, and what the package
returns is an estimate of $-\beta_{0}.$ To correct this, we can flip the sign
of the estimate or use the flipped covariate $-X_{ij}$ as the covariate.

If we deem that $\varepsilon_{i,j}\overset{iid}{\thicksim}\mathrm{LEVI}$ is
more plausible, then the choice probability formula in
(\ref{SEVI_prob_formula}) has to be used, and the approach proposed in this
paper can be easily applied.

Irrespective of whether we are dealing with a minimization or maximization
choice problem, we call a model with SEVI errors a SEVI model and a model with
LEVI errors a LEVI model. In the case of maximization problems, estimating a
LEVI model is more straightforward, while for minimization problems, a SEVI
model is easier to estimate.

\section{Empirical Applications\label{Sec:Emp}}

In this section, we examine four commonly used datasets from well-cited
textbooks and papers, all of which have been included in the
\textsf{mlogit}
package in
\textsf{R}
(cf.,
\citet{Croissant2020}%
). For each dataset, we apply three distinct models: the LEVI model, the SEVI
model, and the NORM model (i.e., a multinomial probit model with iid $N\left(
0,\pi^{2}/6\right)  $ errors). The variances of random utility components of
all three models are the same and equal to $\pi^{2}/6.$ The sole difference
among these three models lies in the shape of the distributions of the random
utility components.\footnote{Note that we are conducting a comparative
analysis solely among the three basic models, and we do not claim that any of
these models is consistent with the true underlying DGP.}

The four datasets we analyze are:

\begin{enumerate}
\item Fishing mode choice: the dataset comes from\ (%
\citet{Herriges_Kling1999}%
), which has been analyzed in several textbooks, including
\citet{ cameron_trivedi_2005}%
,
\citet{cameron_trivedi_2022a}%
, and
\citet{ cameron_trivedi_2022b}%
.

\item Stated preference vehicle choice: the dataset originates from
\citet{BROWNSTONE1996}
and is subsequently used in
\citet{BROWNSTONE1998}
and
\citet{McFadden_Train2000}%
.


\item Choice of brand for saltine crackers: the dataset originates from
\citet{Paap_Franses2000}%
, and a subset of the data was analyzed in
\citet{Jain_etc1994}%
.


\item Choice of NO$_{x}$ emissions compliance strategies: the dataset comes
from
\citet{Fowlie2010}%
.

\end{enumerate}

Table \ref{Table: empirics1} presents the computation times and log-likelihood
performances of the three models. In all cases, the computational time
represents the seconds spent by our MATLAB program on finding the MLE/QMLE.
Although there may be slight fluctuations in the reported computation times
due to concurrent background processes, they serve as a reliable indicator of
the time needed for estimating each model. We will revisit this table when we
discuss each application.%

\begin{table}[tbp] \centering
\caption{Performances of LEVI, SEVI, and NORM on several data sets}%

\begin{tabular}
[c]{c|ccc|ccc|ccc|c}\hline\hline
{\small Data Set} & ${\small n}$ & ${\small J}$ & ${\small L}$ &
\multicolumn{3}{c|}{{\small Computation time}} &
\multicolumn{3}{c|}{{\small Log-likelihood}} & {\small Likelihood}\\\hline
&  &  &  & {\small LEVI} & {\small SEVI} & {\small NORM} & {\small LEVI} &
{\small SEVI} & {\small NORM} & {\small Ranking}\\\hline
\multicolumn{1}{l|}{{\small Fishing}} & \multicolumn{1}{|r}{{\small 1182}} &
\multicolumn{1}{r}{{\small 4}} & \multicolumn{1}{r|}{{\small 8}} &
\multicolumn{1}{|r}{{\small 0.35}} & \multicolumn{1}{r}{{\small 0.58}} &
\multicolumn{1}{r|}{{\small 370}} & \multicolumn{1}{|r}{{\small -1215.14}} &
\multicolumn{1}{r}{{\small -1213.21}} & \multicolumn{1}{r|}{{\small -1218.93}}
& {\small S
$>$
L
$>$
N}\\
\multicolumn{1}{l|}{{\small Vehicles}} & \multicolumn{1}{|r}{{\small 4654}} &
\multicolumn{1}{r}{{\small 6}} & \multicolumn{1}{r|}{{\small 21}} &
\multicolumn{1}{|r}{{\small 3.57}} & \multicolumn{1}{r}{{\small 14.00}} &
\multicolumn{1}{r|}{{\small 2232}} & \multicolumn{1}{|r}{{\small -7394.62}} &
\multicolumn{1}{r}{{\small -7388.75}} & \multicolumn{1}{r|}{{\small -7389.50}}
& {\small S
$>$
N
$>$
L}\\
\multicolumn{1}{l|}{{\small Crackers}} & \multicolumn{1}{|r}{{\small 3289}} &
\multicolumn{1}{r}{{\small 4}} & \multicolumn{1}{r|}{{\small 6}} &
\multicolumn{1}{|r}{{\small 0.31}} & \multicolumn{1}{r}{{\small 0.55}} &
\multicolumn{1}{r|}{{\small 157}} & \multicolumn{1}{|r}{{\small -3347.61}} &
\multicolumn{1}{r}{{\small -3347.13}} & \multicolumn{1}{r|}{{\small -3344.51}}
& {\small N
$>$
S
$>$
L}\\
\multicolumn{1}{l|}{{\small NO}$_{x}${\small : dereg}} &
\multicolumn{1}{|r}{{\small 227}} & \multicolumn{1}{r}{{\small 15}} &
\multicolumn{1}{r|}{{\small 6}} & \multicolumn{1}{|r}{{\small 135.00}} &
\multicolumn{1}{r}{{\small 0.02}} & \multicolumn{1}{r|}{{\small 791}} &
\multicolumn{1}{|r}{{\small -345.35}} & \multicolumn{1}{r}{{\small -339.07}} &
\multicolumn{1}{r|}{{\small -343.21}} & {\small S
$>$
N
$>$
L}\\
\multicolumn{1}{l|}{{\small NO}$_{x}${\small : public}} &
\multicolumn{1}{|r}{{\small 113}} & \multicolumn{1}{r}{{\small 15}} &
\multicolumn{1}{r|}{{\small 6}} & \multicolumn{1}{|r}{{\small 78.00}} &
\multicolumn{1}{r}{{\small 0.01}} & \multicolumn{1}{r|}{{\small 2357}} &
\multicolumn{1}{|r}{{\small -86.30}} & \multicolumn{1}{r}{{\small -78.46}} &
\multicolumn{1}{r|}{{\small -82.38}} & {\small S
$>$
N
$>$
L}\\
\multicolumn{1}{l|}{{\small NO}$_{x}${\small : reg}} &
\multicolumn{1}{|r}{{\small 292}} & \multicolumn{1}{r}{{\small 15}} &
\multicolumn{1}{r|}{{\small 6}} & \multicolumn{1}{|r}{{\small 279.00}} &
\multicolumn{1}{r}{{\small 0.01}} & \multicolumn{1}{r|}{{\small 4600}} &
\multicolumn{1}{|r}{{\small -364.99}} & \multicolumn{1}{r}{{\small -359.74}} &
\multicolumn{1}{r|}{{\small -365.96}} & {\small S
$>$
L
$>$
N}\\\hline\hline
\multicolumn{11}{l}{{\small (i) }${\small n}$$\ ${\small is the \# of choice
situations, }${\small J}$ {\small is the (max) \# of alternatives, and
}${\small L}$ {\small is the total \# of covariates.}}\\
\multicolumn{11}{l}{{\small (ii) \textquotedblleft NO}$_{x}${\small :
dereg\textquotedblright\ is the subsample of deregulated plants;
\textquotedblleft NO}$_{x}${\small : public\textquotedblright\ is the
subsample of public plants,}}\\
\multicolumn{11}{l}{\ \ \ \ \ \ {\small and} {\small \textquotedblleft
NO}$_{x}${\small : reg\textquotedblright\ is the subsample of
\textquotedblleft regulated\textquotedblright\ plants.}}\\
\multicolumn{11}{l}{{\small (iii) All computation times are measured in
seconds. }}\\
\multicolumn{11}{l}{{\small (iv) In the likelihood ranking, S, L, and N are
short for SEVI, LEVI, and NORM, respectively.}}%
\end{tabular}
\label{Table: empirics1}%
\end{table}%

\subsection{Application: Choice of Fishing Mode}

For the fishing mode application, we model an individual's choice of fishing
modes among four mutually exclusive alternatives: beach $(j=1)$, pier $\left(
j=2\right)  $, private boat $(j=3)$, and charter boat $(j=4)$. We assume that
the latent utility is given by
\begin{align*}
U_{ij}  &  =X_{1j}\beta_{x1}+X_{2j}\beta_{x2}+D2_{j}\times\beta_{d2}%
+D3_{j}\times\beta_{d3}+D4_{j}\times\beta_{d4}\\
&  +\left(  Z_{i}\cdot D2_{j}\right)  \times\beta_{z2}+\left(  Z_{i}\cdot
D3_{j}\right)  \times\beta_{z3}+\left(  Z_{i}\cdot D4_{j}\right)  \times
\beta_{z4}+\varepsilon_{ij},
\end{align*}
where $X_{1j}$ and $X_{2j}$ are the price and catch rate of alternative $j$,
$D2_{j},D3_{j}$, and $D4_{j}$ are the dummy variables for alternatives 2, 3,
and 4, respectively, and $Z_{i}$ is the income of individual $i.$ For this
application, the LEVI and SEVI models yield similar parameter estimates and
likelihood function values. The correlation coefficient between the in-sample
predicted choice probabilities for the two models is 0.998, indicating a
nearly perfect positive correlation between them. Both models can be estimated
with minimal computational cost with the optimization algorithm completing in
less than a minute. Nevertheless, the SEVI model has a higher log-likelihood
than the LEVI model. In contrast, the multinomial probit incurs significantly
higher computational costs and yields a lower log-likelihood compared to both
the LEVI and SEVI models.

Figure \ref{figure: fishing_pdiff_vs_plevi} plots the difference between the
predicted choice probabilities (i.e., $P^{\mathrm{SEVI}}(\hat{\beta
}_{\mathrm{SEVI}})-P^{\mathrm{LEVI}}(\hat{\beta}_{\mathrm{LEVI}})$ or
$P^{\mathrm{NORM}}(\hat{\beta}_{\mathrm{NORM}})-P^{\mathrm{LEVI}}(\hat{\beta
}_{\mathrm{LEVI}}))$ against $P^{\mathrm{LEVI}}(\hat{\beta}_{\mathrm{LEVI}}).$
While all three predicted choice probabilities are very highly correlated, the
figure illustrates some subtle distinctions. Specifically, as
$P^{\mathrm{LEVI}}(\hat{\beta}_{\mathrm{LEVI}})$ increases, the variabilities
of $P^{\mathrm{SEVI}}(\hat{\beta}_{\mathrm{SEVI}})-P^{\mathrm{LEVI}}%
(\hat{\beta}_{\mathrm{LEVI}})$ and $P^{\mathrm{NORM}}(\hat{\beta
}_{\mathrm{NORM}})-P^{\mathrm{LEVI}}(\hat{\beta}_{\mathrm{LEVI}})$ also
increase, indicating that the three models tend to align closely when the
choice probabilities are low but exhibit discrepancies otherwise.%

\begin{figure}[ptb]%
\centering
\includegraphics[
height=2.2615in,
width=5.3575in
]%
{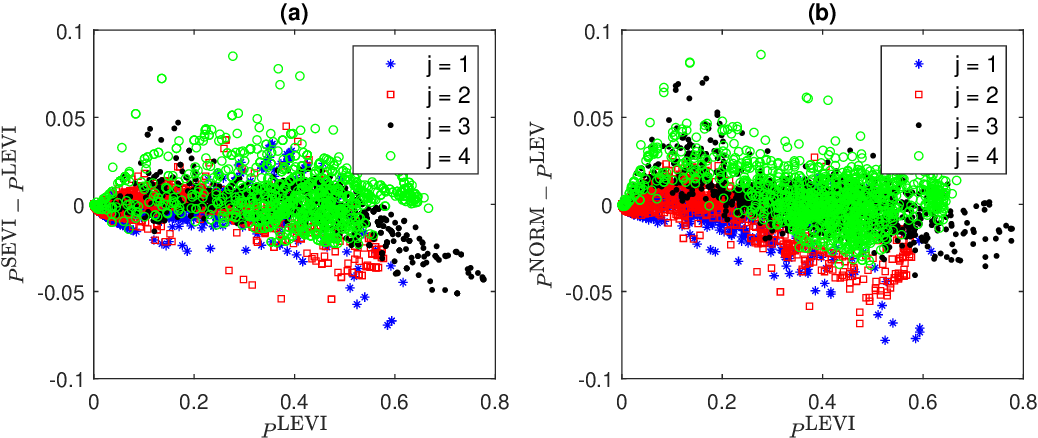}%
\caption{Plot of the difference of the predicted choice probabilities against
the LEVI predicted probability in the fishing mode choice application}%
\label{figure: fishing_pdiff_vs_plevi}%
\end{figure}

To illustrate the difference between the SEVI and LEVI models using the
fishing choice data, we calculate the choice probabilities $P^{\mathrm{LEVI}%
}(\hat{\beta}_{\mathrm{LEVI}})$ and $P^{\mathrm{SEVI}}(\hat{\beta
}_{\mathrm{LEVI}})$ using the same parameter value $\hat{\beta}_{\mathrm{LEVI}%
}.$ Figure \ref{Figure: Fishing_levi_sevi_w_levi_para} plots $P^{\mathrm{SEVI}%
}(\hat{\beta}_{\mathrm{LEVI}})$ against $P^{\mathrm{LEVI}}(\hat{\beta
}_{\mathrm{LEVI}})$. We find that when $P^{\mathrm{LEVI}}(\hat{\beta
}_{\mathrm{LEVI}})$ is small, $P^{\mathrm{SEVI}}(\hat{\beta}_{\mathrm{LEVI}})$
tends to be smaller than $P^{\mathrm{LEVI}}(\hat{\beta}_{\mathrm{LEVI}})$, and
when $P^{\mathrm{LEVI}}(\hat{\beta}_{\mathrm{LEVI}})$ is large,
$P^{\mathrm{SEVI}}(\hat{\beta}_{\mathrm{LEVI}})$ tends to be larger than
$P^{\mathrm{LEVI}}(\hat{\beta}_{\mathrm{LEVI}})$. This observation aligns with
the previously reported simulation results, as seen in Figure
\ref{Figure: Simu_theoretical_pro_DIFF2}, for example.

\subsection{Application: Stated Preferences for Vehicle Choice}

For the stated preference vehicle data, each individual was presented with 6
hypothetical vehicles fueled by methanol, compressed natural gas, electricity
or regular gas. These vehicles have various attributes, such as fuel type,
price, cost per mile, mileage range between refuelings/rechargings, vehicle
type (SUV, sports car, station wagon, truck, van, regular car), pollution
level, and speed, among others. We employ the same set of\ 21 covariates as
in
\citet{McFadden_Train2000}
(Table 1), some of which are the products of vehicle attributes and individual
characteristics, such as household size greater than 2 or not, college
education or not, and a daily commute of less than 5 miles or not.

As shown in Table \ref{Table: empirics1}, in terms of the log-likelihood
achieved, the LEVI model performs worse than the SEVI model and the NORM model.

Figures \ref{figure: car_pdiff_vs_plevi_pnorm} (a) and (b) plot the difference
in the predicted probabilities (i.e., $P^{\mathrm{SEVI}}(\hat{\beta
}_{\mathrm{SEVI}})-P^{\mathrm{LEVI}}(\hat{\beta}_{\mathrm{LEVI}})$ or
$P^{\mathrm{NORM}}(\hat{\beta}_{\mathrm{NORM}})-P^{\mathrm{LEVI}}(\hat{\beta
}_{\mathrm{LEVI}})$ against the LEVI predicted probability (i.e.,
$P^{\mathrm{LEVI}}(\hat{\beta}_{\mathrm{LEVI}})$. There are some intriguing
patterns. Figure \ref{figure: car_pdiff_vs_plevi_pnorm}(a) shows that
$P^{\mathrm{SEVI}}(\hat{\beta}_{\mathrm{SEVI}})$ tends to be higher than
$P^{\mathrm{LEVI}}(\hat{\beta}_{\mathrm{LEVI}})$ when the latter is low and
tends to be lower than $P^{\mathrm{LEVI}}(\hat{\beta}_{\mathrm{LEVI}})$ when
the latter is high.

Note that the difference $P^{\mathrm{SEVI}}(\hat{\beta}_{\mathrm{SEVI}%
})-P^{\mathrm{LEVI}}(\hat{\beta}_{\mathrm{LEVI}})$ can be decomposed as
\[
P^{\mathrm{SEVI}}(\hat{\beta}_{\mathrm{SEVI}})-P^{\mathrm{LEVI}}(\hat{\beta
}_{\mathrm{LEVI}})=\text{ }\underset{\text{I}}{\underbrace{P^{\mathrm{SEVI}%
}(\hat{\beta}_{\mathrm{SEVI}})-P^{\mathrm{SEVI}}(\hat{\beta}_{\mathrm{LEVI}}%
)}}+\underset{\text{II}}{\underbrace{P^{\mathrm{SEVI}}(\hat{\beta
}_{\mathrm{LEVI}})-P^{\mathrm{LEVI}}(\hat{\beta}_{\mathrm{LEVI}})}}.
\]
The first part (I) above is due to the difference in the parameter
estimates\ used ($\hat{\beta}_{\mathrm{SEVI}}$ vs. $\hat{\beta}_{\mathrm{LEVI}%
})$, and the second part (II) above is due to the difference in the predictive
models used (SEVI vs. LEVI). According to an unreported figure, the second
part exhibits the same pattern as in Figure
\ref{Figure: Simu_theoretical_pro_DIFF2}. So the pattern in Figure
\ref{figure: car_pdiff_vs_plevi_pnorm}(a) emerges because the effect of the
first part dominates that of the second part. Figure
\ref{figure: car_pdiff_vs_plevi_pnorm}(b) shows a similar pattern, although
$P^{\mathrm{NORM}}(\hat{\beta}_{\mathrm{NORM}})-P^{\mathrm{LEVI}}(\hat{\beta
}_{\mathrm{LEVI}})$ appears to be less variable than $P^{\mathrm{SEVI}}%
(\hat{\beta}_{\mathrm{SEVI}})-P^{\mathrm{LEVI}}(\hat{\beta}_{\mathrm{LEVI}}).$

Figures \ref{figure: car_pdiff_vs_plevi_pnorm} (c) and (d) are similar to
Figures \ref{figure: car_pdiff_vs_plevi_pnorm} (a) and (b) but use
$P^{\mathrm{NORM}}(\hat{\beta}_{\mathrm{NORM}})$ as the benchmark. According
to these two figures, $P^{\mathrm{SEVI}}(\hat{\beta}_{\mathrm{SEVI}%
})-P^{\mathrm{NORM}}(\hat{\beta}_{\mathrm{NORM}})$ and $P^{\mathrm{LEVI}}%
(\hat{\beta}_{\mathrm{LEVI}})-P^{\mathrm{NORM}}(\hat{\beta}_{\mathrm{NORM}})$
display opposite patterns. On one hand, $P^{\mathrm{SEVI}}(\hat{\beta
}_{\mathrm{SEVI}})$ tends to be higher than $P^{\mathrm{NORM}}(\hat{\beta
}_{\mathrm{NORM}})$ when the latter is low and tends to be lower than
$P^{\mathrm{NORM}}(\hat{\beta}_{\mathrm{NORM}})$ when the latter is high. On
the other hand, $P^{\mathrm{LEVI}}(\hat{\beta}_{\mathrm{LEVI}})$ tends to be
lower than $P^{\mathrm{NORM}}(\hat{\beta}_{\mathrm{NORM}})$ when the latter is
low and tends to be higher than $P^{\mathrm{NORM}}(\hat{\beta}_{\mathrm{NORM}%
})$ when the latter is high. Clearly, the decision to employ SEVI, NORM, and
LEVI random components can have a significant effect on the predicted choice probabilities.%

\begin{figure}[ptb]%
\centering
\includegraphics[
height=3.7671in,
width=5.6507in
]%
{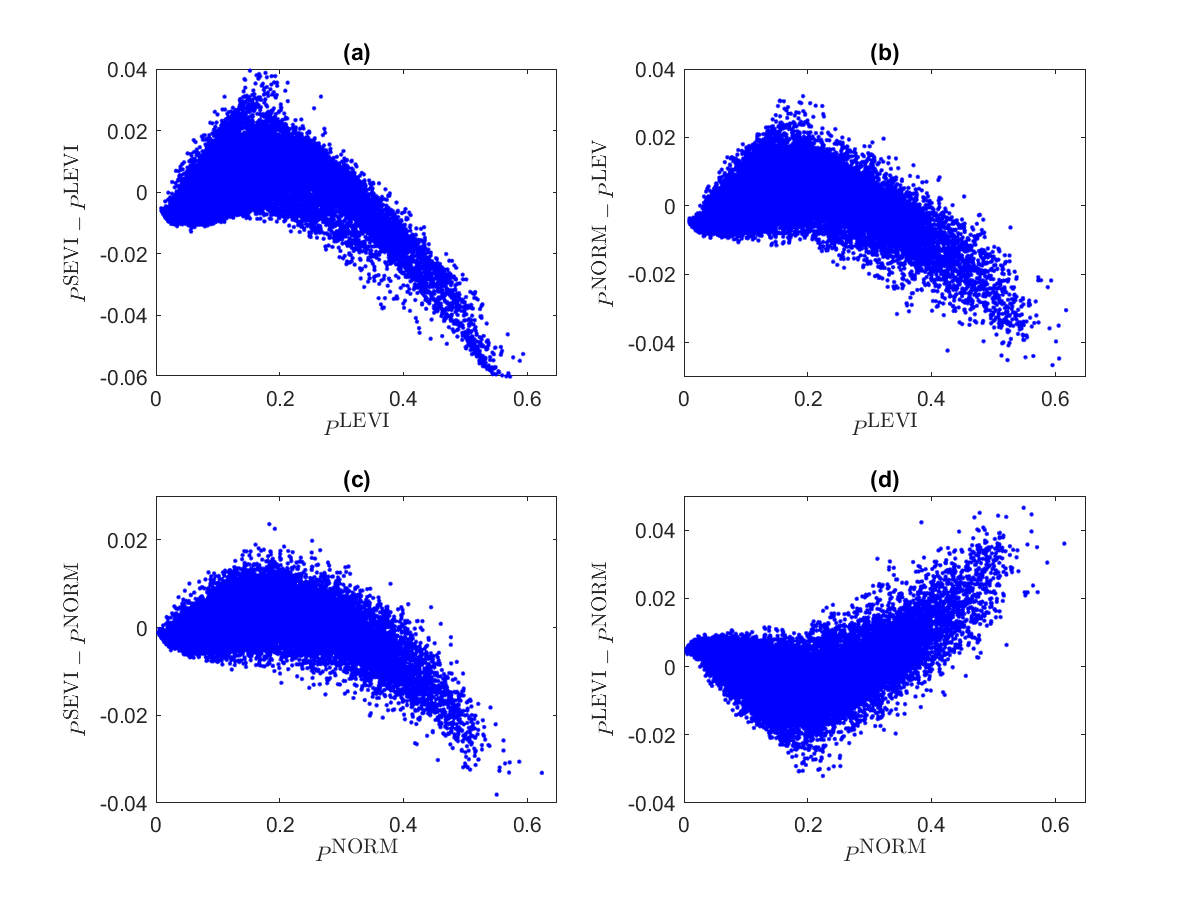}%
\caption{Plot of the difference of the predicted choice probabilities against
the LEVI or NORM predicted probability in the vehicle choice application}%
\label{figure: car_pdiff_vs_plevi_pnorm}%
\end{figure}

We also report the point estimates of the parameters under the LEVI and SEVI
model, along with their 95\% confidence intervals. See Figure
\ref{figure: car_b_levi_sevi_w_ci} in the supplementary appendix. While the
CI's for 14 parameters clearly overlap, those for $6$ parameters either do not
overlap or barely overlap. In practical terms, this divergence suggests that
the two models may lead to different conclusions regarding the impact of
certain attributes on vehicle choices. Specifically, attributes such as
mileage range between refuelings/rechargings, cost per mile, and vehicle types
exhibit different effects in the LEVI and SEVI models. Compared to the LEVI
model, the SEVI model suggests that the mileage range between
refuelings/rechargings and the vehicle's classification as an SUV (rather than
a regular car) are less influential factors. On the other hand, the cost per
mile and other vehicle types (rather than a regular car) appear to be more
influential in the SEVI model.

\subsection{Application: Choice of Saltine Crackers}

The dataset on saltine cracker choices contains information on all purchases
of crackers (3292) of 136 households, including the brand chosen from four
alternative brands, prices of all four brands (\texttt{Price}), whether there
was an in-store display of the brand (\texttt{Display} dummy) and/or whether
the brand was featured in the store's newspaper advertising (\texttt{Feature}
dummy) at the time of purchase. The latent utility is specified as follows:%
\begin{align*}
U_{ij}  &  =\mathtt{Price}\text{ }\times\beta_{x1}+\mathtt{Display}\text{
}\times\beta_{x2}+\mathtt{Feature}\text{ }\times\beta_{x3}\\
&  +D2_{j}\times\beta_{d2}+D3_{j}\times\beta_{d3}+D4_{j}\times\beta
_{d4}+\varepsilon_{ij},
\end{align*}
where $j=1,2,3,4$ corresponding to the brands \textquotedblleft
Sunshine\textquotedblright, \textquotedblleft Keebler\textquotedblright,
\textquotedblleft Nabisco\textquotedblright, \textquotedblleft
Private\textquotedblright\ (a private-label brand), respectively, and
$D2,\,D3,$ and $D4$ are dummies for the latter three brands (the dummy for the
first band has been omitted to achieve identification).

For this application, individuals made repeated purchase decisions. Choice
situations belonging to the same individual are possibly correlated. To
account for the correlation, we treat each individual as a cluster and employ
cluster-robust standard errors. An unreported figure similar to Figure
\ref{figure: car_b_levi_sevi_w_ci} shows that the estimated coefficients on
the three main covariates (\texttt{Price}, \texttt{Display}, and
\texttt{Feature}) are similar across the LEVI and SEVI models. However, the
alternative-specific coefficients on D2, D3, and D4 are different across these
two models. As shown in Table \ref{Table: empirics1}, the SEVI model has a
slightly higher log-likelihood compared to the LEVI model, although the
difference is small.

As shown in Table \ref{Table: empirics1}, the NORM model has a higher
log-likelihood than both the LEVI model and the SEVI model. This may indicate
that the random utility component is symmetrically distributed so that a
normal distribution, which is symmetric, might provide a more accurate
characterization than the LEVI and SEVI distributions. Figure
\ref{figure: cracker_pdiff_vs_pnorm} plots the difference in the predicted
choice probabilities using the NORM probability as the benchmark, highlighting
distinct deviations between the SEVI and LEVI probabilities from the NORM probability.%

\begin{figure}[ptb]%
\centering
\includegraphics[
height=2.6403in,
width=6.1808in
]%
{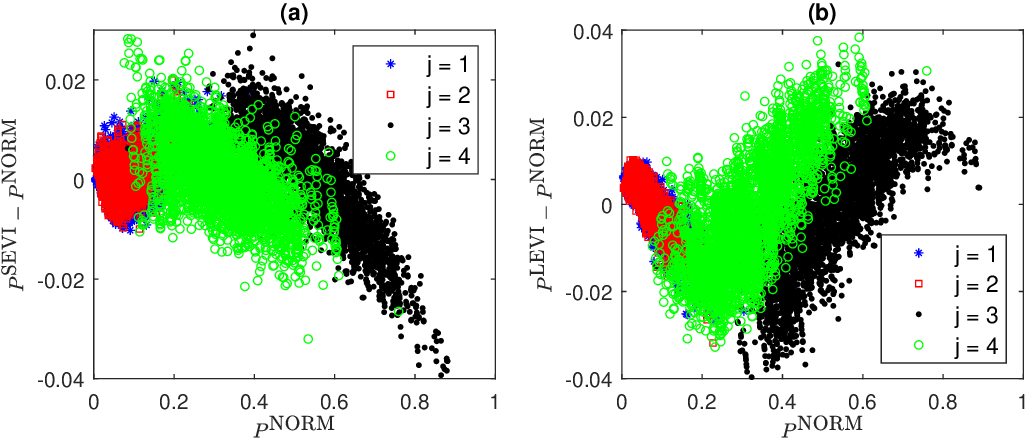}%
\caption{Plot of the difference of the predicted choice probabilities against
the multinomial probit predicted probability in the cracker choice
application}%
\label{figure: cracker_pdiff_vs_pnorm}%
\end{figure}

\subsection{Application: Choice of a NO$_{x}$ Emissions Reduction Technology}

Managers of coal-fired electric power plants chose a NO$_{x}$ emissions
reduction technology in order to comply with a regional NO$_{x}$ emissions
trading program. The manager of a plant unit, say unit $i,$ chooses a
compliance strategy from $J_{i}$ alternatives indexed by $j=1,\ldots,J_{i}$ to
minimize the unobserved overall cost $C_{ij}.$ We assume that
\[
C_{i,j}=X_{ij}\beta+\varepsilon_{i,j}%
\]
for $j=1,\ldots,J_{i}$ where $X_{ij}$ comprises observed factors that affect
the cost and $\varepsilon_{i,j}$ captures unobserved cost shifters. We assume
that these two components are independent of each other. The observed choice
is indicated by $Y_{i}=\arg\min_{j=1}^{J_{i}}C_{i,j}$. The model is
essentially the same as that outlined in Section \ref{Subsection: Cost} but
there are two differences. First, the maximum number of options is 15 but not
all options are available for each plant. To accommodate this, we simply
assign an infinite systematic cost to any unavailable option, ensuring it is
never selected. This may not be computationally efficient, but the number of
options becomes the same for all choice situations. In addition, this method
effectively provides us with an empirical example when the number of
alternatives is quite large. Second, some plants have the same manager so
there may be dependence in the unobserved costs across these plants. Following
standard practice, we will not consider the dependence when constructing the
likelihood function. That is, we will use a partial likelihood approach, but
we account for the dependence in estimating the asymptotic variance by
clustering the managers.

The compliance options have different capital costs (denoted by \texttt{kcost}%
$_{ij}$) and variable operating costs (denoted by \texttt{vcost}$_{ij}$)$.$
Each compliance option incorporates one or more of the following three
technologies: post-combustion pollution control technology (indicated by the
dummy \texttt{post}$_{j}$), combustion modification technology (indicated by
the dummy \texttt{cm}$_{j}$), and low NO$_{x}$ burners technology (indicated
by the dummy \texttt{lnb}$_{j}$). Taking these into consideration, we follow
\citet{Fowlie2010}
and let $X_{ij}=$ (\texttt{post}$_{j},$ \texttt{cm}$_{j},$ \texttt{lnb}$_{j},$
\texttt{kcost}$_{ij},$ \texttt{vcost}$_{ij},$ \texttt{kage}$_{ij}$) where
\texttt{kage }= \texttt{kcost} $\times$ \texttt{age} and \texttt{age} is the
age of a plant.

The plants are in one of the three regulatory environments: they are either
\textquotedblleft deregulated\textquotedblright, \textquotedblleft
public\textquotedblright,\ or \textquotedblleft regulated\textquotedblright.
The central question investigated by
\citet{Fowlie2010}
is whether and how regulatory environments influence a manager's choice of
compliance strategies. The question can be addressed by comparing the
coefficient estimates across the subsamples of \textquotedblleft
deregulated\textquotedblright, \textquotedblleft public\textquotedblright,\ or
\textquotedblleft regulated\textquotedblright\ plants.

We apply the LEVI, SEVI, and NORM models to each subsample and evaluate their
performances. Table \ref{Table: empirics1} shows that the SEVI model, which
assumes that $\varepsilon_{i,j}\overset{iid}{\thicksim}\mathrm{SEVI}$,
achieves the highest likelihood among the three models and for all three
subsamples. We note that
\citet{Fowlie2010}
made this SEVI assumption implicitly, and in doing so, happened to use the
model out of these three with the highest likelihood.

The parameter estimates and the clustered-robust standard errors are reported
in Table \ref{Table: empirics2}. For each of the three models, the
coefficients vary significantly across different regulatory environments. In
particular, the coefficient on capital cost (\texttt{kcost}) based on the
\textquotedblleft deregulated\textquotedblright\ subsample is positive and
statistically significant at the 1\% level, but those based on the
\textquotedblleft public\textquotedblright\ subsample and the
\textquotedblleft regulated\textquotedblright\ subsample are negative, albeit
not statistically significant. This implies that plant managers give more
weight to capital cost if their plants are in a \textquotedblleft
deregulated\textquotedblright\ market rather than a \textquotedblleft
public\textquotedblright\ or \textquotedblleft regularized\textquotedblright%
\ market. Our findings confirm the main result of
\citet{Fowlie2010}
and show that the result is robust to the LEVI specification of the random
cost component.%

\begin{table}[htbp] \centering
{\footnotesize
\begin{tabular}
[c]{lcccccc}\hline\hline
& $\hat{\beta}_{\mathrm{LEVI}}$ & se:$\hat{\beta}_{\mathrm{LEVI}}$ &
$\hat{\beta}_{\mathrm{SEVI}}$ & se:$\hat{\beta}_{\mathrm{SEVI}}$ & $\hat
{\beta}_{\mathrm{NORM}}$ & se:$\hat{\beta}_{\mathrm{NORM}}$\\\hline
& \multicolumn{6}{c}{Deregulated}\\\hline
\texttt{post} & 1.502 & 0.106 & 0.862 & 0.056 & 1.245 & 0.177\\
\texttt{cm} & 1.538 & 0.120 & 0.859 & 0.068 & 1.250 & 0.155\\
\texttt{lnb} & 1.551 & 0.121 & 0.784 & 0.064 & 1.193 & 0.175\\
\texttt{vcost} & 0.188 & 0.045 & 0.112 & 0.030 & 0.145 & 0.041\\
\texttt{kcost} & 0.060 & 0.020 & 0.036 & 0.012 & 0.050 & 0.018\\
\texttt{kage} & 0.037 & 0.008 & 0.028 & 0.004 & 0.037 & 0.010\\
& \multicolumn{6}{c}{Public}\\\hline
\texttt{post} & 5.706 & 1.050 & 3.890 & 0.601 & 4.927 & 0.824\\
\texttt{cm} & 4.433 & 0.442 & 2.685 & 0.193 & 3.589 & 0.384\\
\texttt{lnb} & 3.964 & 0.307 & 2.532 & 0.201 & 3.323 & 0.463\\
\texttt{vcost} & 1.564 & 0.604 & 0.840 & 0.293 & 1.177 & 0.270\\
\texttt{kcost} & -0.039 & 0.255 & -0.100 & 0.118 & -0.086 & 0.088\\
\texttt{kage} & 0.080 & 0.046 & 0.024 & 0.027 & 0.049 & 0.038\\
& \multicolumn{6}{c}{Regulated}\\\hline
\texttt{post} & 2.665 & 0.142 & 1.680 & 0.070 & 2.302 & 0.213\\
\texttt{cm} & 1.911 & 0.103 & 1.250 & 0.058 & 1.686 & 0.151\\
\texttt{lnb} & 2.208 & 0.111 & 1.377 & 0.057 & 1.884 & 0.180\\
\texttt{vcost} & 0.278 & 0.045 & 0.171 & 0.028 & 0.196 & 0.047\\
\texttt{kcost} & -0.008 & 0.023 & -0.005 & 0.013 & -0.023 & 0.022\\
\texttt{kage} & 0.023 & 0.005 & 0.014 & 0.003 & 0.010 & 0.008\\\hline\hline
\end{tabular}
}\caption{Parameter estimates and standard errors: environmental compliance
choices.}\label{Table: empirics2}%
\end{table}%

Note that for the same subsample, the coefficients in the LEVI and SEVI models
may be directly compared. For the \textquotedblleft
deregulated\textquotedblright\ subsample, each element of $\hat{\beta
}_{\mathrm{LEVI}}$ is larger than the corresponding element of $\hat{\beta
}_{\mathrm{SEVI}}$ by more than 50\%. For the \textquotedblleft
public\textquotedblright\ or \textquotedblleft regularized\textquotedblright%
\ subsamples, the coefficients on \texttt{post}$,$ \texttt{cm}$,$%
\texttt{lnb}$,$ \texttt{vcost}$,$ and \texttt{kage} in the LEVI model are
larger than those in the SEVI model by at least 50\%. This indicates the
significance of the model choice on the parameter estimates with consequential
implications on policy recommendations.

\subsection{Further Remarks on the Empirical Applications}

We conclude the empirical application by making some further comments on the
rankings of the three models according to their likelihood values and
computation times.

First, in terms of likelihood values, the SEVI model outperforms the LEVI
model across all considered applications. With the exception of the cracker
brand application, the SEVI model achieves the highest likelihood value among
all three models. Notably, in the application of choosing an environmental
compliance strategy, the SEVI model yields a choice probability in the
conventional logit form, as the focus is on minimizing, rather than
maximizing, an objective function. Overall, regardless of whether an agent
aims to minimize or maximize an objective function, the SEVI distribution
proves to be an important contender for characterizing the random components
in a discrete choice model.

Second, for all applications, the computation time is shortest when the choice
probabilities take the standard LEVI-based logit form. When the number of
alternatives, $J,$ is not larger than 6, the computation time for the model
with the extended choice probability formula, as given in
(\ref{SEVI_prob_formula}), is only two to four times as large as that for the
computationally simplest model and is much lower than that for the multinomial
probit model. When $J=15,$ the computation time for the model with the
extended choice probability formula remains significantly lower than that of
the multinomial probit model, although the competitive advantage is not as
pronounced as when $J$ is smaller.

\section{Conclusions and Further Research Directions\label{Sec: conclusion}}

In a random utility model, individuals make choices among a set of
alternatives based on their utilities. These utilities are typically assumed
to consist of a systematic component, representing aspects of the choice
alternatives that are observable to the analyst, and a random component,
capturing factors that may be known only to the decision makers. Typically,
the random components are assumed to follow the right-skewed \emph{largest}
extreme value Type I (LEVI) distribution, leading to the standard multinomial
and conditional logit models and their generalizations. In this paper, we ask
the question: what if the random component followed, instead, the left-skewed
\emph{smallest} extreme value type I (SEVI) distribution? \emph{A priori},
there is no reason to support choosing LEVI over SEVI as correct beyond one's
prior belief with respect to the skewness of the RUM's random component.

We show that the SEVI-based choice probabilities still have closed-form
expressions, although they are not as simple as those in the standard logit
models. We use the MLE to estimate the RUM when the random component follows
the smallest extreme value type I distribution and show that standard
asymptotic theory, including the asymptotic normality of the MLE, works well
in finite samples. We provide an approach based on Gosper's hack to
efficiently compute the choice probabilities and the MLE for the SEVI model.

If one incorrectly assumes the random component follows LEVI when SEVI is
appropriate and applies the standard logit, the resulting QMLE can have a
large bias arising from misspecification. A distinctive feature of this setup
is that this bias increases with the size of the choice set. The bias also has
important effects on estimating standard economic quantities such as average
partial effects and elasticities, as well as predicting the out-of-sample
choice probabilities. The LEVI and SEVI models exhibit distinctly different
patterns of substitution, with the former conforming to IIA and the latter
exhibiting a particular way of weighting outcomes from all subsets comprised
of two or more alternatives.

In an important way, the distinction between the LEVI and SEVI models is more
striking than the long-standing comparisons of a LEVI model versus a
multinomial probit model. The latter comparisons have focused on the
theoretical convenience of a multinomial probit to allow for correlated error
components. In a SEVI model, the error components exhibit no correlation, yet
the patterns of substitutions are considerably more intricate compared to a
LEVI model. In particular, IIA does not hold in a SEVI model with more than
two alternatives, even though the random utilities are iid and still follow an
extreme value distribution. Furthermore, for the same estimate of an agent's
preference parameters, the predicted probabilities based on the LEVI and
SEVI\ models are typically different in predictable and economically
meaningful ways.

The fundamental question when choosing between the LEVI and SEVI models
revolves around whether the random utility is more likely to be positive, as
is the case with LEVI, or negative, as is the case with SEVI. This is crucial
in a choice context because increasing the magnitude of the random utility is
typically of little or no consequence in a choice model since the
\textquotedblleft scale\textquotedblright\ is not separately identified as
long as all of the random components come from the same iid distribution. This
is not true in general for the skewness of the distribution when there are
more than two alternatives. A key takeaway from this paper is that we should
pay more attention to the skewness of a distribution when it is used to model
the random components in a discrete choice model. From a practical point of
view, we may use past empirical evidence, if any, to decide whether the
stochastic utility is likely to be positively or negatively skewed. If we do
not have such guidance, we may use the model that has a higher likelihood
function in an AIC/BIC sense or employ Vuong's test to decide which model
better fits the data at hand.

There are many possible extensions of our work. Below, we discuss a few
possibilities within the utility maximization framework, keeping in mind that
all discussions apply to discrete decision-making problems aimed at minimizing
an objective function.

First, the LEVI-based conditional and multinomial logit models have been
extended in a number of ways, often in an effort to allow a more flexible
substitution pattern than the IIA. These extensions are generally amenable to
using the SEVI-based conditional logit core in place of the current LEVI one.
Consider, for example, the simplest version of the latent class conditional
logit model, where there are two classes with each characterized by a standard
conditional logit model with potentially different vectors of preference
parameters and a mixing fraction. The standard conditional logit model can be
swapped out here for its SEVI-based counterpart and a variant of the Vuong
test described earlier can be used to help distinguish between them. As
discussed in \ref{Subsection: mixed LEVI-SEVI}, another variant of the latent
class framework would be to allow one of the classes to be based on a
LEVI-conditional logit model and the other to be characterized by a
SEVI-conditional logit model with the same preference parameters across the
two classes. Still, a different version would extend the scale heterogeneity
LEVI model to the SEVI context by allowing for two SEVI conditional logit
models with common preference parameters but different error component
variances, influenced by exogenous factors. These factors might include
sources of the samples, such as whether the data pertains to observed
decisions or stated preference surveys, and whether the data comes from
different regulatory environments, as exemplified in one of our empirical
applications. Obviously, these can all be extended to more than two latent
classes and combined in various ways, such as LEVI and SEVI classes with
different preference parameters. One could envision an EM algorithm that
alternates between fixing preference parameters and estimating the mixing
proportion, and vice versa.

Second, while latent class LEVI variants are quite popular in marketing
because of their \textquotedblleft segmentation\textquotedblright%
\ interpretation and the use of covariates to predict segment memberships,
random parameter variants of conditional logit models have attracted more
attention from economists who find the notion of a continuum of preference
parameters in the population appealing. Nothing in this appeal is tied to the
LEVI versus SEVI distinction in the nature of the random utility component. As
such, there is no conceptual issue involved in swapping the LEVI-core for a
SEVI-core against which the random parameter distributions will be defined.

Third, there are various other flavors of logit models used by economists. The
nested logit explicitly allows a correlation between different subsets of
alternatives that are uncorrelated within a nest. Exploring the properties of
the SEVI analogue would be of interest, given the quite different patterns of
substitution within a nest. The generalized multinomial logit's gamma
parameter determines how the random component of the mixed logit interacts
with the scale parameter (%
\citet{Fiebig2010}%
). It is unclear how this tradeoff is influenced by a SEVI or LEVI core.
Another interesting question is how the cut-off points and parameter estimates
are influenced in an ordered logit model when switching its foundation from
LEVI to SEVI error components. As discussed in Section \ref{Subsec: relation},
the SEVI analogue of the LEVI-type model for completely rank-ordered can be
developed, but the exact details warrant further studies.

Fourth, while the conditional logit and its generalizations dominate empirical
economic work, the statistically equivalent multinomial/polychotomous logit
models are heavily used in biomedical sciences and in other social sciences.
The same issues involving LEVI and SEVI arise in these models. For instance,
in a political election, such as an open party primary for a legislative seat,
beyond the set of attributes available in the official voter's pamphlet, is
the voter likely to learn information that, on average, improves their
perception of the candidate's quality/suitability (LEVI) or decreases (SEVI)?
Predominance of negative advertising suggests SEVI might be more likely in
some contexts.

Fifth, multinomial logistic regression is one of the basic tools taught to
budding data scientists. Its kernel forms the basis of other more complex
procedures, including some flavors of neural networks. Further, multinomial
classification problems are often addressed using the multinomial logit
objective function coupled with variable section procedures like LASSO when
there are a large number of potential predictors. There are SEVI variants of
all of these that could be explored.

Finally, along a very different path, cleverly designed discrete choice
experiments can help shed light on whether all/most agents have random utility
components that are left-skewed, symmetric, or right-skewed. There are a
number of interesting questions along this avenue. For example, does the
skewness of the random component vary with deeper aspects of preferences, such
as risk aversion from economics or openness to new experiences from psychology
(%
\citet{Jagelka2024}
and
\citet{Jiang2024}%
)? Does the skewness vary with observable demographic characteristics like age
and gender, or is it largely a function of educational attainment? Does it
vary by potentially quantifiable differences in how agents obtain information?
Does it also vary with specific aspects of the larger context in which the
choice opportunity is embedded? Does observing repeated choices by the same
agent in similar contexts over time offer more power to identifying an agent's
skewness type, or is it possible that agents shift from one type when the
choice set is novel to another type when the choice opportunity is relatively
routine? There are a number of comparisons of LEVI versus SEVI in different
contexts that should be of broad interest: revealed versus stated preference
data, in-store versus online purchases, private versus public goods, eliciting
the best versus the worst alternative in a set, choices involving different
forms of uncertainty, and choices involving gains versus losses.\ Our work
provides tools to examine these issues.

\bigskip

\bigskip\pagebreak

\section{Appendix of Proofs}

\begin{proof}
[Proof of Theorem \ref{Theorem: choice_prob_under_SEV}]A direct algebraic
proof is given in the supplementary appendix (see Section
\ref{alternative_proof}). Here, we prove the result using the
inclusion-exclusion principle. For each $k\in\mathcal{J}_{-j}:=\left\{
1,\ldots,J\right\}  \backslash\left\{  j\right\}  $, define the event%
\[
\mathcal{E}_{k,j}:v_{ij}+\varepsilon_{ij}\geq v_{ik}+\varepsilon_{ik},
\]
where $\varepsilon_{ij}\thicksim iid$ $\mathrm{SEVI}$ over $j\in\mathcal{J}$.
Equivalently, we define the event $\mathcal{E}_{k,j}$ as
\[
\mathcal{E}_{k,j}:-v_{ij}+\tilde{\varepsilon}_{ij}\leq-v_{ik}+\tilde
{\varepsilon}_{ik},
\]
where $\tilde{\varepsilon}_{ij}\overset{d}{=}-\varepsilon_{ij}$ and
$\tilde{\varepsilon}_{ij}\thicksim iid$ $\mathrm{LEVI}$ over $j\in\mathcal{J}%
$. Let $\mathcal{E}_{k,j}^{c}:-v_{ij}+\tilde{\varepsilon}_{ij}>-v_{ik}%
+\tilde{\varepsilon}_{ik}$ be the complement of $\mathcal{E}_{k,j}$. Then
\begin{align}
P^{\mathrm{SEVI}}\left(  Y_{i}=j|v_{i}\right)   &  =\Pr\left(  \cap
_{k\in\mathcal{J}_{-j}}\mathcal{E}_{k,j}|v_{i}\right)  =1-\Pr\left(
\cup_{_{k\in\mathcal{J}_{-j}}}\mathcal{E}_{k,j}^{c}|v_{i}\right) \nonumber\\
&  =1-\sum_{k_{1}:\left\{  k_{1}\right\}  \subseteq\mathcal{J}_{-j}}\Pr\left(
\mathcal{E}_{k,j}^{c}|v_{i}\right)  +\sum_{k_{1}<k_{2}:\left\{  k_{1}%
,k_{2}\right\}  \subseteq\mathcal{J}_{-j}}\Pr\left(  \mathcal{E}_{k_{1},j}%
^{c}\cap\mathcal{E}_{k_{2},j}^{c}|v_{i}\right) \nonumber\\
&  +\ldots+\left(  -1\right)  ^{J-1}\sum_{k_{1}<\ldots<k_{J-1}:\left\{
k_{1},\ldots,k_{J-1}\right\}  \subseteq\mathcal{J}_{-j}}\Pr\left(
\mathcal{E}_{k_{1,j}}^{c}\cap\mathcal{E}_{k_{2},j}^{c}\cap\ldots
\cap\mathcal{E}_{k_{J-1},j}^{c}|v_{i}\right)  , \label{inclusion_exclusion}%
\end{align}
where we have used the inclusion-exclusion principle (e.g., Proposition 4.4
of
\citet{ross2013first}%
) for computing $\Pr(\cup_{_{k\in\mathcal{J}_{-j}}}\mathcal{E}_{k,j}^{c}%
|v_{i})$, the probability of a union of $\left(  J-1\right)  $ events.

Given $v_{i},$ the event $\mathcal{E}_{k_{1},j}^{c}\cap\mathcal{E}_{k_{2}%
,j}^{c}\cap\ldots\cap\mathcal{E}_{k_{\ell},j}^{c}$ is the same as
\[
-v_{ij}+\tilde{\varepsilon}_{ij}>-v_{ik}+\tilde{\varepsilon}_{ik}\text{ for
all }k\in\left\{  k_{1},\ldots,k_{\ell}\right\}  \subseteq\mathcal{J}_{-j}.
\]
The probability of the above event is the same as the probability of choosing
the $j$-th alternative out of the set of alternatives $\left\{  k_{1}%
,\ldots,k_{\ell},j\right\}  $ in the standard LEVI model with $\left\{
-v_{ik_{1}},\ldots,-v_{ik_{\ell}},-v_{ij}\right\}  $ as the systematic
utilities. Hence, this probability is equal to
\[
\frac{\exp\left(  -v_{ij}\right)  }{\exp\left(  -v_{ij}\right)  +\sum
_{k\in\left\{  k_{1},\ldots,k_{\ell}\right\}  }\exp\left(  -v_{ik}\right)  }.
\]
Plugging the above into (\ref{inclusion_exclusion}) yields the desired formula.
\end{proof}

\bigskip

\begin{proof}
[Proof of Proposition \ref{prop: mirror image}]The first result is the same as
that given in Theorem \ref{Theorem: choice_prob_under_SEV}. To prove the
second result, we note that
\begin{align*}
&  P^{\mathrm{LEVI}}\left(  Y_{i}=j|\mathcal{J},V_{i}=v_{i}\right) \\
&  =\Pr(Y_{i}=j|\mathcal{J},V_{i}=v_{i},\varepsilon_{im}\thicksim iid\text{
}\mathrm{LEVI}\text{ for }m\in\mathcal{J})\\
&  =\Pr\left(  \cap_{k=1}^{J}\left\{  V_{ij}+\varepsilon_{ij}\geq
V_{ik}+\varepsilon_{ik}\right\}  |V_{i}=v_{i},\varepsilon_{im}\thicksim
iid\text{ }\mathrm{LEVI}\text{ for }m\in\mathcal{J}\right) \\
&  =\Pr\left(  \cap_{k=1}^{J}\left\{  -V_{ij}-\varepsilon_{ij}\leq
-V_{ik}-\varepsilon_{ik}\right\}  |V_{i}=v_{i},\varepsilon_{im}\thicksim
iid\text{ }\mathrm{LEVI}\text{ for }m\in\mathcal{J}\right) \\
&  =\Pr\left(  \cap_{k=1}^{J}\left\{  V_{ij}+\varepsilon_{ij}\leq
V_{ik}+\varepsilon_{ik}\right\}  |V_{i}=-v_{i},\varepsilon_{im}\thicksim
iid\text{ }\mathrm{SEVI}\text{ for }m\in\mathcal{J}\right) \\
&  =1-\Pr\left(  \left[  \cap_{k=1}^{J}\left\{  V_{ij}+\varepsilon_{ij}\leq
V_{ik}+\varepsilon_{ik}\right\}  \right]  ^{c}|V_{i}=-v_{i},\varepsilon
_{im}\thicksim iid\text{ }\mathrm{SEVI}\text{ for }m\in\mathcal{J}\right) \\
&  =1-\Pr\left(  \cup_{k=1}^{J}\left[  \left\{  V_{ij}+\varepsilon_{ij}\leq
V_{ik}+\varepsilon_{ik}\right\}  ^{c}\right]  |V_{i}=-v_{i},\varepsilon
_{im}\thicksim iid\text{ }\mathrm{SEVI}\text{ for }m\in\mathcal{J}\right) \\
&  =1-\Pr\left(  \cup_{k=1}^{J}\left[  \left\{  V_{ij}+\varepsilon_{ij}%
>V_{ik}+\varepsilon_{ik}\right\}  \right]  |V_{i}=-v_{i},\varepsilon
_{im}\thicksim iid\text{ }\mathrm{SEVI}\text{ for }m\in\mathcal{J}\right)  .
\end{align*}
Using the inclusion-exclusion principle as in the proof of Theorem
\ref{Theorem: choice_prob_under_SEV}, we have
\begin{align*}
&  P^{\mathrm{LEVI}}\left(  Y_{i}=j|\mathcal{J},V_{i}=v_{i}\right) \\
&  =1+\sum_{\ell=1}^{J-1}\left(  -1\right)  ^{\ell}\sum_{k_{1}<\ldots<k_{\ell
}:\mathcal{S}_{\ell}=\left\{  k_{1},\ldots,k_{\ell}\right\}  \subseteq
\mathcal{J}_{-j}}P^{\mathrm{SEVI}}\left(  Y_{i}=j|\mathcal{S}_{\ell}\cup
j,V_{i}=-v_{i}\right)  .
\end{align*}

\end{proof}

\bigskip

\begin{proof}
[Proof of Proposition \ref{Prop: surplus}]Recall that $G\left(  \cdot\right)
$ and $g\left(  \cdot\right)  $ are the CDF and the pdf of the LEVI
distribution. For $\varepsilon_{j}\thicksim iid$ $\mathrm{SEVI}$ and
$\tilde{\varepsilon}_{j}\thicksim iid$ $\mathrm{LEVI},$ we have%
\begin{align*}
&  E\left[  \max_{j=1,\ldots,J}\left\{  v_{j}+\varepsilon_{j}\right\}  \right]
\\
&  =\sum_{j=1}^{J}E\left[  v_{j}+\varepsilon_{j}|v_{j}+\varepsilon_{j}\geq
v_{k}+\varepsilon_{k},\text{ }k\in\mathcal{J}_{-j}\right]  P^{\mathrm{SEVI}%
}\left(  Y=j|v\right) \\
&  =\sum_{j=1}^{J}E\left[  v_{j}-\tilde{\varepsilon}_{j}|v_{j}-\tilde
{\varepsilon}_{j}\geq v_{k}-\tilde{\varepsilon}_{k},\text{ }k\in
\mathcal{J}_{-j}\right]  P^{\mathrm{SEVI}}\left(  Y=j|v\right) \\
&  =\sum_{j=1}^{J}E\left[  v_{j}-\tilde{\varepsilon}_{j}|-v_{j}+\tilde
{\varepsilon}_{j}\leq-v_{k}+\tilde{\varepsilon}_{k},\text{ }k\in
\mathcal{J}_{-j}\right]  P^{\mathrm{SEVI}}\left(  Y=j|v\right) \\
&  =\sum_{j=1}^{J}\left[  v_{j}-E\left(  \tilde{\varepsilon}_{j}|v_{k}%
-v_{j}+\tilde{\varepsilon}_{j}\leq\tilde{\varepsilon}_{k},\text{ }%
k\in\mathcal{J}_{-j}\right)  \right]  P^{\mathrm{SEVI}}\left(  Y=j|v\right) \\
&  =\sum_{j=1}^{J}v_{j}P^{\mathrm{SEVI}}\left(  Y=j|v\right) \\
&  -\sum_{j=1}^{J}E\left(  \tilde{\varepsilon}_{j}|v_{k}-v_{j}+\tilde
{\varepsilon}_{j}\leq\tilde{\varepsilon}_{k},\text{ }k\in\mathcal{J}%
_{-j}\right)  P^{\mathrm{SEVI}}\left(  Y=j|v\right)  .
\end{align*}
Next,
\begin{align*}
&  E\left(  \tilde{\varepsilon}_{j}|v_{k}-v_{j}+\tilde{\varepsilon}_{j}%
\leq\tilde{\varepsilon}_{k},\text{ }k\in\mathcal{J}_{-j}\right)
P^{\mathrm{SEVI}}\left(  Y=j|v\right) \\
&  =E\left(  \tilde{\varepsilon}_{j}1\left\{  v_{k}-v_{j}+\tilde{\varepsilon
}_{j}\leq\tilde{\varepsilon}_{k},\text{ }k\in\mathcal{J}_{-j}\right\}  \right)
\\
&  =E\left\{  E\left(  \tilde{\varepsilon}_{j}1\left\{  v_{k}-v_{j}%
+\tilde{\varepsilon}_{j}\leq\tilde{\varepsilon}_{k},\text{ }k\in
\mathcal{J}_{-j}\right\}  |\tilde{\varepsilon}_{j}\right)  \right\} \\
&  =\int_{-\infty}^{\infty}z\prod_{k\in\mathcal{J}_{-j}}\left[  1-G\left(
v_{k}-v_{j}+z\right)  \right]  g\left(  z\right)  dz\\
&  =\gamma+\sum_{\ell=1}^{J-1}\left(  -1\right)  ^{\ell}\sum_{k_{1}%
<\ldots<k_{\ell}:\left\{  k_{1},\ldots,k_{\ell}\right\}  \subseteq
\mathcal{J}_{-j}}\int_{-\infty}^{\infty}zG\left(  v_{k_{1}}-v_{j}+z\right)
G\left(  v_{k_{2}}-v_{j}+z\right)  \ldots G\left(  v_{k_{\ell}}-v_{j}%
+z\right)  g\left(  z\right)  dz,
\end{align*}
where $\gamma=\int_{-\infty}^{\infty}zg\left(  z\right)  dz$ is the
Euler--Mascheroni constant.

For the integral in the preceding expression, we have
\begin{align*}
&  \int_{-\infty}^{\infty}zG\left(  v_{k_{1}}-v_{j}+z\right)  G\left(
v_{k_{2}}-v_{j}+z\right)  \ldots G\left(  v_{k_{\ell}}-v_{j}+z\right)
g\left(  z\right)  dz\\
&  =\int_{-\infty}^{\infty}z\exp\left[  -\sum_{k\in\left\{  k_{1}%
,\ldots,k_{\ell}\right\}  }\exp\left(  -v_{k}+v_{j}-z\right)  \right]
\exp\left[  \left(  -z-\exp\left(  -z\right)  \right)  \right]  dz\\
&  =\int_{-\infty}^{\infty}z\exp\left[  -\exp\left(  -z\right)  \sum
_{k\in\left\{  k_{1},\ldots,k_{\ell}\right\}  }\exp\left(  -v_{k}%
+v_{j}\right)  \right]  \exp\left[  \left(  -z-\exp\left(  -z\right)  \right)
\right]  dz\\
&  =\int_{-\infty}^{\infty}z\exp\left[  -\exp\left(  -z\right)  \sum
_{k\in\left\{  k_{1},\ldots,k_{\ell}\right\}  }\exp\left(  -v_{k}%
+v_{j}\right)  \right]  \exp\left[  \left(  -\exp\left(  -z\right)  \right)
\right]  d\exp\left(  -z\right)  .
\end{align*}
By employing a change of variable\ $(s=\exp(-z))$ and letting
\[
b\left(  k_{1},k_{2},\ldots,k_{\ell},j\right)  =1+\sum_{k\in\mathcal{S}_{\ell
}}\exp\left(  -v_{k}+v_{j}\right)
\]
for $\mathcal{S}_{\ell}:=\left\{  k_{1},k_{2},\ldots,k_{\ell}\right\}  ,$ we
then obtain
\begin{align*}
&  \int_{-\infty}^{\infty}zG\left(  v_{k_{1}}-v_{j}+z\right)  G\left(
v_{k_{2}}-v_{j}+z\right)  \ldots G\left(  v_{k_{\ell}}-v_{j}+z\right)
g\left(  z\right)  dz\\
&  =\int_{0}^{\infty}\ln\left(  \frac{1}{s}\right)  \exp\left[  -s\sum
_{k\in\mathcal{S}_{\ell}}\exp\left(  -v_{k}+v_{j}\right)  \right]  \exp\left(
-s\right)  ds\\
&  =\int_{0}^{\infty}\ln\left(  \frac{1}{s}\right)  \exp\left[  -s\left(
\sum_{k\in\mathcal{S}_{\ell}}\exp\left(  -v_{k}+v_{j}\right)  +1\right)
\right]  ds\\
&  =\int_{0}^{\infty}\ln\left(  \frac{1}{s}\right)  \exp\left[  -sb\left(
k_{1},k_{2},\ldots,k_{\ell},j\right)  \right]  ds\\
&  =\frac{\gamma+\ln\left(  b\left(  k_{1},k_{2},\ldots,k_{\ell},j\right)
\right)  }{b\left(  k_{1},k_{2},\ldots,k_{\ell},j\right)  }\\
&  =\left[  \gamma+v_{j}+\ln\left(  \sum_{k\in\mathcal{S}_{\ell}\cup\left\{
j\right\}  }\exp\left(  -v_{k})\right)  \right)  \right]  \frac{\exp\left(
-v_{j}\right)  }{\sum_{k\in\mathcal{S}_{\ell}\cup\left\{  j\right\}  }%
\exp\left(  -v_{k})\right)  }.
\end{align*}
where the second-to-last line follows because for any $\xi>0,$%
\begin{align*}
\int_{0}^{\infty}\ln(\frac{1}{s})\exp\left(  -s\xi\right)  ds  &  =\frac
{1}{\xi}\int_{0}^{\infty}\left[  \ln(\frac{1}{s\xi})+\ln(\xi)\right]
\exp\left(  -s\xi\right)  d\left(  s\xi\right) \\
&  =\frac{\int_{0}^{\infty}\ln(\frac{1}{s})\exp\left(  -s\right)
ds+\ln\left(  \xi\right)  }{\xi}=\frac{\gamma+\ln\left(  \xi\right)  }{\xi}.
\end{align*}

For notational simplicity, let $\mathcal{S}_{\ell,j}:=\mathcal{S}_{\ell}%
\cup\left\{  j\right\}  .$ Then
\begin{align*}
&  E\left[  \max_{j=1,\ldots,J}\left\{  v_{j}+\varepsilon_{j}\right\}  \right]
\\
&  =\sum_{j=1}^{J}v_{j}P^{\mathrm{SEVI}}\left(  Y=j|v\right) \\
&  -\sum_{j=1}^{J}\left(  \gamma+\sum_{\ell=1}^{J-1}\left(  -1\right)  ^{\ell
}\sum_{k_{1}<\ldots<k_{\ell}:\mathcal{S}_{\ell}=\left\{  k_{1},\ldots,k_{\ell
}\right\}  \subseteq\mathcal{J}_{-j}}\left[  \gamma+v_{j}+\ln\left(
\sum_{k\in\mathcal{S}_{\ell,j}}\exp\left(  -v_{k})\right)  \right)  \right]
\frac{\exp\left(  -v_{j}\right)  }{\sum_{k\in\mathcal{S}_{\ell,j}}\exp\left(
-v_{k})\right)  }\right) \\
&  =\sum_{j=1}^{J}v_{j}P^{\mathrm{SEVI}}\left(  Y=j|v\right)  -\gamma
\sum_{j=1}^{J}\left[  1+\sum_{\ell=1}^{J-1}\left(  -1\right)  ^{\ell}%
\sum_{k_{1}<\ldots<k_{\ell}:\mathcal{S}_{\ell}=\left\{  k_{1},\ldots,k_{\ell
}\right\}  \subseteq\mathcal{J}_{-j}}\frac{\exp\left(  -v_{j}\right)  }%
{\sum_{k\in\mathcal{S}_{\ell,j}}\exp\left(  -v_{k})\right)  }\right] \\
&  -\sum_{j=1}^{J}\sum_{\ell=1}^{J-1}\left(  -1\right)  ^{\ell}\sum
_{k_{1}<\ldots<k_{\ell}:\mathcal{S}_{\ell}=\left\{  k_{1},\ldots,k_{\ell
}\right\}  \subseteq\mathcal{J}_{-j}}\left[  v_{j}+\ln\left(  \sum
_{k\in\mathcal{S}_{\ell,j}}\exp\left(  -v_{k})\right)  \right)  \right]
\frac{\exp\left(  -v_{j}\right)  }{\sum_{k\in\mathcal{S}_{\ell,j}}\exp\left(
-v_{k})\right)  }\\
&  =\sum_{j=1}^{J}v_{j}\left[  P^{\mathrm{SEVI}}\left(  Y=j|v\right)
-\sum_{\ell=1}^{J-1}\left(  -1\right)  ^{\ell}\sum_{k_{1}<\ldots<k_{\ell
}:\mathcal{S}_{\ell}=\left\{  k_{1},\ldots,k_{\ell}\right\}  \subseteq
\mathcal{J}_{-j}}\frac{\exp\left(  -v_{j}\right)  }{\sum_{k\in\mathcal{S}%
_{\ell,j}}\exp\left(  -v_{k})\right)  }\right]  -\gamma\\
&  -\sum_{j=1}^{J}\sum_{\ell=1}^{J-1}\left(  -1\right)  ^{\ell}\sum
_{k_{1}<\ldots<k_{\ell}:\mathcal{S}_{\ell}=\left\{  k_{1},\ldots,k_{\ell
}\right\}  \subseteq\mathcal{J}_{-j}}\frac{\exp\left(  -v_{j}\right)  }%
{\sum_{k\in\mathcal{S}_{\ell,j}}\exp\left(  -v_{k})\right)  }\ln\left(
\sum_{k\in\mathcal{S}_{\ell,j}}\exp\left(  -v_{k})\right)  \right) \\
&  =\left(  \sum_{j=1}^{J}v_{j}-\gamma\right)  -\sum_{j=1}^{J}\sum_{\ell
=1}^{J-1}\left(  -1\right)  ^{\ell}\sum_{k_{1}<\ldots<k_{\ell}:\mathcal{S}%
_{\ell}=\left\{  k_{1},\ldots,k_{\ell}\right\}  \subseteq\mathcal{J}_{-j}%
}\frac{\exp\left(  -v_{j}\right)  }{\sum_{k\in\mathcal{S}_{\ell,j}}\exp\left(
-v_{k})\right)  }\ln\left(  \sum_{k\in\mathcal{S}_{\ell,j}}\exp\left(
-v_{k})\right)  \right)  ,
\end{align*}
where the third and last equalities use the result:
\[
1+\sum_{\ell=1}^{J-1}\left(  -1\right)  ^{\ell}\sum_{k_{1}<\ldots<k_{\ell
}:\mathcal{S}_{\ell}=\left\{  k_{1},\ldots,k_{\ell}\right\}  \subseteq
\mathcal{J}_{-j}}\frac{\exp\left(  -v_{j}\right)  }{\sum_{k\in\mathcal{S}%
_{\ell,j}}\exp\left(  -v_{k})\right)  }=P^{\mathrm{SEVI}}\left(  Y=j|v\right)
.
\]
Hence,
\begin{align*}
&  E\left[  \max_{j=1,\ldots,J}\left\{  v_{j}+\varepsilon_{j}\right\}  \right]
\\
&  =-\sum_{j=1}^{J}\sum_{\ell=0}^{J-1}\left(  -1\right)  ^{\ell}\sum
_{k_{1}<\ldots<k_{\ell}:\mathcal{S}_{\ell}=\left\{  k_{1},\ldots,k_{\ell
}\right\}  \subseteq\mathcal{J}_{-j}}\frac{\exp\left(  -v_{j}\right)  }%
{\sum_{k\in\mathcal{S}_{\ell,j}}\exp\left(  -v_{k})\right)  }\ln\left(
\sum_{k\in\mathcal{S}_{\ell,j}}\exp\left(  -v_{k})\right)  \right)  -\gamma\\
&  =\sum_{\ell=1}^{J}\left(  -1\right)  ^{\ell}\sum_{\mathcal{S}_{\ell
}\subseteq\mathcal{J}}\left[  \sum_{j\in\mathcal{S}_{\ell}}\frac{\exp\left(
-v_{j}\right)  }{\sum_{k\in\mathcal{S}_{\ell}}\exp\left(  -v_{k}\right)
}\right]  \ln\left(  \sum_{k\in\mathcal{S}_{\ell}}\exp\left(  -v_{k}\right)
\right)  -\gamma\\
&  =\sum_{\ell=1}^{J}\left(  -1\right)  ^{\ell}\sum_{\mathcal{S}_{\ell
}\subseteq\mathcal{J}}\ln\left(  \sum_{k\in\mathcal{S}_{\ell}}\exp\left(
-v_{k}\right)  \right)  -\gamma\\
&  =\sum_{\ell=1}^{J}\left(  -1\right)  ^{\ell}\sum_{k_{1}<k_{2}%
<\ldots<k_{\ell}:\left\{  k_{1},\ldots,k_{\ell}\right\}  \subseteq\mathcal{J}%
}\ln\left[  \sum_{k\in\left\{  k_{1},\ldots,k_{\ell}\right\}  }\exp\left(
-v_{k}\right)  \right]  -\gamma,
\end{align*}
where $\mathcal{S}_{0}=\varnothing$, the empty set, and the third line follows
from exchanging the orders of summation and organizing the outer sum according
to the value of $\ln\left(  \sum_{k\in\mathcal{S}_{\ell,j}}\exp\left(
-v_{k}\right)  \right)  $ for each $\mathcal{S}_{\ell}\subseteq\mathcal{J}%
_{-j}.$

Now%
\begin{align*}
W^{\mathrm{LEVI}}(v)  &  =\sum_{\ell=1}^{J}\left(  -1\right)  ^{\ell}%
\sum_{k_{1}<k_{2}<\ldots<k_{\ell}:\left\{  k_{1},\ldots,k_{\ell}\right\}
\subseteq\mathcal{J}}\ln\left(  \sum_{k\in\left\{  k_{1},\ldots,k_{\ell
}\right\}  }\exp\left(  -v_{k}\right)  \right)  -\gamma\\
&  -\sum_{\ell=1}^{J}\left(  -1\right)  ^{\ell}\sum_{k_{1}<k_{2}%
<\ldots<k_{\ell}:\left\{  k_{1},\ldots,k_{\ell}\right\}  \subseteq\mathcal{J}%
}\ln\left(  \sum_{k\in\left\{  k_{1},\ldots,k_{\ell}\right\}  }\exp\left(
-0\right)  \right)  +\gamma\\
&  =\sum_{\ell=1}^{J}\left(  -1\right)  ^{\ell}\sum_{k_{1}<k_{2}%
<\ldots<k_{\ell}:\left\{  k_{1},\ldots,k_{\ell}\right\}  \subseteq\mathcal{J}%
}\ln\left[  \frac{1}{\ell}\sum_{k\in\left\{  k_{1},\ldots,k_{\ell}\right\}
}\exp\left(  -v_{k}\right)  \right]  ,
\end{align*}
as desired.

The strict convexity of $W^{\mathrm{SEVI}}(v)$ follows from Theorem 6 of
\citet{SORENSEN2022}%
, as it is easy to check that Conditions (C) and (S) in that paper hold.
\end{proof}

\bigskip

\begin{proof}
[Proof of Proposition \ref{Prop: proportional}]Let $\mathbf{1}_{J}$ be the
$J\times1$ vector of ones. Since the choice probabilities do not change when a
constant $\mu\in\mathbb{R}$ is added to each $V_{ij}\left(  \beta\right)  $,
$E\sum_{k=1}^{J}Y_{ik}\log P^{\mathrm{Q}}(Y_{i}=k|\mathcal{J},V_{i}\left(
\beta\right)  +\mu\mathbf{1}_{J})$ is a constant function of $\mu.$ This
implies that, as a function of $\mu,$ the derivative of $E\sum_{k=1}^{J}%
Y_{ik}\log P^{\mathrm{Q}}(Y_{i}=k|\mathcal{J},V_{i}\left(  \beta\right)
+\mu\mathbf{1}_{J})$ with respect to $\mu$ is a zero function for any
$V_{i}\left(  \beta\right)  .$ But this derivative at $\mu=0$ is
\begin{align*}
&  E\sum_{k=1}^{J}Y_{ik}\frac{\partial\log P^{\mathrm{Q}}(Y_{i}=k|V_{i}\left(
\beta\right)  +\mu\mathbf{1}_{J})}{\partial\mu}|_{\mu=0}\\
&  =E\sum_{k=1}^{J}Y_{ik}\frac{\partial\log P^{\mathrm{Q}}(Y_{i}%
=k|V_{i}\left(  \beta\right)  )}{\partial V_{i}\left(  \beta\right)  ^{\prime
}}\mathbf{1}_{J}\\
&  =E\left[  \mathbf{1}_{J}^{\prime}\sum_{k=1}^{J}Y_{ik}\frac{\partial\log
P^{\mathrm{Q}}(Y_{i}=k|V_{i}\left(  \beta\right)  )}{\partial V_{i}\left(
\beta\right)  }\right]  =E\left[  \mathbf{1}_{J}^{\prime}\Lambda
_{i}^{\mathrm{Q}}\left(  X_{i}\beta\right)  \overrightarrow{Y}_{i}\right]  .
\end{align*}
Hence, $E\left[  \mathbf{1}_{J}^{\prime}\Lambda_{i}^{\mathrm{Q}}\left(
X_{i}\beta\right)  \overrightarrow{Y}_{i}\right]  =0$ for any $\beta\in
int(\mathcal{B)}$. In particular, $E\left[  \mathbf{1}_{J}^{\prime}\Lambda
_{i}^{\mathrm{Q}}\left(  X_{i}\beta_{\mathrm{Q}}^{\ast}\right)
\overrightarrow{Y}_{i}\right]  =0.$ Combining this with Assumption (i) of the
Proposition, we conclude that $\beta=\beta_{\mathrm{Q}}^{\ast}$ is the unique
solution to
\begin{align*}
E\left[  \mathbf{1}_{J}^{\prime}\Lambda_{i}^{\mathrm{Q}}\left(  X_{i}%
\beta\right)  \overrightarrow{Y}_{i}\right]   &  =0,\\
E\left[  X_{i}^{\prime}\Lambda_{i}^{\mathrm{Q}}\left(  X_{i}\beta\right)
\overrightarrow{Y}_{i}\right]   &  =0.
\end{align*}

Since $\beta_{0}\neq0,$ we know that $\beta_{0,\ell}\neq0$ for some
$\ell=1,\ldots,$ or $L.$ Without loss of generality, we assume that
$\beta_{0,1}\neq0.$ We partition $X_{i}$ into $X_{i}=(X_{i}^{(1)},X_{i}%
^{(2)})\in\mathbb{R}^{J\times L}$ with $X_{i}^{(1)}\in\mathbb{R}^{J\times1}$
and $X_{i}^{(2)}\in\mathbb{R}^{J\times\left(  L-1\right)  }.$ For any
$\beta\in\mathbb{R}^{L},$ partition it conformably as $\beta=\left(  \beta
_{1},\beta_{2}^{\prime}\right)  ^{\prime}$ where $\beta_{1}\in\mathbb{R}$ and
$\beta_{2}\in\mathbb{R}^{L-1}.$ In particular, we write $\beta_{0}=\left(
\beta_{0,1}^{\prime},\beta_{0,2}^{\prime}\right)  ^{\prime}$ and
$\beta_{\mathrm{Q}}^{\ast}=\left(  \beta_{\mathrm{Q},1}^{\ast\prime}%
,\beta_{\mathrm{Q},2}^{\ast\prime}\right)  ^{\prime},$ where $\beta_{0,1}$ and
$\beta_{\mathrm{Q},1}^{\ast}$ are the first elements of $\beta_{0}$ and
$\beta_{\mathrm{Q}}^{\ast},$ respectively. Define $\delta_{1}=\beta_{1}%
/\beta_{0,1}$ and $\delta_{2}=\beta_{2}-\delta_{1}\beta_{0,2}.$ Then, we have
the following reparametrization:
\begin{align*}
X_{i}\beta &  =X_{i}\beta_{0}+X_{i}\left(  \beta-\beta_{0}\right) \\
&  =X_{i}\beta_{0}+X_{i}^{\left(  1\right)  }\left(  \beta_{1}-\beta
_{0,1}\right)  +X_{i}^{\left(  2\right)  }\left(  \beta_{2}-\beta_{0,2}\right)
\\
&  =X_{i}\beta_{0}+X_{i}^{\left(  1\right)  }\beta_{0,1}\left(  \delta
_{1}-1\right)  +X_{i}^{\left(  2\right)  }\left(  \delta_{2}+\delta_{1}%
\beta_{0,2}-\beta_{0,2}\right) \\
&  =\left[  X_{i}\beta_{0}-X_{i}^{\left(  1\right)  }\beta_{0,1}%
-X_{i}^{\left(  2\right)  }\beta_{0,2}\right]  +\left[  X_{i}^{\left(
1\right)  }\beta_{0,1}+X_{i}^{\left(  2\right)  }\beta_{0,2}\right]
\delta_{1}+X_{i}^{\left(  2\right)  }\delta_{2}\\
&  =\left(  X_{i}\beta_{0}\right)  \delta_{1}+X_{i}^{\left(  2\right)  }%
\delta_{2}:=\mathcal{X}_{i}\delta,
\end{align*}
where $\mathcal{X}_{i}=(\mathcal{X}_{i}^{\left(  1\right)  },\mathcal{X}%
_{i}^{\left(  2\right)  })$, $\mathcal{X}_{i}^{(1)}=X_{i}\beta_{0}%
\in\mathbb{R}^{J\times1},$ and $\mathcal{X}_{i}^{(2)}=X_{i}^{\left(  2\right)
}\in\mathbb{R}^{J\times\left(  L-1\right)  }.$ Note that%
\[
\mathcal{X}_{i}=X_{i}\left(
\begin{array}
[c]{cc}%
\delta_{1}\beta_{0,1}, & \mathbf{0}_{1\times\left(  L-1\right)  }\\
\delta_{1}\beta_{0,2}, & I_{\left(  L-1\right)  \times\left(  L-1\right)  }%
\end{array}
\right)  _{L\times L}.
\]
So, when $\delta_{1}\neq0,$ there is a one-to-one invertible map between
$\mathcal{X}_{i}$ and $X_{i}.$ As a result,\ when $\delta_{1}\neq0,$ the
system of equations below
\begin{equation}
\left(
\begin{array}
[c]{c}%
E\left[  X_{i}^{(1)}\right]  ^{\prime}\Lambda_{i}^{\mathrm{Q}}\left(
X_{i}\beta\right)  \overrightarrow{Y}_{i}\\
E\left[  X_{i}^{(2)}\right]  ^{\prime}\Lambda_{i}^{\mathrm{Q}}\left(
X_{i}\beta\right)  \overrightarrow{Y}_{i}%
\end{array}
\right)  =0 \label{X system equations}%
\end{equation}
is equivalent to
\begin{equation}
\left(
\begin{array}
[c]{c}%
E\left[  \mathcal{X}_{i}^{(1)}\right]  ^{\prime}\Lambda_{i}^{\mathrm{Q}%
}\left(  \mathcal{X}_{i}\delta\right)  \overrightarrow{Y}_{i}\\
E\left[  \mathcal{X}_{i}^{(2)}\right]  ^{\prime}\Lambda_{i}^{\mathrm{Q}%
}\left(  \mathcal{X}_{i}\delta\right)  \overrightarrow{Y}_{i}%
\end{array}
\right)  =0. \label{Z system equations}%
\end{equation}
It then follows that if $\beta_{\mathrm{Q}}^{\ast}$ is the unique solution to
(\ref{X system equations}), then $\delta^{\ast}=\left(  \delta_{1}^{\ast
},\delta_{2}^{\ast\prime}\right)  ^{\prime}$ for
\[
\delta_{1}^{\ast}=\frac{\beta_{\mathrm{Q},1}^{\ast}}{\beta_{0,1}},\delta
_{2}^{\ast}=\beta_{\mathrm{Q},2}^{\ast}-\frac{\beta_{\mathrm{Q},1}^{\ast}%
}{\beta_{0,1}}\beta_{0,2}%
\]
is the unique solution to (\ref{Z system equations}).

We now show that (\ref{Z system equations}) has a solution of the form
$\delta_{1}=\delta_{1}^{\circ}$ and $\delta_{2}=0$ for some scaler $\delta
_{1}^{\circ}.$ We choose $\delta_{1}=\delta_{1}^{\circ}\neq0$ such that the
first equation in (\ref{Z system equations}) holds when $\delta_{2}$ is set
equal to zero. That is, a nonzero $\delta_{1}^{\circ}$ is chosen to satisfy
\[
E\left[  \mathcal{X}_{i}^{(1)}\right]  ^{\prime}\Lambda_{i}^{\mathrm{Q}%
}\left(  \mathcal{X}_{i}^{\left(  1\right)  }\delta_{1}^{\circ}\right)
\overrightarrow{Y}_{i}=0.
\]
This is possible under Assumption (ii) of the Proposition. Note that
\[
E\left[  \overrightarrow{Y}_{i}|\mathcal{X}_{i}^{(1)},\mathcal{X}_{i}^{\left(
2\right)  }\right]  =E\left[  \overrightarrow{Y}_{i}|X_{i}\beta_{0}%
,\mathcal{X}_{i}^{\left(  2\right)  }\right]  =E\left[  \overrightarrow{Y}%
_{i}|X_{i}\beta_{0}\right]  =E\left[  \overrightarrow{Y}_{i}|\mathcal{X}%
_{i}^{(1)}\right]  ,
\]
because $\overrightarrow{Y}_{i}$ is a function of $X_{i}\beta_{0}$ and
$\varepsilon_{i}$ only, and $\varepsilon_{i}$ is independent of $X_{i}.$ Now,
by the law of iterated expectations, we have
\begin{align*}
&  E\left\{  \left[  \mathcal{X}_{i}^{(2)}\right]  ^{\prime}\Lambda
_{i}^{\mathrm{Q}}\left(  \mathcal{X}_{i}^{\left(  1\right)  }\delta
_{1}+\mathcal{X}_{i}^{\left(  2\right)  }\delta_{2}\right)  \overrightarrow{Y}%
_{i}\right\}  _{\delta_{1}=\delta_{1}^{\circ},\text{ }\delta_{2}=\mathbf{0}}\\
&  =E\left\{  \left[  \mathcal{X}_{i}^{(2)}\right]  ^{\prime}\Lambda
_{i}^{\mathrm{Q}}\left(  \mathcal{X}_{i}^{(1)}\delta_{1}^{\circ}\right)
\overrightarrow{Y}_{i}\right\}  =E\left\{  \left[  \mathcal{X}_{i}%
^{(2)}\right]  ^{\prime}\Lambda_{i}^{\mathrm{Q}}\left(  \mathcal{X}_{i}%
^{(1)}\delta_{1}^{\circ}\right)  E\left[  \overrightarrow{Y}_{i}%
|\mathcal{X}_{i}^{(1)},\mathcal{X}_{i}^{\left(  2\right)  }\right]  \right\}
\\
&  =E\left\{  \left[  \mathcal{X}_{i}^{(2)}\right]  ^{\prime}\Lambda
_{i}^{\mathrm{Q}}\left(  \mathcal{X}_{i}^{(1)}\delta_{1}^{\circ}\right)
E\left[  \overrightarrow{Y}_{i}|\mathcal{X}_{i}^{(1)}\right]  \right\}
=E\left\{  E\left[  \mathcal{X}_{i}^{(2)}|\mathcal{X}_{i}^{\left(  1\right)
}\right]  ^{\prime}\Lambda_{i}^{\mathrm{Q}}\left(  \mathcal{X}_{i}^{(1)}%
\delta_{1}^{\circ}\right)  E\left[  \overrightarrow{Y}_{i}|\mathcal{X}%
_{i}^{(1)}\right]  \right\} \\
&  =E\left\{  E\left[  \mathcal{X}_{i}^{(2)}|\mathcal{X}_{i}^{\left(
1\right)  }\right]  ^{\prime}\Lambda_{i}^{\mathrm{Q}}\left(  \mathcal{X}%
_{i}^{(1)}\delta_{1}^{\circ}\right)  \overrightarrow{Y}_{i}\right\}  .
\end{align*}
Under Assumption (iii) of the Proposition, we have
\[
\underset{J\times\left(  L-1\right)  }{\underbrace{E\left[  \mathcal{X}%
_{i}^{(2)}|\mathcal{X}_{i}^{\left(  1\right)  }\right]  }}=\underset{J\times
1}{\underbrace{\mathcal{X}_{i}^{(1)}}}\underset{1\times\left(  L-1\right)
}{\underbrace{\Theta_{1}^{\prime}}}+\underset{J\times
1}{\underbrace{\boldsymbol{1}_{J}}}\underset{1\times\left(  L-1\right)
}{\underbrace{\Theta_{0}^{\prime}}}%
\]
where $\Theta_{0}=\left(  \Theta_{0,2},\ldots,\Theta_{0,L}\right)  ^{\prime}$
and $\Theta_{1}=\left(  \Theta_{1,2},\ldots,\Theta_{1,L}\right)  ^{\prime}.$
Using the above, $E\left\{  \boldsymbol{1}_{J}^{\prime}\Lambda_{i}%
^{\mathrm{Q}}\left(  \mathcal{X}_{i}\delta\right)  \overrightarrow{Y}%
_{i}\right\}  =0$ for any $\delta,$ and $E\left\{  \left[  \mathcal{X}%
_{i}^{(1)}\right]  ^{\prime}\Lambda_{i}^{\mathrm{Q}}\left(  \mathcal{X}%
_{i}^{(1)}\delta_{1}^{\circ}\right)  \overrightarrow{Y}_{i}\right\}  =0$, we
obtain%
\begin{align*}
&  E\left\{  \left[  \mathcal{X}_{i}^{(2)}\right]  ^{\prime}\Lambda
_{i}^{\mathrm{Q}}\left(  \mathcal{X}_{i}^{\left(  1\right)  }\delta
_{1}+\mathcal{X}_{i}^{\left(  2\right)  }\delta_{2}\right)  \overrightarrow{Y}%
_{i}\right\}  _{\delta_{1}=\delta_{1}^{\circ},\text{ }\delta_{2}=0}\\
&  =\Theta_{0}E\left\{  \boldsymbol{1}_{J}^{\prime}\Lambda_{i}^{\mathrm{Q}%
}\left(  \mathcal{X}_{i}^{(1)}\delta_{1}^{\circ}\right)  \overrightarrow{Y}%
_{i}\right\}  +\Theta_{1}E\left\{  \left[  \mathcal{X}_{i}^{(1)}\right]
^{\prime}\Lambda_{i}^{\mathrm{Q}}\left(  \mathcal{X}_{i}^{(1)}\delta
_{1}^{\circ}\right)  \overrightarrow{Y}_{i}\right\}  =0.
\end{align*}
We have therefore shown that $\delta=\left(  \delta_{1}^{\circ},0_{1\times
\left(  L-1\right)  }\right)  ^{\prime}$ is a solution to
(\ref{Z system equations}).

Since (\ref{Z system equations}) has a unique solution, we must have
\[
\delta_{1}^{\ast}=\delta_{1}^{\circ}\text{ and }\delta_{2}^{\ast}%
=\beta_{\mathrm{Q},2}^{\ast}-\frac{\beta_{\mathrm{Q},1}^{\ast}}{\beta_{0,1}%
}\beta_{0,2}=0.
\]
That is, $\beta_{\mathrm{Q},2}^{\ast}=\delta_{1}^{\circ}\beta_{0,2}$. Hence,
$\beta_{\mathrm{Q}}^{\ast}=\delta_{1}^{\circ}\beta_{0}.$
\end{proof}

\bigskip

\bibliographystyle{apalike}
\bibliography{sevi_ref.bib}

\clearpage\pagebreak

\section*{Online Supplementary Appendix}

\renewcommand{\thesubsection}{S.\Alph{subsection}}
\setcounter{equation}{0}\setcounter{figure}{0}\setcounter{subsection}{0} \renewcommand{\theequation}{\textcolor{red}{S.\arabic{equation}}}

\pagenumbering{arabic} \renewcommand{\thepage}{S.\arabic{page}}

\renewcommand{\thefigure}{\textcolor{red}{S.\arabic{figure}}} \renewcommand{\theHfigure}{S.\arabic{figure}}

\pagenumbering{arabic} \renewcommand{\thepage}{S.\arabic{page}}

\subsection{An alternative proof of Theorem
\ref{Theorem: choice_prob_under_SEV}\label{alternative_proof}}

For $\varepsilon_{ij}\thicksim iid$ $\mathrm{SEVI}$ over $j\in\mathcal{J}$, we
have
\[
P^{\mathrm{SEVI}}\left(  Y_{i}=j|v_{i}\right)  =\Pr(v_{ij}+\varepsilon
_{ij}\geq v_{ik}+\varepsilon_{ik}\text{ for all }k\in\mathcal{J}_{-j}|v_{i}).
\]
Using the distributional equivalence $\varepsilon_{ij}\overset{d}{=}%
-\tilde{\varepsilon}_{ij}$ for $\tilde{\varepsilon}_{ij}\thicksim
iid\mathrm{LEVI}$ over $j\in\mathcal{J}$, we obtain%
\begin{align}
P^{\mathrm{SEVI}}\left(  Y_{i}=j|v_{i}\right)   &  =\Pr(v_{ij}-\tilde
{\varepsilon}_{ij}\geq v_{ik}-\tilde{\varepsilon}_{ik}\text{ for all }%
k\in\mathcal{J}_{-j}|v_{i})\nonumber\\
&  =\Pr(-v_{ij}+\tilde{\varepsilon}_{ij}\leq-v_{ik}+\tilde{\varepsilon}%
_{ik}\text{ for all }k\in\mathcal{J}_{-j}|v_{i})\nonumber\\
&  =\Pr(v_{ik}-v_{ij}+\tilde{\varepsilon}_{ij}\leq\tilde{\varepsilon}%
_{ik}\text{ for all }k\in\mathcal{J}_{-j}|v_{i}).\nonumber
\end{align}
Hence,
\begin{align}
&  P^{\mathrm{SEVI}}\left(  Y_{i}=j|v_{i}\right) \nonumber\\
&  =\int_{-\infty}^{\infty}\left\{  \prod_{k\in\mathcal{J}_{-j}}\left[
1-G\left(  v_{ik}-v_{ij}+z\right)  \right]  \right\}  g\left(  z\right)
dz\nonumber\\
&  =1+\sum_{\ell=1}^{J-1}\left(  -1\right)  ^{\ell}\sum_{k_{1}<\ldots<k_{\ell
}:\left\{  k_{1},\ldots,k_{\ell}\right\}  \subseteq\mathcal{J}_{-j}%
}\nonumber\\
&  \int_{-\infty}^{\infty}G\left(  v_{ik_{1}}-v_{ij}+z\right)  G\left(
v_{ik_{2}}-v_{ij}+z\right)  \ldots G\left(  v_{ik_{\ell}}-v_{ij}+z\right)
g\left(  z\right)  dz.
\end{align}
But, the integral in the preceding expression is%
\begin{align*}
&  \int_{-\infty}^{\infty}G\left(  v_{ik_{1}}-v_{ij}+z\right)  G\left(
v_{ik_{2}}-v_{ij}+z\right)  \ldots G\left(  v_{ik_{\ell}}-v_{ij}+z\right)
g\left(  z\right)  dz\\
&  =\int_{-\infty}^{\infty}\exp\left[  -\sum_{k\in\left\{  k_{1}%
,\ldots,k_{\ell}\right\}  }\exp\left(  -v_{ik}+v_{ij}-z\right)  \right]
\exp\left(  \left[  -z-\exp\left(  -z\right)  \right]  \right)  dz\\
&  =\int_{-\infty}^{\infty}\exp\left[  -\exp\left(  -z\right)  \sum
_{k\in\left\{  k_{1},\ldots,k_{\ell}\right\}  }\exp\left(  -v_{ik}%
+v_{ij}\right)  \right]  \exp\left(  \left[  -\exp\left(  -z\right)  \right]
\right)  d\exp\left(  -z\right) \\
&  =\int_{0}^{\infty}\exp\left[  -\lambda\sum_{k\in\left\{  k_{1}%
,\ldots,k_{\ell}\right\}  }\exp\left(  -v_{ik}+v_{ij}\right)  \right]
\exp\left(  -\lambda\right)  d\lambda\\
&  =\int_{0}^{\infty}\exp\left\{  -\lambda\left[  1+\sum_{k\in\left\{
k_{1},\ldots,k_{\ell}\right\}  }\exp\left(  -v_{ik}+v_{ij}\right)  \right]
\right\}  d\lambda\\
&  =\frac{1}{1+\sum_{k\in\left\{  k_{1},\ldots,k_{\ell}\right\}  }\exp\left(
-v_{ik}+v_{ij}\right)  }.
\end{align*}
So,
\begin{align}
P^{\mathrm{SEVI}}\left(  Y_{i}=j|v_{i}\right)   &  =1+\sum_{\ell=1}%
^{J-1}\left(  -1\right)  ^{\ell}\sum_{k_{1}<\ldots<k_{\ell}:\left\{
k_{1},\ldots,k_{\ell}\right\}  \subseteq\mathcal{J}_{-j}}\frac{1}{1+\sum
_{k\in\left\{  k_{1},\ldots,k_{\ell}\right\}  }\exp\left(  -v_{ik}%
+v_{ij}\right)  }\nonumber\\
&  =1+\sum_{\ell=1}^{J-1}\left(  -1\right)  ^{\ell}\sum_{k_{1}<\ldots<k_{\ell
}:\left\{  k_{1},\ldots,k_{\ell}\right\}  \subseteq\mathcal{J}_{-j}}\frac
{\exp\left(  -v_{ij}\right)  }{\exp\left(  -v_{ij}\right)  +\sum_{k\in\left\{
k_{1},\ldots,k_{\ell}\right\}  }\exp\left(  -v_{ik}\right)  }.
\end{align}

\subsection{Comparing SEVI and LEVI Draws}

Let
\[
F^{\mathrm{LEVI}}\left(  x\right)  :=\Pr(\left\vert \varepsilon_{1}%
-\varepsilon_{2}\right\vert <x|\varepsilon_{1},\varepsilon_{2}\thicksim
^{iid}\mathrm{LEVI},\varepsilon_{1},\varepsilon_{2}>-\ln(\ln(2))
\]
be the conditional CDF of $\left\vert \varepsilon_{1}-\varepsilon
_{2}\right\vert $ conditioning on $\varepsilon_{1}>-\ln(\ln(2))$ and
$\varepsilon_{2}>-\ln(\ln(2)),$ when $\varepsilon_{1},\varepsilon_{2}%
\thicksim^{iid}\mathrm{LEVI}$. Similarly, we define%
\begin{align*}
F^{\mathrm{NORM}}\left(  x\right)   &  :=\Pr(\left\vert \varepsilon
_{1}-\varepsilon_{2}\right\vert <x|\varepsilon_{1},\varepsilon_{2}%
\thicksim^{iid}\mathrm{N}(0,\pi^{2}/6),\text{ }\varepsilon_{1},\varepsilon
_{2}>0),\\
F^{\mathrm{SEVI}}\left(  x\right)   &  :=\Pr(\left\vert \varepsilon
_{1}-\varepsilon_{2}\right\vert <x|\varepsilon_{1},\varepsilon_{2}%
\thicksim^{iid}\mathrm{SEVI},\text{ }\varepsilon_{1},\varepsilon_{2}>\ln
(\ln(2))).
\end{align*}
Note that the medians of the $\mathrm{LEVI}$, $\mathrm{N}(0,\pi^{2}/6),$ and
$\mathrm{SEVI}$ distributions are $-\ln(\ln(2))=0.37,0,$ and $\ln
(\ln(2)=-0.37$, respectively. The above functions are the CDFs of the gap of
two above-median draws from the three distributions with identical variance.
Figure \ref{Figure: cdf_diff_two_above_median_draws} plots these three CDF's,
showing that the gap under the LEVI first-order stochastically dominates the
gap under the normal distribution, which in turn first-order stochastically
dominates the gap under the SEVI.%

\begin{figure}[h]%
\centering
\includegraphics[
height=2.8556in,
width=4.0326in
]%
{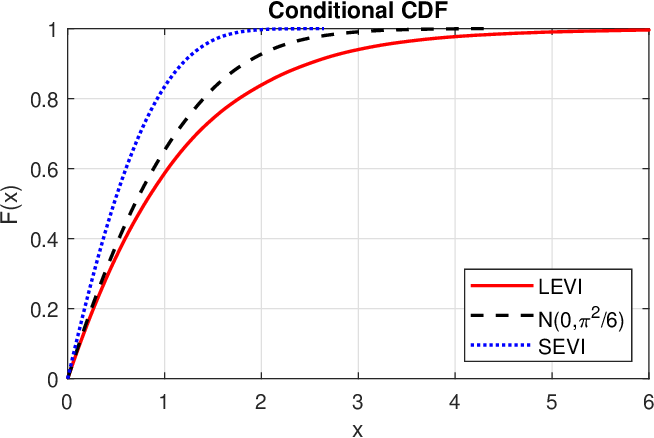}%
\caption{CDF of the gap (absolute difference) between two above-median draws
from LEVI, N$(0,\pi^{2}/6),$ and SEVI. }%
\label{Figure: cdf_diff_two_above_median_draws}%
\end{figure}

As a further illustration of the differences among the three distributions,
Figure \ref{figure; mean_differences_levi_norm_sev} plots the means of
$\left\{  \varepsilon_{\lbrack J]}-\varepsilon_{\lbrack J-1]}\right\}  $ and
$\max\left\{  \varepsilon_{2}-\varepsilon_{1},\ldots,\varepsilon
_{J}-\varepsilon_{1}\right\}  $ against $J$ when $\left\{  \varepsilon
_{j}\right\}  _{j=1}^{J}$ is iid LEVI, $\mathrm{N}(0,\pi^{2}/6),$\ or SEVI.
Here, $\varepsilon_{\lbrack J]}$ and $\varepsilon_{\lbrack J-1]}$ represent
the largest and second-largest values of $\left\{  \varepsilon_{j}\right\}
_{j=1}^{J}$, and so $\varepsilon_{\lbrack J]}-\varepsilon_{\lbrack J-1]}$ is
the gap between the two largest random components. Figure
\ref{figure; mean_differences_levi_norm_sev}(a) shows that, for a given total
number of alternatives, the gap is smallest under the SEVI, largest under the
LEVI, and takes an intermediate value under the normal distribution. Figure
\ref{figure; mean_differences_levi_norm_sev}(b) shows that the expected
maximum difference of the random components has the same ranking across the
three distributions: it is the largest under the LEVI and the smallest under
the SEVI. The figure indicates that, when the systematic utilities are the
same across alternatives, the random components under the LEVI specification
play a larger role in determining the winning choice than under the SEVI specification.%

\begin{figure}[ptb]%
\centering
\includegraphics[
height=2.4656in,
width=5.6377in
]%
{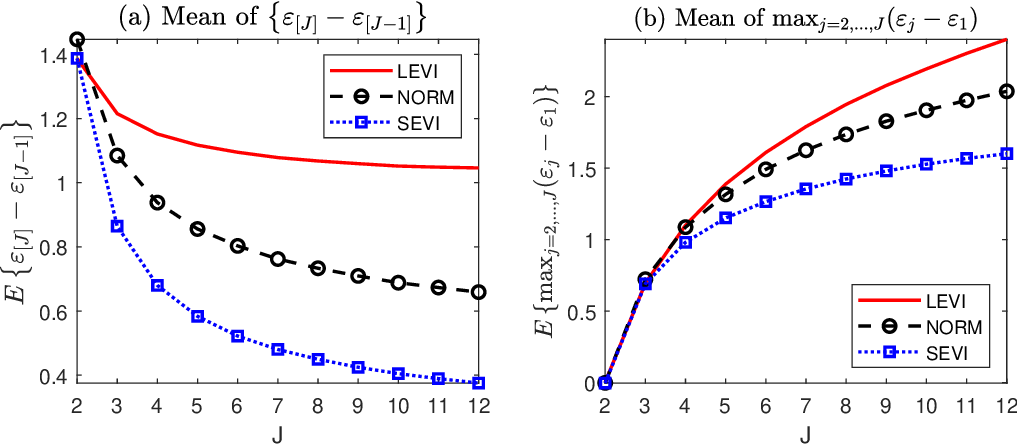}%
\caption{The mean of the gap between the largest two random components and the
mean of the maximum difference between two random components. }%
\label{figure; mean_differences_levi_norm_sev}%
\end{figure}

\subsection{Comparing SEVI and LEVI choice probabilities: further numerical
illustrations\label{Subsec: SEVI_LEVI_diff_supp}}

To further highlight the distinctions between a SEVI model and a LEVI model,
we consider the case with $J=3$ so that $v=\left(  v_{1},v_{2},v_{3}\right)
^{\prime}.$ Define $s=\left(  s_{1},s_{2},s_{3}\right)  ^{\prime}=\left(
\exp\left(  v_{1}\right)  ,\exp\left(  v_{2}\right)  ,\exp\left(
v_{3}\right)  \right)  ^{\prime}.$ Drawing from the terminology of
mathematical psychology (e.g.,
\citet{YELLOTT1977}%
), we refer to $s_{j}$ as the \emph{scale value} of alternative $j.$ We create
a contour plot illustrating the difference in the probabilities of choosing
the first alternative (i.e., $P^{\mathrm{SEVI}}(Y=1|v)-P^{\mathrm{LEVI}%
}(Y=1|v)$) against $\left[  s_{2},s_{3}\right]  ,$ the scale values of the
second and third alternatives for a fixed $s_{1}.$ For this fixed $s_{1},$ we
also make a plot of the probability difference against $s_{2}$ while holding
$s_{2}+s_{3}$ constant at a level $C,$ that is, setting $s_{3}=C-s_{2}.$
Figures \ref{Figure: contourDIFF} and \ref{figure: prob1_diff_against_expv2}
present these two plots when $v_{1}=1.5$ and so $s_{1}=\exp\left(
v_{1}\right)  \approx4.5.$ Both figures demonstrate that the difference is a
highly nonlinear function of $s_{2}$ and $s_{3}.$ For a given value of
$s_{2}+s_{3},$ the probability difference $P^{\mathrm{SEVI}}%
(Y=1|v)-P^{\mathrm{LEVI}}(Y=1|v)$ depends on the relative sizes of $s_{2}$ and
$s_{3}.$ Figure \ref{figure: prob1_diff_against_expv2} shows that when
$s_{2}+s_{3}<\exp\left(  2.193\right)  \approx$\thinspace$9.0,$ so that
alternative 1 has a relatively high scale value and $P^{\mathrm{LEVI}%
}(Y=1|v)>1/3,$ the SEVI probability is greater than the LEVI probability.
Otherwise, if $s_{2}+s\geq\exp\left(  2.193\right)  $, alternative 1 may have
a relatively low scale value, and the SEVI probability of choosing alternative
1 can be less than the LEVI counterpart. When $s_{2}+s_{3}<\exp\left(
1.70\right)  \approx5.47,$ the closer $s_{2}$ and $s_{3}$ are, the stronger
alternative 1 is perceived as a market leader, and the larger the difference
$P^{\mathrm{SEVI}}(Y=1|v)-P^{\mathrm{LEVI}}(Y=1|v)$ becomes.\footnote{When
$s_{1}=4.5$ and $s_{2}=s_{3}=0.5\exp(1.70),$ we have $P^{\mathrm{LEVI}%
}(Y=1|v)\approx0.45.$} In summary, in line with Figures
\ref{Figure: Simu_theoretical_pro}--\ref{Figure: Simu_theoretical_pro_DIFF2},
Figures \ref{Figure: contourDIFF} and \ref{figure: prob1_diff_against_expv2}
show that the SEVI model produces a higher choice probability for the most
favored alternative but a lower choice probability for the least favored
alternative than the LEVI model does.

\begin{figure}[ptb]
\centering
\includegraphics[
height=2.739in,
width=3.694in
]{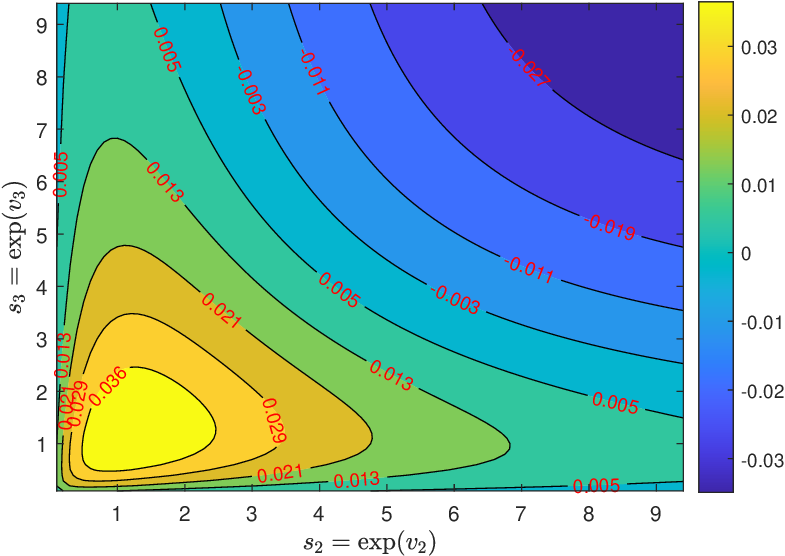}\caption{Contour plot of $P^{\mathrm{SEVI}}%
(Y=1)-P^{\mathrm{LEVI}}(Y=1)$ against $\left[  s_{2},s_{3}\right]
=[\exp(v_{2}),\exp(v_{3})],$ the scale values of alternatives 2 and 3, for a
fixed $s_{1}=\exp\left(  1.5\right)  \approx4.5$}%
\label{Figure: contourDIFF}%
\end{figure}

\begin{figure}[ptb]
\centering
\includegraphics[
height=2.739in,
width=3.694in
]{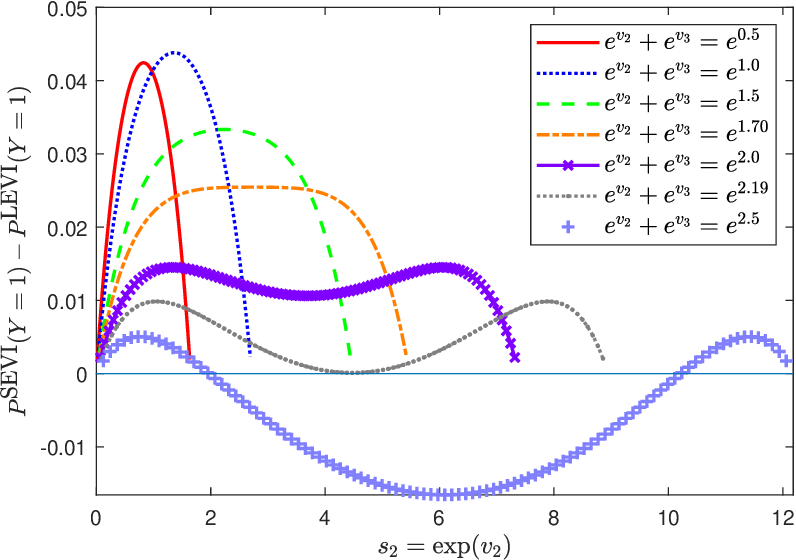}\caption{Plot of $P^{\mathrm{SEVI}%
}(Y=1)-P^{\mathrm{LEVI}}(Y=1)$ against $s_{2}:=\exp(v_{2})$ while holding
$s_{2}+s_{3}:=\exp(v_{2})+\exp(v_{3})$ constant, for a fixed $s_{1}%
=\exp\left(  1.5\right)  \approx4.5$}%
\label{figure: prob1_diff_against_expv2}%
\end{figure}

Figures \ref{Figure: contourDIFF} and \ref{figure: prob1_diff_against_expv2}
are motivated by the following empirical inquiry: Imagine a market with three
firms, where firm 1 is the market leader with a scale value $s_{1}=4.5$ and
firms 2 and 3 have a combined scale value $s_{2}+s_{3}$ of 2.25. Under the
LEVI assumption, firm 1 dominates the market with a market share of 2/3, while
firms 2 and 3 together hold only 1/3. The question is, are consumers more
inclined to purchase from the market leader when the other two firms have
identical market power (i.e., $s_{2}=s_{3}=1.125)$, or when one of them
dominates the other (e.g., $s_{2}=2$ and $s_{3}=0.25)$? The answer depends on
the model used. The standard LEVI logit provides the same prediction for the
dominant firm in both scenarios of competitor configurations, as in the logit
case, only the sum of the competitors' scale values matters. However, the SEVI
model suggests that consumers are 4.4 percentage points more likely to choose
the market leader when the remaining two firms have equal market powers; see
Figure \ref{figure: prob1_diff_against_expv2}. This gap between the LEVI and
SEVI choice probabilities for the dominant firm narrows as the scale values of
the other two firms diverge.

\subsection{Violation of IIA in a standard MNP\label{Subsec: MNP_IIA}}

We use the same setup as discussed at the end of Section \ref{Subsec: IIA} to
illustrate the violation of IIA in a standard MNP, i.e., an MNP with iid
normal random utilities.

Figure\ \ref{Figure: norm_contour_plot_IIA} is similar to Figure
\ref{Figure: contour_plot_IIA}, but we plot the contour of the function
below:
\[
D_{\mathrm{PR}}^{\mathrm{NORM}}(s_{2},s_{3})=\frac{P^{\mathrm{NORM}}%
(Y=2|v_{1},v_{2},v_{3})}{P^{\mathrm{NORM}}(Y=1|v_{1},v_{2},v_{3})}%
-\frac{P^{\mathrm{NORM}}(Y=2|v_{1},v_{2})}{P^{\mathrm{NORM}}(Y=1|v_{1},v_{2}%
)},
\]
where $P^{\mathrm{NORM}}\left(  \cdot|\cdot\right)  $ denotes the choice
probability under the multinomial probit with iid $N(0,\pi^{2}/6)$ errors. The
figure clearly shows that $D_{\mathrm{PR}}^{\mathrm{NORM}}(s_{2},s_{3})$ is
not identically zero, indicating a violation of the IIA assumption.
Additionally, the figure highlights the difference between the multinomial
probit model with iid $N(0,\pi^{2}/6)$ errors and the LEVI model when there
are two and three alternatives. In the LEVI case, the difference in
probability ratios is identically zero, providing a visual contrast with the
non-zero values observed in the multinomial probit case. Therefore, caution is
advised when asserting that the LEVI and MNP models yield comparable results.

\begin{figure}[h]
\centering
\includegraphics[
height=2.739in,
width=3.694in
]{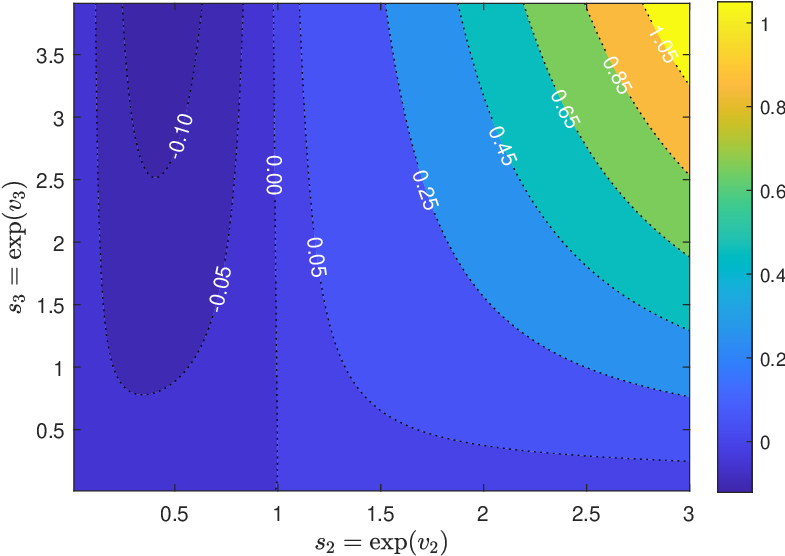}\caption{Contour plot of the difference of the
probability ratios $D_{\mathrm{PR}}^{\mathrm{NORM}}(s_{2},s_{3})$ against
$\left[  s_{2},s_{3}\right]  =\left[  \exp\left(  v_{2}\right)  ,\exp
(v_{3})\right]  $ when $s_{1}=\exp(v_{1})=1$}%
\label{Figure: norm_contour_plot_IIA}%
\end{figure}

\subsection{Supplementary Figures}

\subsubsection{The difference between SEVI and LEVI: compensating variation
and prediction probabilities}

This subsection contains two figures. Figure
\ref{Figure: cv_notequalpro_norm_ratio_vs_ki} plots the ratio of the expected
compensating variation under the SEVI to that under LEVI against the
alternative to be eliminated. We consider eliminating one and only one
alternative at a time, and the label on the horizontal axis indicates the
alternative to be removed. Figure \ref{Figure: difference in choice prob}
plots the differences in the out-of-sample choice probabilities based on
fitting SEVI and LEVI models when $J=4$ against the true SEVI\ choice
probability.
\begin{figure}[h]%
\centering
\includegraphics[
height=3.966in,
width=4.8871in
]%
{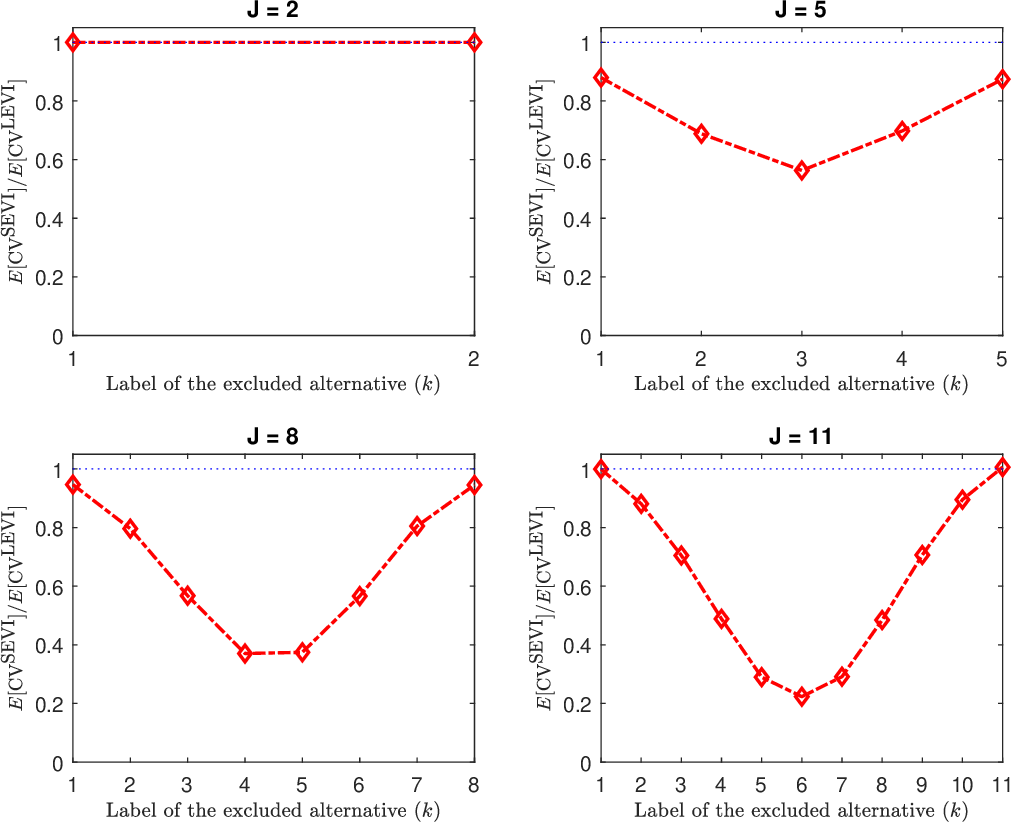}%
\caption{Plot of the ratio of the expected compensating variations across the
SEVI and LEVI models against the label of the excluded alternative with the
same DGPs as those in Figure \ref{Figure: Simu_theoretical_pro} ($X_{ij,\ell}$
is iid $N(0,\pi^{2}\omega_{j}^{2}/36))$}%
\label{Figure: cv_notequalpro_norm_ratio_vs_ki}%
\end{figure}
%

\begin{figure}[h]%
\centering
\includegraphics[
height=3.8086in,
width=4.7072in
]%
{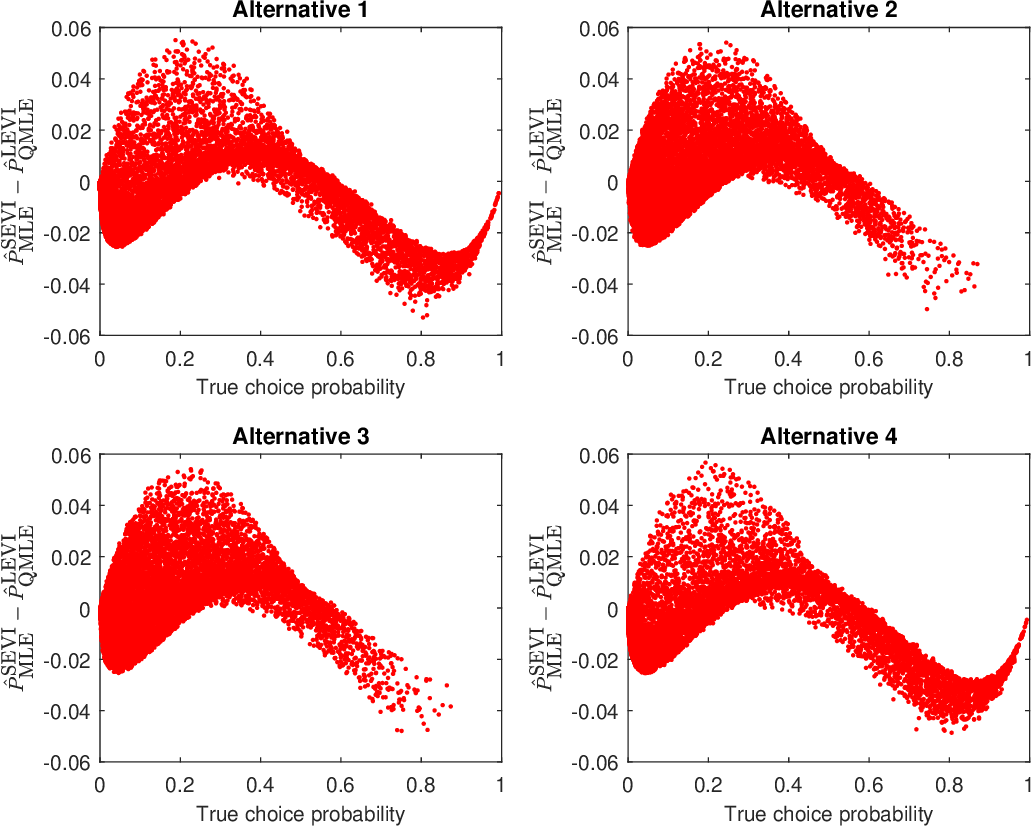}%
\caption{Scatter plot of $\hat{P}_{i_{o},\mathrm{MLE}}^{\mathrm{SEVI}}%
(j)-\hat{P}_{i_{o},\mathrm{QMLE}}^{\mathrm{LEVI}}(j)$ against $P_{i_{o}%
,0}^{\mathrm{SEVI}}\left(  j\right)  $ when $J=4$ and $X_{ij,\ell}\thicksim
iidN(0,\pi^{2}\omega_{j}^{2}/36)$}%
\label{Figure: difference in choice prob}%
\end{figure}

\subsubsection{Boxplot of the MLEs in a mixed LEVI-SEVI Model}

\begin{figure}[h]
\centering
\includegraphics[
trim=0.461773in 0.286556in 0.697062in 0.383074in,
height=3.7051in,
width=4.701in
]{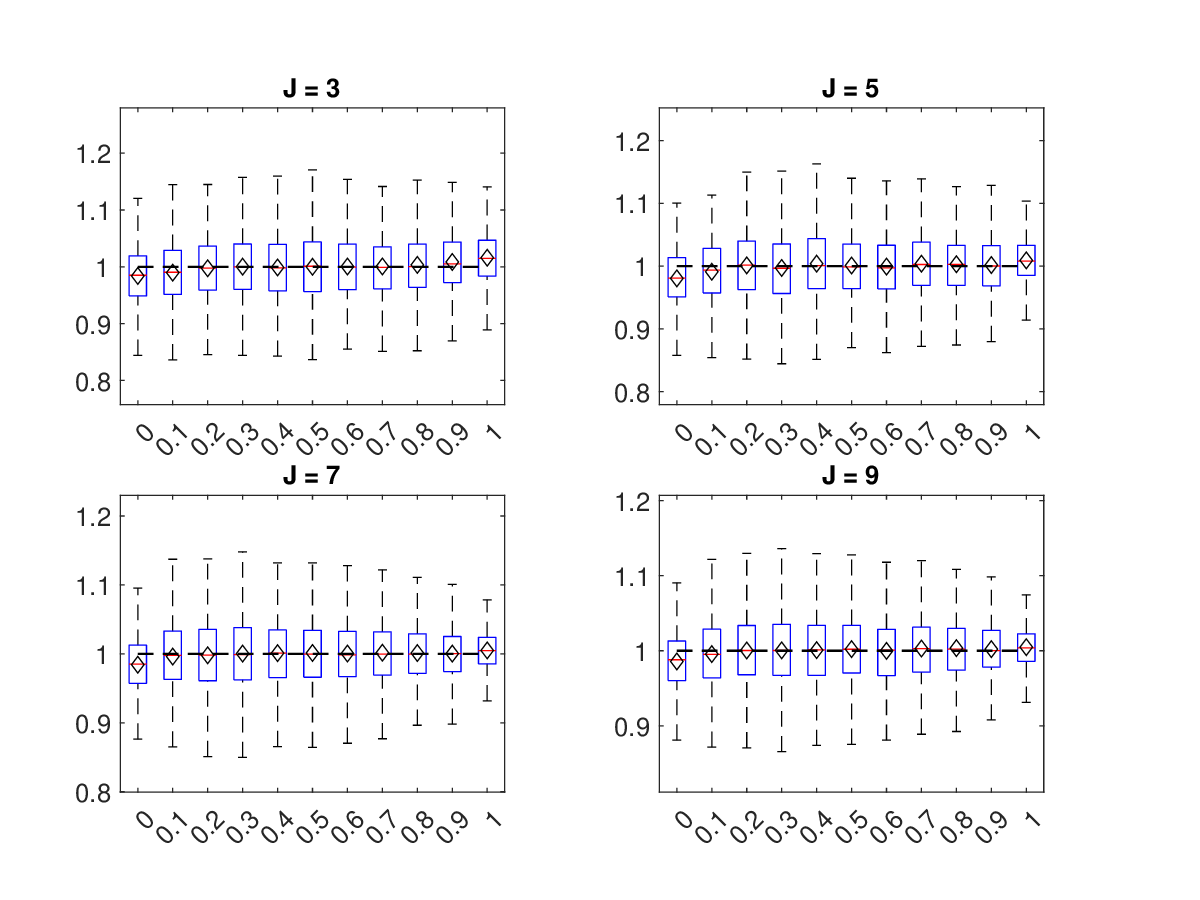}\caption{Boxplot of $\hat{\beta
}_{1,0}$ when the true value $\beta_{1,0}$ is $1$ against $\rho_{0}%
=0,0.1,...,1$, with sample size $5000$ and $J=3,5,7,9$}%
\label{Fig: boxplot_bhat1_vs_rho0_no_outlier_n5000}%
\end{figure}

\subsubsection{Supplementary figures: empirical applications}

\begin{figure}[h]
\centering
\includegraphics[
height=3.7026in,
width=4.5172in
]{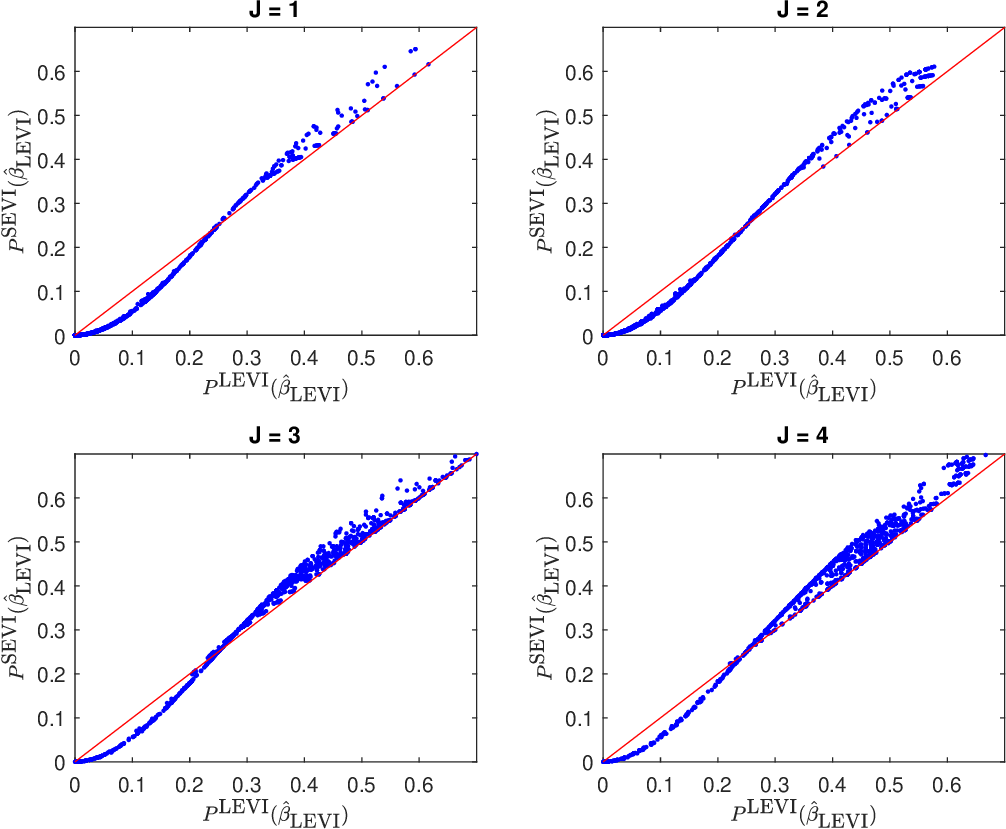}\caption{Scatter plot of $P^{\mathrm{SEVI}%
}(\hat{\beta}_{\mathrm{LEVI}})$ against $P^{\mathrm{LEVI}}(\hat{\beta
}_{\mathrm{LEVI}})$ based on the fishing mode choice data}%
\label{Figure: Fishing_levi_sevi_w_levi_para}%
\end{figure}

\begin{figure}[h]
\centering
\includegraphics[
height=3.1376in,
width=4.1742in
]{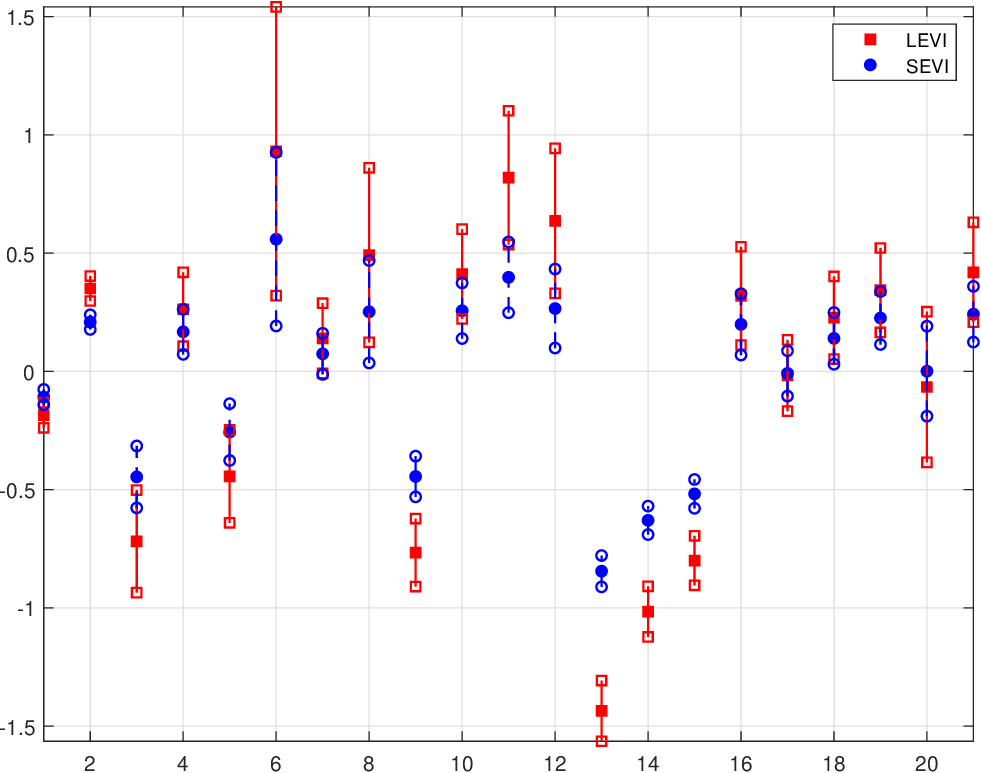}\caption{Point estimates and 95\% confidence
intervals for the 21 parameters in the LEVI and SEVI models for the vehicle
choice application}%
\label{figure: car_b_levi_sevi_w_ci}%
\end{figure}

\end{document}